\documentclass[a4paper,onecolumn,11pt,accepted=2021-04-15]{quantumarticle}
\pdfoutput=1

\usepackage{ulem}

\newcommand{\Add}[1]{\textcolor{black}{#1}}	
\newcommand{\ReviseFurther}[1]{\textcolor{black}{#1}}	
\newcommand{\ReviseFurtherM}[1]{\textcolor{black}{#1}}	

\usepackage[utf8,latin1]{inputenc}   	  		 
\usepackage[T1]{fontenc}          	    		 
\usepackage[english]{babel}        			 
\usepackage[dvipsnames,x11names]{xcolor} 		 
\usepackage[babel]{microtype}                            
\expandafter\let\csname equation*\endcsname\relax        
\expandafter\let\csname endequation*\endcsname\relax
\usepackage{amsmath,amssymb,amsthm,bm,amsfonts,bbm}      
\usepackage{mathtools}					 
\usepackage{soul}
\usepackage{hyperref}     
\usepackage[numbers,sort&compress]{natbib} 

\usepackage{float}

\newcommand{\ketbra}[2]{\vert #1 \rangle \langle #2 \vert}

\newcommand{\etal}{\textit{et al}. }
\newcommand{\ie}{\textit{i.e.}}
\newcommand{\eg}{\textit{e.g.}}

\makeatletter
\newcommand{\doublewidetilde}[1]{{%
  \mathpalette\double@widetilde{#1}%
}}
\newcommand{\double@widetilde}[2]{%
  \sbox\z@{$\m@th#1\widetilde{#2}$}%
  \ht\z@=.9\ht\z@
  \widetilde{\box\z@}%
}

\newcommand{\map}[1]{\widetilde{\mathcal{#1}}{}}
\newcommand{\smap}[1]{\doublewidetilde{{\mathcal{#1}}_{}}{}}
\newcommand{\choi}{\mathfrak{C}}

\newcommand{\lnsp}[1]{{\mathcal{#1}}}

\makeatother

\usepackage{braket}
\usepackage{ascmac}
\DeclareMathOperator*{\bigplus}{\scalerel*{+}{\sum}}
\usepackage{scalerel}

\makeatletter
\newtheoremstyle{indented}
  {30pt}
  {30pt}
  {
	\addtolength{\@totalleftmargin}{3.5em}
   \addtolength{\linewidth}{-3.5em}
   \parshape 1 2em \linewidth}
  {-1.5em}
  {\bfseries}
  {.}
  {.5em}
  {}
\makeatother

\theoremstyle{indented}
\newtheorem{theo}{Theorem}
\newtheorem*{main1*}{Main result 1}
\newtheorem*{main2*}{Main result 2}
\newtheorem{defi}{Definition}
\newtheorem{lemm}{Lemma}
\newtheorem{coro}{Corollary}
\newtheorem{prop}{Proposition}

\newcommand{\1}{\mbox{1}\hspace{-0.25em}\mbox{l}}
\newcommand{\relmiddle}[1]{\mathrel{}\middle#1\mathrel{}}

\newcommand{\dket}[1]{\ket{#1} \rangle}
\newcommand{\dbra}[1]{\langle \bra{#1}}
\newcommand{\bradket}[1]{\braket{#1} \rangle}

\newcommand{\tket}[1]{\ket{#1} \rangle \rangle}

\newcommand{\dbratket}[1]{\langle \braket{#1} \rangle \rangle}

\newcommand{\partsp}{reduced subspace}
\newcommand{\Partsp}{Reduced subspace}

\newcommand{\vspan}[1]{\mathop{\mathrm{SPAN}} \left[ #1 \right]}
\DeclareRobustCommand{\redsp}[3]{{{\vphantom{#1}}^{#2}\!\!\left( #1 \right)}}
\DeclareRobustCommand{\redspn}[3]{{{}^{#2}\!{#1}}}

\DeclareMathOperator{\Tr}{Tr}
\DeclareMathOperator{\Ima}{Im}
\DeclareMathOperator{\projop}{Proj}
\DeclareMathOperator{\suppa}{supp}

\makeatletter
\def\@mkboth#1#2{}
\newlength\appendixwidth
\preto\appendix{\addtocontents{toc}{\protect\patchl@section}}
\newcommand{\patchl@section}{%
  \settowidth{\appendixwidth}{\textbf{Appendix }}%
  \addtolength{\appendixwidth}{1.5em}%
  \patchcmd{\l@section}{1.5em}{\appendixwidth}{}{\ddt}%
}
\makeatother

\usepackage[toc,page]{appendix}
\usepackage{etoolbox}
\appto\appendix{\addtocontents{toc}{\protect\setcounter{tocdepth}{1}}}

\appto\listoffigures{\addtocontents{lof}{\protect\setcounter{tocdepth}{1}}}
\appto\listoftables{\addtocontents{lot}{\protect\setcounter{tocdepth}{1}}}

\newcommand{\todai}{Department of Physics, Graduate School of Science, The University of Tokyo, Hongo 7-3-1, Bunkyo-ku, Tokyo 113-0033, Japan}

\begin{document}
\title{Consequences of preserving reversibility in quantum superchannels}

\author{Wataru Yokojima}
\affiliation{\todai}
\author{Marco T\'{u}lio Quintino}
\affiliation{\todai}
\author{Akihito Soeda}
\affiliation{\todai}
\author{Mio Murao}
\affiliation{\todai}
\maketitle

\begin{abstract}
	Similarly to quantum states, quantum operations can also be transformed by means of quantum superchannels, also known as process matrices. Quantum superchannels with multiple slots are deterministic transformations which take \Add{independent} quantum operations as inputs\Add{. While they} are enforced to respect the laws of quantum mechanics\Add{,} the use of input operations may lack a definite causal order\Add{\ReviseFurther{, }and characterizations of general superchannels in terms of quantum objects with a physical implementation have been missing.}
	In this paper\ReviseFurther{, }we provide a mathematical characterization for pure superchannels with two slots (also known as bipartite pure processes), which are superchannels preserving the reversibility of quantum operations. We show that the reversibility preserving condition restricts all pure superchannels with two slots to be either a quantum circuit only consisting of unitary operations or a coherent superposition of two \Add{unitary quantum circuits where the two input operations are differently ordered}. The latter may be seen as a generalization of the {quantum switch}, allowing a physical interpretation for pure two-slot superchannels. An immediate corollary is that purifiable bipartite processes cannot violate device-independent causal inequalities.
\end{abstract}


\section{Introduction} \label{sec:intro}


Understanding physical transformations between quantum systems is one of the central pillars of quantum mechanics. In quantum information theory, transformations between states are viewed as devices which take a quantum state as inputs and outputs. The properties of such devices, mathematically formalized by quantum channels (deterministic transformations) and quantum instruments (probabilistic transformations), are a research field on \ReviseFurther{their} own and play a fundamental role in quantum information processing \cite{holevobook,wildebook,wolflecturenotes}. Similarly to states, quantum operations and the devices implementing the operations may also be subjected to transformations. In these higher-order operations, quantum operations play the role of inputs \cite{chiribella08,zyczkowski08}. Deterministic higher-order operations are named \textit{quantum superchannels} or \textit{process matrices} \cite{oreshkov11,araujo16}. They are the most general deterministic transformations between multiple \Add{independent} operations allowed by quantum mechanics.

The quantum circuit formalism provides a concrete way to transform quantum operations. For instance, one can prepare a quantum circuit \ReviseFurther{with} a few empty slots where operations may be ``plugged'' to construct a new operation as output, \Add{called a quantum comb} \cite{chiribella07}. \Add{Quantum combs are a special class of superchannels.}
Recently, it has been noted that the standard quantum circuit formalism imposes restrictions on how to transform operations. In particular, the standard quantum circuit formalism has a strict notion of causal order between input operations which recent research suggests to be an \textit{ad hoc} hypothesis \cite{chiribella09,oreshkov11}. Namely, quantum theory admits superchannels which make use of input operations in an indefinite causal order. \Add{Mathematical \ReviseFurther{frameworks} for higher-order quantum theory \ReviseFurther{have} been studied \cite{perinotti17,bisio19,kissinger19}.}

A seminal and illustrative example of a superchannel with indefinite causal order is the quantum switch \cite{chiribella09}.
The quantum switch is a superchannel with two slots that maps a pair of unitary operators $U_A$ and $U_B$ into the unitary operator which is a coherent superposition of operator $U_A$ before $U_B$ with the operator $U_B$ before $U_A$.
The set of superchannels with indefinite causal order is in principle compatible with quantum mechanics and is not restricted to elements equivalent to the quantum switch.
This set includes superchannels which allow device-independent indefinite causal order certification \cite{oreshkov11} and cannot be decomposed \ReviseFurther{into a} simple coherent superposition of ordered circuits. Hence, a possible physical interpretation for arbitrary superchannels is even less obvious than the quantum switch.

In addition to the fundamental interest, quantum superchannels with indefinite causal order have \ReviseFurther{been} shown to be \ReviseFurther{valuable resources} for information processing tasks.
The quantum switch has found applications in several tasks such as discriminating channels with memory \cite{chiribella11}, reducing complexity of quantum computing \cite{araujo14}, exponential reduction of certain communication costs \cite{guerin16},  improving classical and quantum communication \cite{ebler18,salek18}, and semi-device-independent certification of indefinite causal order \cite{bavaresco19}. The quantum switch has also found limitations in the tasks of device-independent indefinite causal order certification \cite{oreshkov11,araujo15} and transforming unitary operations \cite{quintino19a,quintino19b} where only non-switch indefinite quantum superchannels display an advantage.

References \cite{procopio15,rubino17,goswami18} have exploited quantum interferometry to obtain experimental implementations of the quantum switch, and although there is an ongoing  debate on what would be a ``fair implementation'' of the quantum switch, the physical implementation/interpretation and the mathematical structure of the quantum switch are simpler than general superchannels with indefinite causal order \cite{oreshkov2018timedelocalized,araujo16}. For instance, there is no simple classification or universal procedure to analyze general superchannels.

Reversibility is a key concept in several physical theories.
Motivated by the role of reversibility in physical theories, Ref.\,\cite{araujo16}
introduced the definition of pure superchannels as superchannels preserving the reversibility of input operations. More formally, a superchannel is said to be pure if it transforms \Add{independent} unitary channels
into a unitary channel.
The authors proposed a purification postulate stating that superchannels which cannot be purified do not have a fair physical implementation. Reference\,\cite{oreshkov2018timedelocalized} shows that pure two-slot superchannels have a realization on suitably defined ``time-delocalized'' subsystems and Ref.\,\cite{guerin18} proves an equivalence between pure processes and multilinear maps that admit a description in terms of consistent causal reference frames.

In this paper we analyze the restrictions imposed by the reversibility preserving property of superchannels and present a simple decomposition for this class of processes. We show that pure quantum superchannels with two slots can be divided \ReviseFurther{into} two cases: \Add{1) quantum combs realizable only with unitary operations (pure combs); 2) coherent superposition of two pure combs of which input operations are differently ordered}, which may be seen as a generalization of the quantum switch\footnote{\ReviseFurther{Related results are recently presented in Ref.\,\cite{lorenz2020causal} and Ref.\,\,\cite{barrett20} by Barrett \etal using different proof techniques. While they employ techniques developed to investigate quantum common causes and quantum causal models \cite{allen2016quantum}, independently developed alternative techniques to specifically analyze quantum superchannels are used in this paper and the proofs are self-contained. See \nameref{sec:NoteAdded} in detail.}}. \ReviseFurther{As case 1) is a special case of case 2) where one of the two pure combs is trivial, we present the result in the form of case 2) in the following sections.} Our characterization provides a potential interpretation for every pure superchannel with two slots.

This paper is structured as follows. In Sec.\,\ref{sec:overview} we present a non-technical overview of key concepts on quantum superchannels and state our main results.
In Sec.\,\ref{sec:preliminary} we review mathematical definitions of quantum superchannels required for the formal statement of our results. In Sec.\,\ref{sec:main} we present the formal statement of our main results and a few direct implications.
{We conclude in Sec.\,\ref{sec:discussions}.}
In Appendix \ref{secApp:notation} we present relevant properties for decomposing linear spaces.
{Appendix \ref{secApp:ProofOfPureCombDecomposition} presents a proof of Thm.\,\ref{thm:PureCombDecomposition}. Appendix \ref{secApp:ProofsOfpreparation} and Appendix \ref{secApp:idea} present a proof of Thm.\,\ref{thm:equivalent}.}


\section{Non-technical overview of concepts and main results} \label{sec:overview}


In this \ReviseFurther{section,} we exploit the quantum circuit model to illustrate fundamental concepts and state our main results. In quantum circuits, wires describe quantum systems and boxes describe \textit{quantum channels} which are general deterministic quantum operations%
\footnote{We stress that, unless explicitly specified, channels represented by white boxes should not be assumed to be unitary. In Figure \ref{fig:stinespring} we present a connection between general quantum channels and unitary channels.}
(Figure \ref{fig:circuitChannel}). Wires take quantum states through boxes from left to right.
\begin{figure}[H]
	\centering \includegraphics[keepaspectratio, scale=0.30]{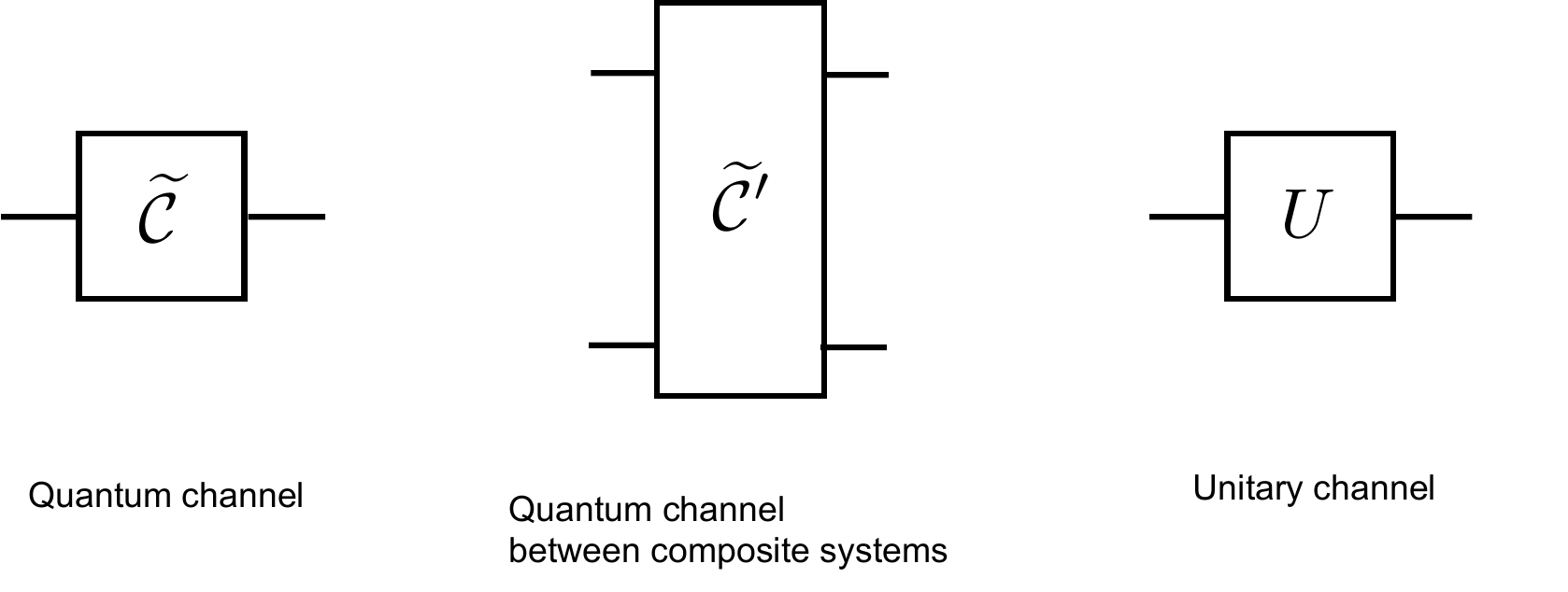}
	\caption{Circuit representation of quantum channels.}
	\label{fig:circuitChannel}
\end{figure}
We represent quantum states by putting the semi-circular objects from the left and discarding systems by putting the inverted semi-circular objects from the right (Figure \ref{fig:circuitState}).
\begin{figure}[H]
	\centering \includegraphics[keepaspectratio, scale=0.30]{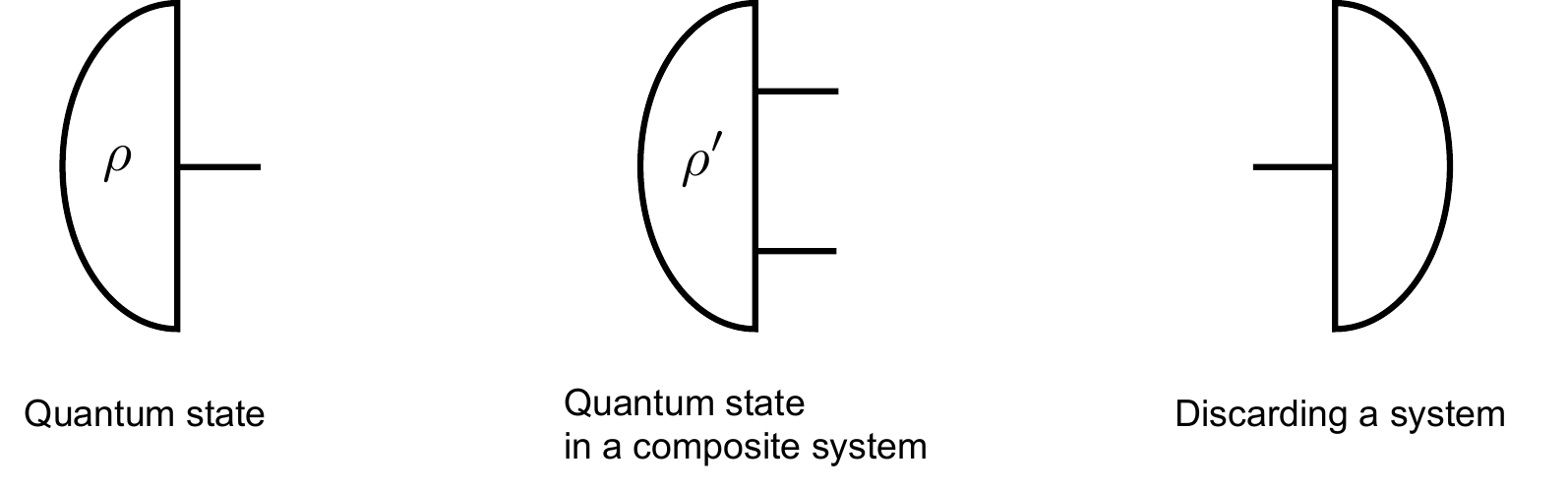}
	\caption{Circuit representation of quantum states and the process of discarding a quantum system. }
	\label{fig:circuitState}
\end{figure}

An important result in characterizing quantum operation is given by the Stinespring dilation \cite{stinespring55}, a theorem which proves that every quantum channel can be realized by a unitary channel with an auxiliary state and possibly discarding a part of the quantum system (see Figure \ref{fig:stinespring}).
\begin{figure}[H]
	\centering \includegraphics[keepaspectratio, scale=0.30]{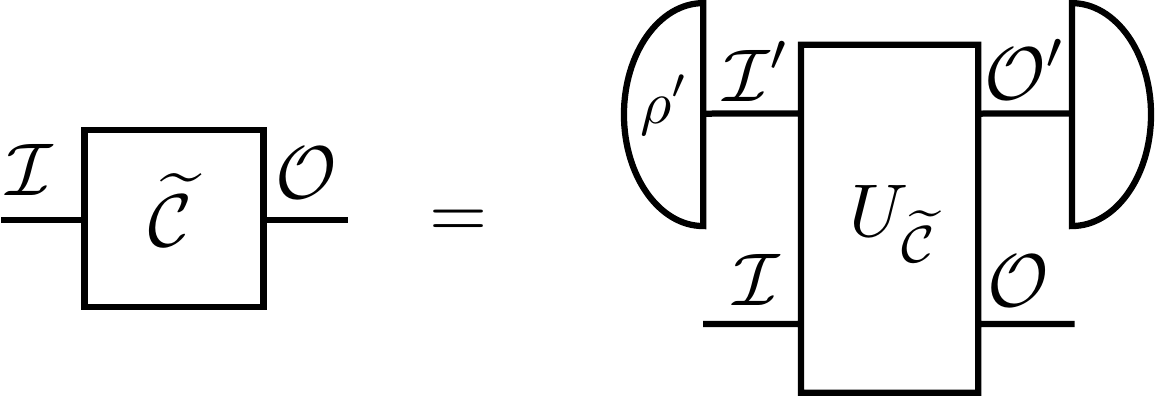}
	\caption{Stinespring dilation: every quantum channel can be realized by a unitary channel combined with an auxiliary state and possibly discarding a part of the quantum system. }
	\label{fig:stinespring}
\end{figure}

In addition to these components, we use colored boxes to represent abstract \textit{superchannels}, which are  deterministic higher-order quantum operations that can transform quantum channels, and boxes with color white in figures represent quantum channels. Strictly speaking, general \ReviseFurther{multiple-slot} superchannels may not be valid components of quantum circuits, but they are depicted as if so \ReviseFurther{for convenience}.


\subsection{Causally ordered superchannels}  \label{sec:aboutCausallyOrdered}


\begin{itemize}
	\item A \textit{quantum superchannel with one slot} is a linear supermap which maps quantum channels into quantum channels (even if the supermap acts on a subsystem of the input quantum channels)  (Figure \ref{fig:superchan}).
	      \begin{figure}[H]
		      \centering \includegraphics[keepaspectratio, scale=0.30]{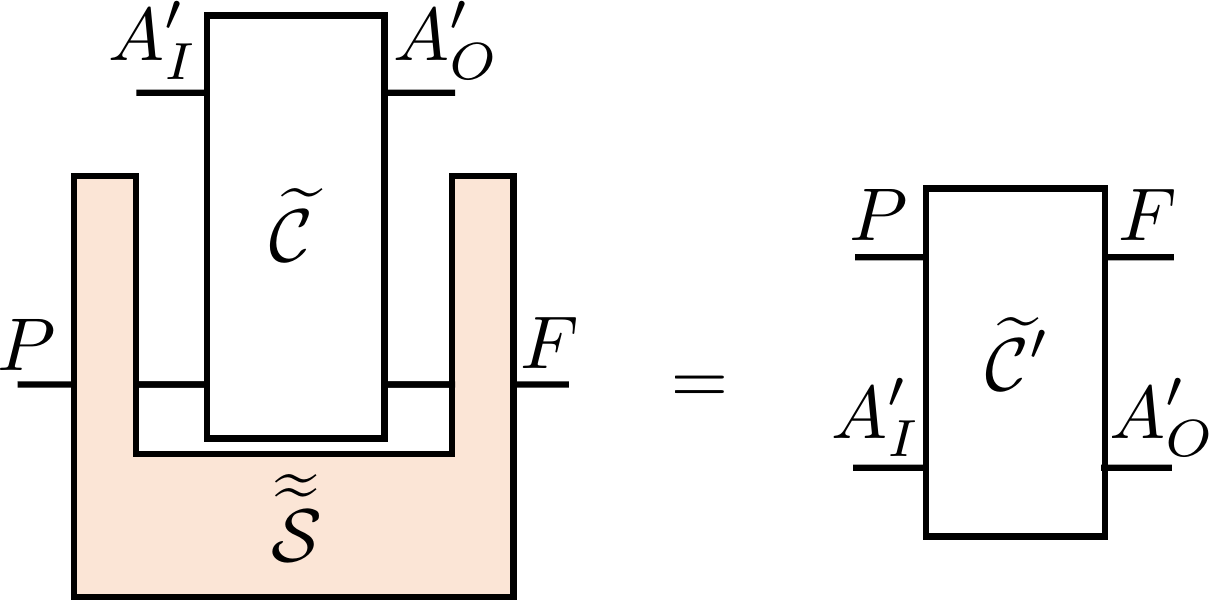}
		      \caption{$\smap{S}$ is a quantum superchannel with one slot, which maps a quantum channel $\map{C}$ into a quantum channel $\map{C}^\prime$.}
		      \label{fig:superchan}
	      \end{figure}
	      It has been shown that any quantum superchannel with one slot can be realized by a quantum circuit (Figure \ref{fig:supermap=2comb}) \cite{chiribella08}.
	      \begin{figure}[H]
		      \centering \includegraphics[keepaspectratio, scale=0.30]{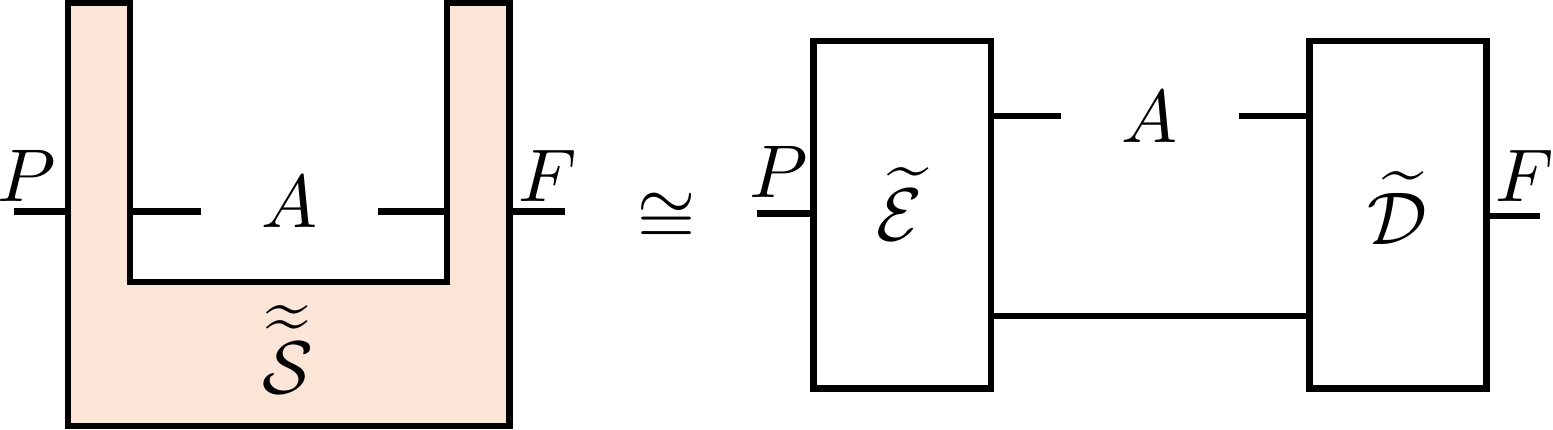}
		      \caption{Every quantum superchannel $\smap{S}$ with one slot is realized with quantum channels $\map{E}, \map{D}$.}
		      \label{fig:supermap=2comb}
	      \end{figure}

	      A quantum superchannel with one slot is \textit{pure} if the superchannel maps unitary channels into unitary channels (even if the superchannel acts on subsystems of the input unitary channels) (Figure \ref{fig:pureSuperchan}).
	      \begin{figure}[H]
		      \centering \includegraphics[keepaspectratio, scale=0.30]{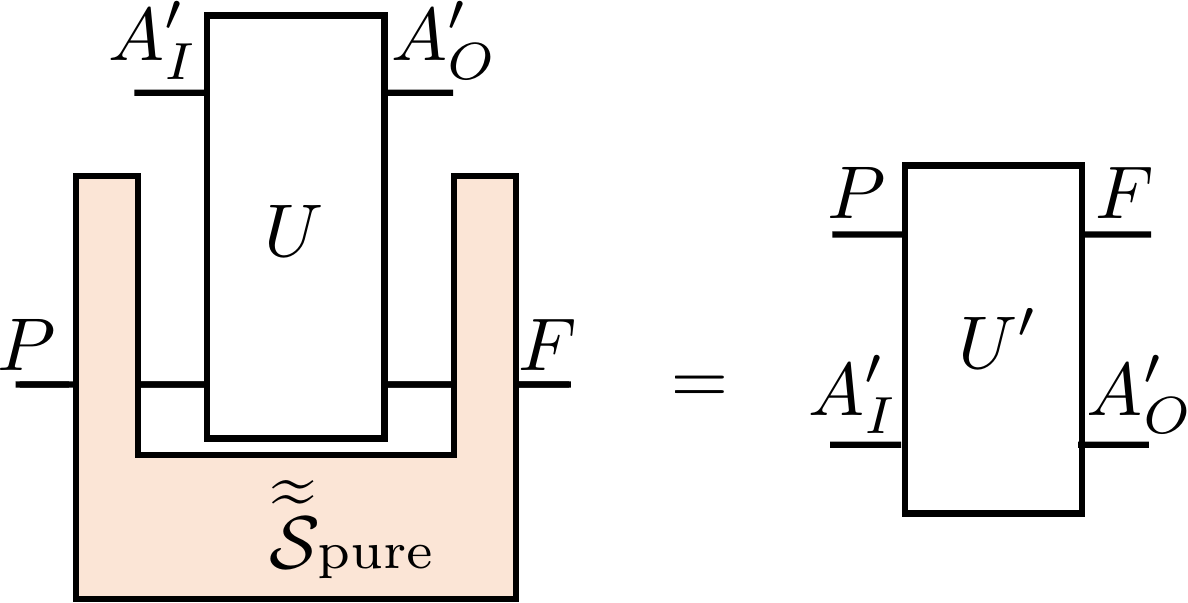}
		      \caption{$\smap{S}_{\mathrm{pure}}$ is a pure superchannel, which maps a unitary channel $U$ into a unitary channel $U^\prime$.}
		      \label{fig:pureSuperchan}
	      \end{figure}

	\item \Add{A \textit{quantum comb} \cite{chiribella07} (see also \cite{kretschmann05,gutoski07}) is a superchannel that can be realized by a quantum circuit (Figure \ref{fig:quantum_comb}).}
	      A quantum comb with no slots is a quantum channel and a quantum comb with $N$ slot(s) is a linear supermap which maps every quantum comb with $N-1$ slot(s) into a quantum channel (Figure \ref{fig:comb_maps_comb_to_comb}).
	      \Add{We call quantum combs \textit{causally ordered} superchannels.}
	      \begin{figure}[H]
		      \centering \includegraphics[keepaspectratio, scale=0.3]{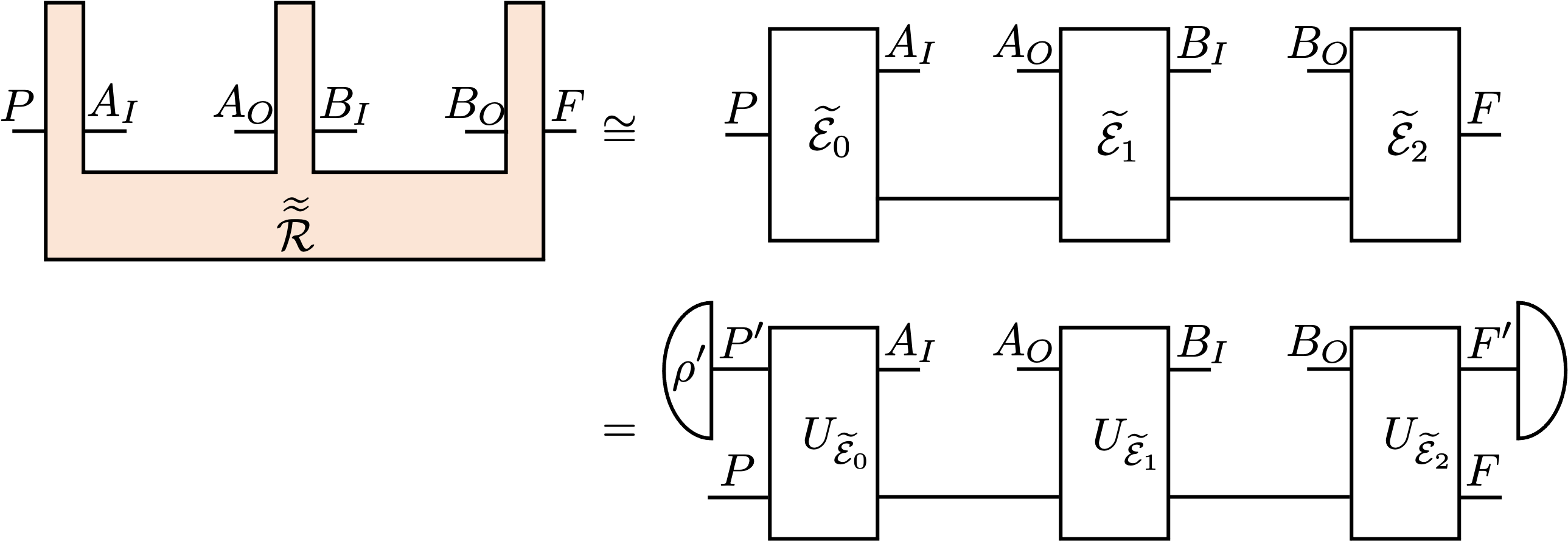}
		      \caption{Every quantum comb can be realized by quantum circuits. This picture shows a two-slot comb $\smap{R}$ which can be realized by quantum channels $\map{E}_0, \map{E}_1,$ and $\map{E}_2$, which are not necessarily unitary. When additional auxiliary states and discarding subsystems are allowed, the Stinespring dilation (see Figure \ref{fig:stinespring}) ensures every quantum circuit can be realized by unitary channels.}
		      \label{fig:quantum_comb}
	      \end{figure}
	      \begin{figure}[H]
		      \centering \includegraphics[keepaspectratio, scale=0.30]{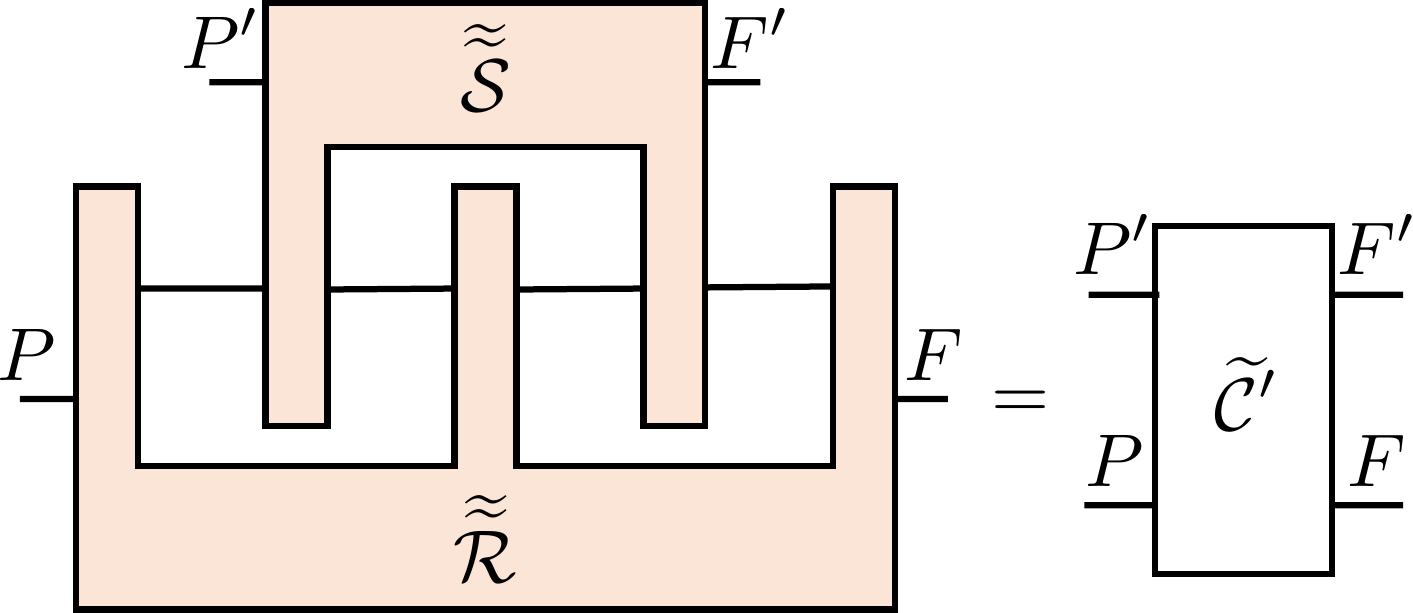}
		      \caption{$\smap{R}$ is a quantum comb with two slots, which maps a quantum comb $\smap{S}$ with a slot into a quantum channel $\map{C}^\prime$.}
		      \label{fig:comb_maps_comb_to_comb}
	      \end{figure}

	      A quantum comb is \textit{pure} if the quantum comb maps a group of \Add{independent} unitary channels into unitary channels (even if the comb acts on subsystems of the input unitary channels) \Add{\cite{araujo16}}.
\end{itemize}

\begin{main1*} (Thm.\,\ref{thm:PureCombDecomposition})
	Every pure quantum comb, independently of the number of slots, can be realized by a quantum circuit containing only unitary channels without any need of auxiliary systems \Add{in global past} or discarding of subsystems \Add{in global future} (Figure \ref{fig:f89ga89hujap}).
\end{main1*}
\begin{figure}[H]
	\centering \includegraphics[keepaspectratio, scale=0.30]{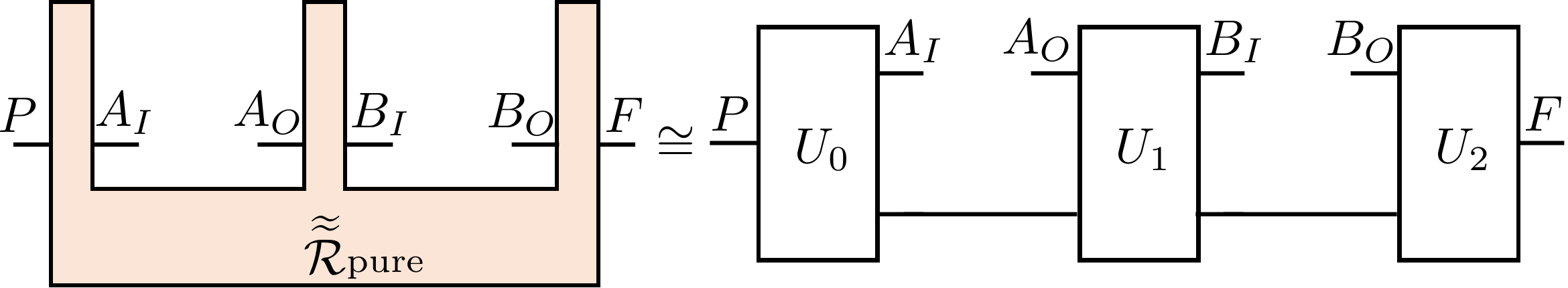}
	\caption{Every pure quantum comb with two slots $\smap{R}_{\mathrm{pure}}$  can be realized by a quantum circuit with some unitary channels $U_0$, $U_1$, and $U_2$ and does not require discarding part of the quantum system. This figure should be contrasted with Figure \ref{fig:quantum_comb}, which presents a circuit construction for general quantum combs and makes use of auxiliary systems $P'$ and $F'$. From this perspective, a quantum comb is pure if and only if it can be realized by quantum circuits with unitary channels  $U_0$, $U_1$, and $U_2$  without making use of the auxiliary systems $P'$ and $F'$. }
	\label{fig:f89ga89hujap}
\end{figure}
We emphasize that the dimensions of systems $P$, $A_I$, $A_O$, $B_I$, $B_O$, and $F$ represented in Figure \ref{fig:f89ga89hujap} are not necessarily the same. Note that, since the comb $\smap{R}_{\mathrm{pure}}$ {is represented by a sequence of unitary operators as} in Figure \ref{fig:f89ga89hujap}, these dimensions should satisfy certain constraints. {For instance, when $\dim A_I = \dim A_O (=: d_A)$ and $\dim B_I = \dim B_O (=: d_B)$,}
we must have $\dim P=\dim F$.
Moreover, $\dim P$ must be a multiple of {$d_A$} and $\dim F$ must be a multiple of {$d_B$}. In particular, when the dimension of systems $P$, $A_I$, $A_O$, $B_I$, $B_O$, and $F$ are all equivalent, pure combs can be represented without any additional wires, see Figure \ref{fig:pureCombWithSameDimension} and Cor.\,\ref{cor:dddanddd2d} (1).

\begin{figure}[H]
	\centering \includegraphics[keepaspectratio, scale=0.30]{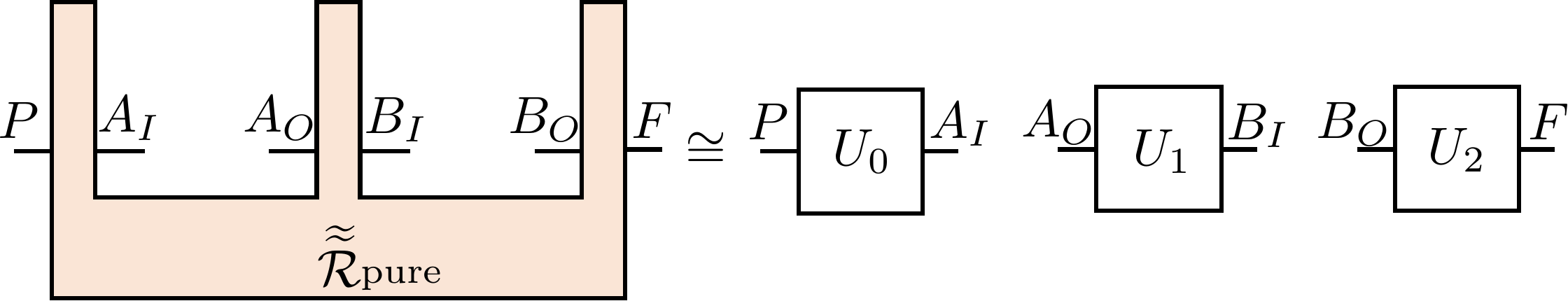}
	\caption{{When the dimensions of involved systems satisfy $\dim P = \dim A_I = \dim A_O = \dim B_I = \dim B_O = \dim F$, the pure comb $\smap{R}_{\mathrm{pure}}$ can be realized by a quantum circuit with some unitary channels $U_0$, $U_1$, and $U_2$ without any additional wires and without discarding any quantum system.}}
	\label{fig:pureCombWithSameDimension}
\end{figure}

\subsection{General superchannels with two slots}
\label{sec:aboutGeneral}


\begin{itemize}
	\item A \textit{quantum superchannel} with $N$ slots ($N \geq 0$)
	      is a linear supermap which maps every group of $N$ \Add{independent} quantum channel(s) into a quantum channel (even if the supermap acts on subsystems of the input channels) (Figure \ref{fig:processt4qg}).
	      \begin{figure}[H]
		      \centering \includegraphics[keepaspectratio, scale=0.30]{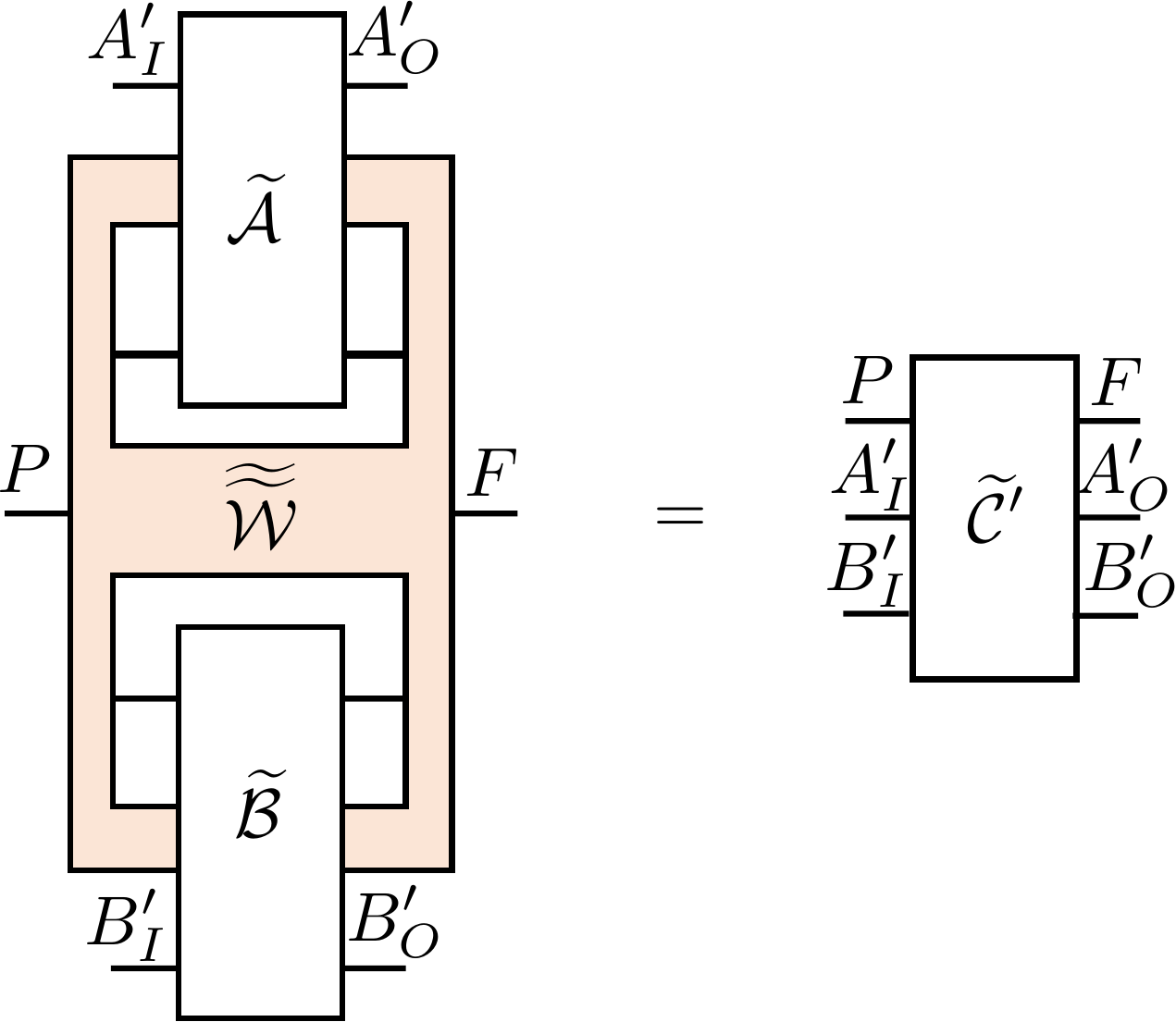}
		      \caption{$\smap{\mathcal{W}}$ represents a quantum superchannel with two slots. A two slot superchannel maps any pair of \Add{independent} input-channels represented by the tensor product $ \map{\mathcal{A}} \otimes \map{\mathcal{B}} $ into an output-channel $\map{\mathcal{C}}$. The shape of $\smap{W}$ in this figure illustrates that $\smap{W}$ may have indefinite causal order.}
		      \label{fig:processt4qg}
	      \end{figure}
	\item A \textit{quantum switch} is a quantum superchannel with two slots whose causal order is coherently controlled by a control qubit (Figure \ref{fig:qswitch_uaub}). It is defined to transform any pair of two \Add{independent} unitary channels represented by unitary operators $U_A$ and $U_B$ into a unitary channel represented by $\ketbra{0}{0} \otimes U_B U_A + \ketbra{1}{1} \otimes U_A U_B \Add{=U_BU_A\oplus U_AU_B}$\footnote{\Add{$\oplus$ denotes the orthogonal direct sum of linear subspaces. See Appendix \ref{secApp:linearAlgebra}.}}.
	      \begin{figure}[H]
		      \centering \includegraphics[keepaspectratio, scale=0.30]{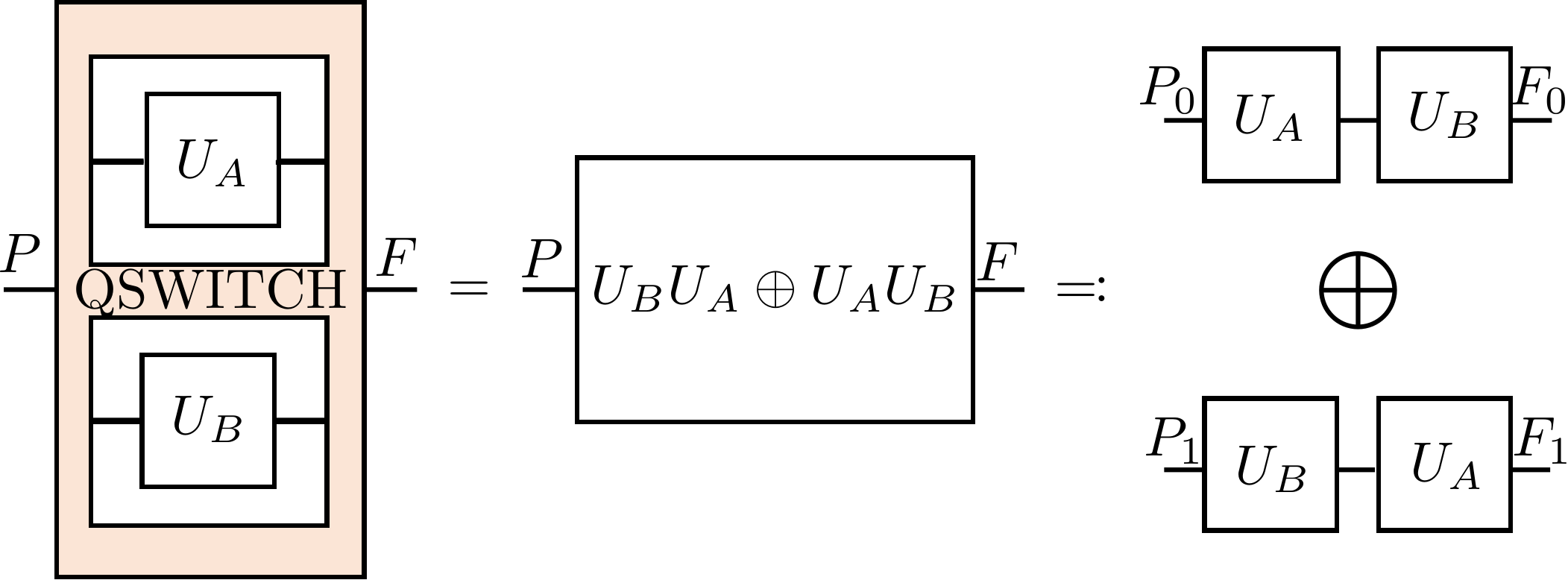}
		      \caption{$\mathrm{QSWITCH}$ is a quantum switch, which maps a pair of two \Add{independent} unitary channels represented by unitary operators $U_A$ and $U_B$ into a unitary channel represented by
			      $\ketbra{0}{0} \otimes U_B U_A + \ketbra{1}{1} \otimes U_A U_B = U_B U_A \oplus U_A U_B$. On the right hand side of this image, $ U_B U_A \oplus U_A U_B$ is illustrated as a direct sum of two circuits
			      where $P = P_0 \oplus P_1$ and $F = F_0 \oplus F_1$.}
		      \label{fig:qswitch_uaub}
	      \end{figure}

	\item A quantum superchannel is \textit{pure} if the quantum superchannel maps every group of \Add{independent} unitary channels into unitary channels (even if the superchannel acts on subsystems of the input unitary channels).

	      It has been shown that every pure superchannel corresponds to a unitary operator \cite{araujo16} (Figure \ref{fig:pureprocess_unitary}).
	      \begin{figure}[H]
		      \centering \includegraphics[keepaspectratio, scale=0.30]{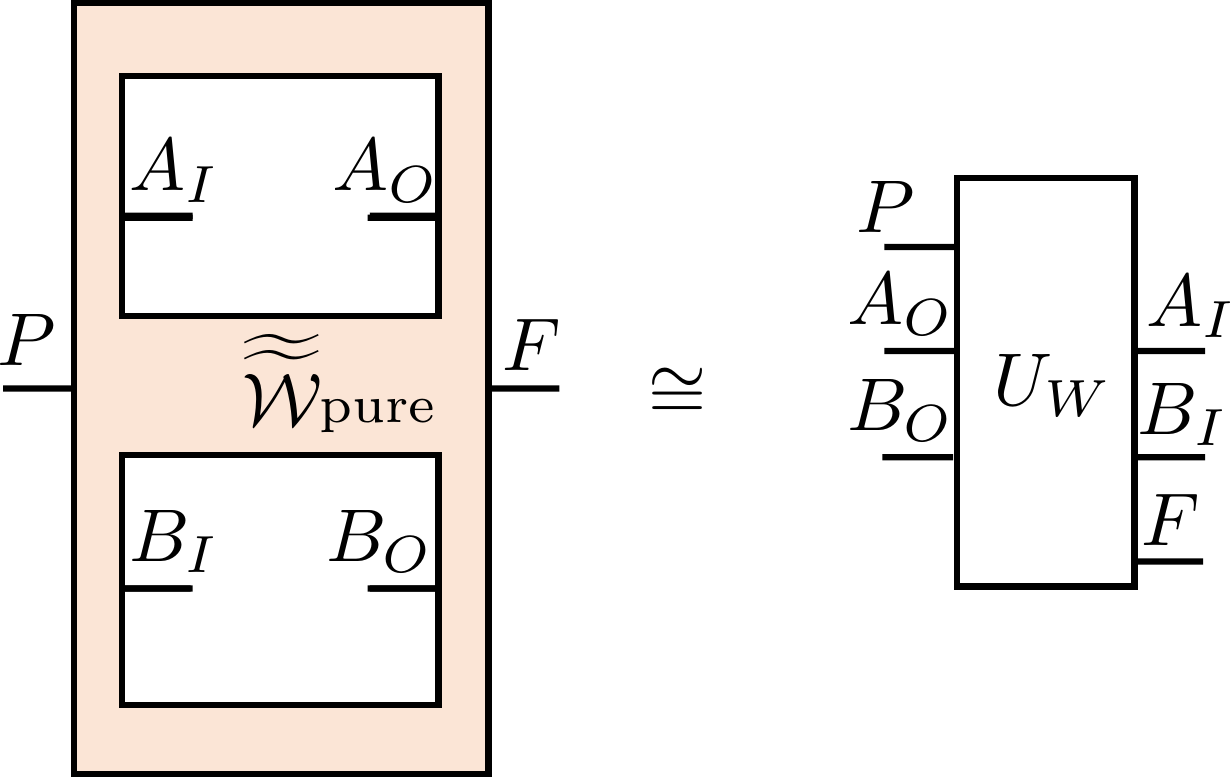}
		      \caption{$\smap{W}_{\mathrm{pure}}$ is a pure superchannel with two slots, which is represented by a unitary operator
		      $U_{W} \colon P \otimes A_O \otimes B_O \to A_I \otimes B_I \otimes F$ as $W_{\mathrm{pure}} = \dket{U_{{W}}} \dbra{U_{{W}}}$,
		      where $W_{\mathrm{pure}}$ and $\dket{U_W}$ are the Choi operator of the superchannel $\smap{W}_{\mathrm{pure}}$ and the Choi vector of the unitary operator $U_W$, respectively, defined in Sec.\,\ref{sec:CJisomorphism}.}
		      \label{fig:pureprocess_unitary}
	      \end{figure}

	      A quantum switch is an example of pure superchannels and can be represented by the unitary operator
	      \begin{align} \label{eq:U_QS}
		      U_{\mathrm{QS}}:= & \ket{0}^{F_c} \bra{0}^{P_c} \otimes \1^{P_t \to A_I} \otimes \1^{A_O \to B_I} \otimes \1^{B_O \to F_t}            \\
		      +                 & \ket{1}^{F_c} \bra{1}^{P_c} \otimes \1^{P_t \to B_I} \otimes \1^{B_O \to A_I} \otimes \1^{A_O \to F_t}, \nonumber
	      \end{align}
	      where \Add{$P_c$ and $P_t$ are component spaces of $P = P_c \otimes P_t$, $F_c$ and $F_t$ are those of $F=F_c\otimes F_t$, $\dim P_c = \dim F_c = 2$, and} $\1^{\mathcal{H}_I \to \mathcal{H}_O} : \mathcal{H}_I \to \mathcal{H}_O$ is the identity operator (Figure \ref{fig:qswitch}).

	      \begin{figure}[H]
		      \centering \includegraphics[keepaspectratio, scale=0.30]{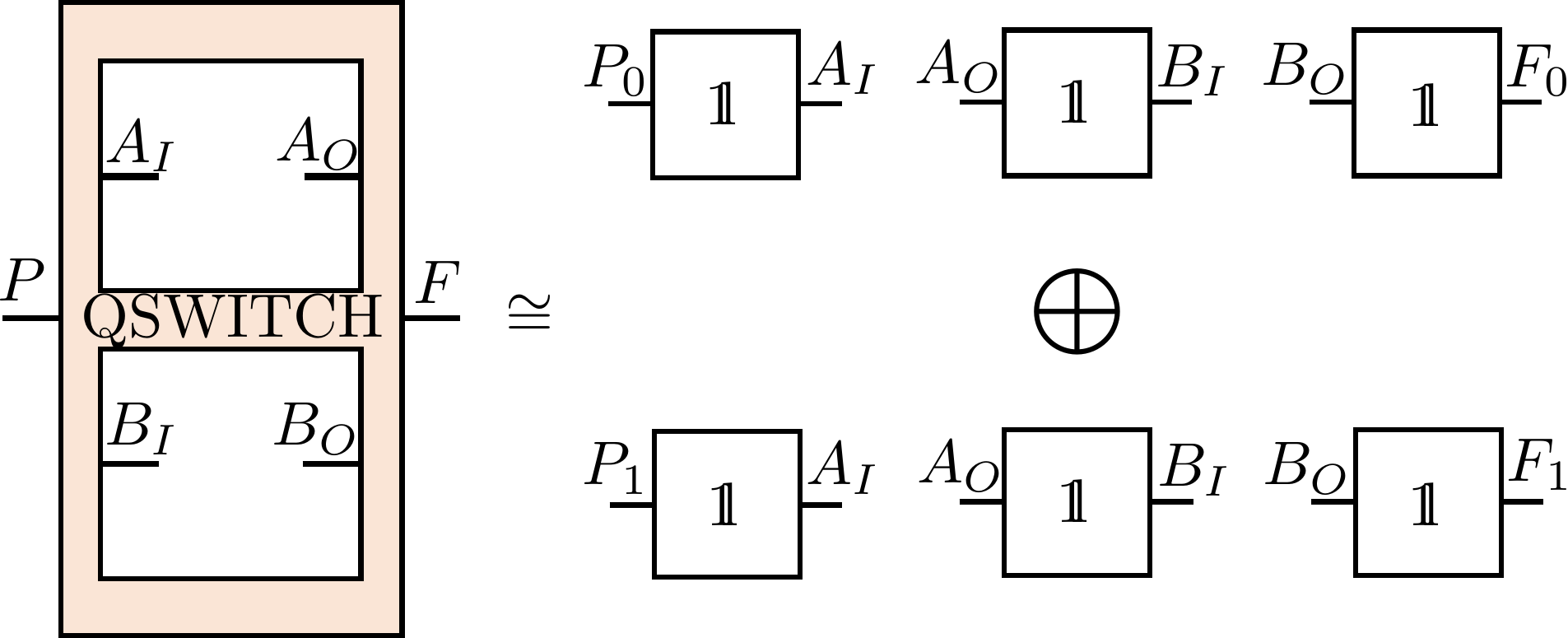}
		      \caption{The quantum switch $\mathrm{QSWITCH} = \dket{U_{\mathrm{QS}}} \dbra{U_{\mathrm{QS}}}$ is represented by the unitary operator
			      $U_{\mathrm{QS}}$ defined in Eq.\,\eqref{eq:U_QS}.}
		      \label{fig:qswitch}
	      \end{figure}
\end{itemize}

\begin{main2*} (Thm.\,\ref{thm:newgoal})
	Any pure superchannel with two slots is represented by a direct sum of pure quantum combs with different causal orders (Figure \ref{fig:main_result_2}).
\end{main2*}
\begin{figure}[H]
	\centering \includegraphics[keepaspectratio, scale=0.30]{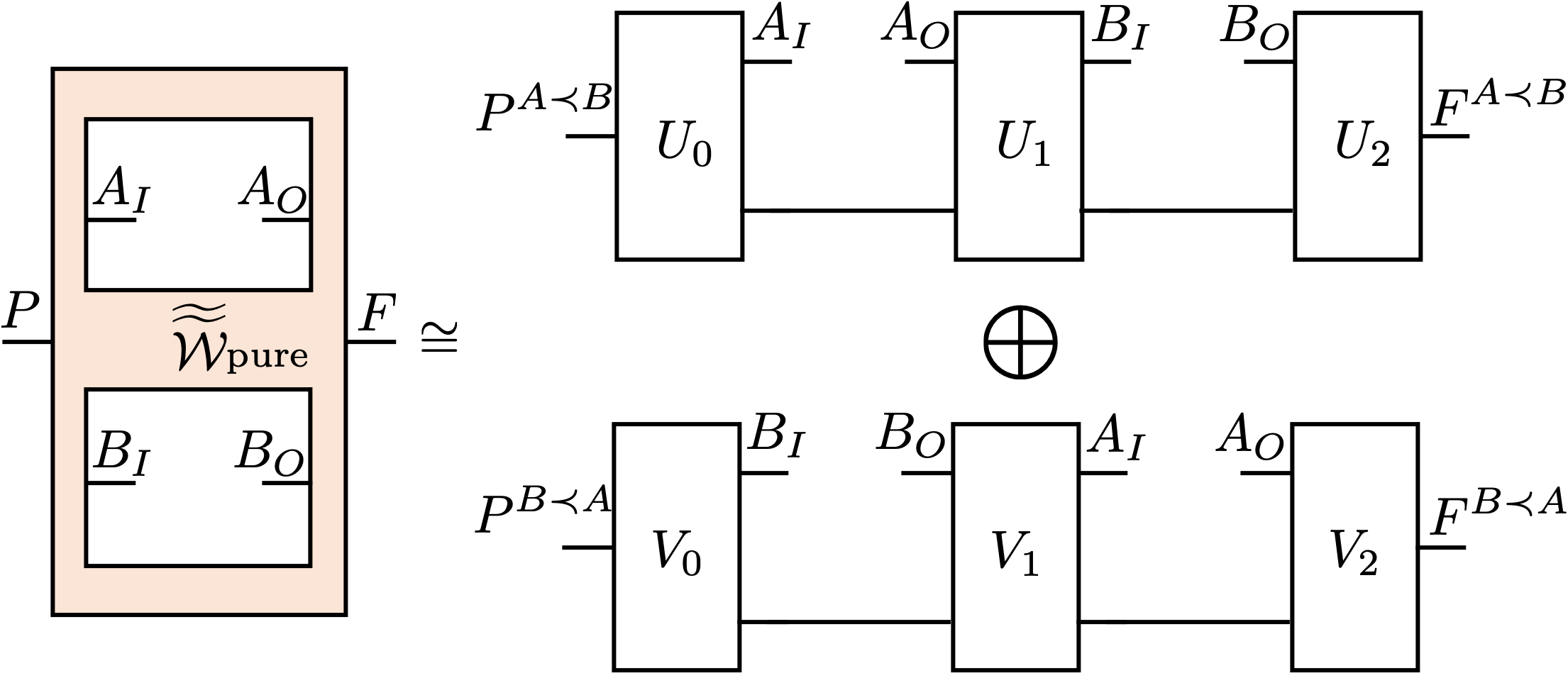}
	\caption{A pure superchannel $W_{\mathrm{pure}} = \dket{U_{W}} \dbra{U_{W}}$ with two slots is represented by a unitary operator $U_{W} = U_{\mathrm{comb}}^{A \prec B} \oplus V_{\mathrm{comb}}^{B \prec A}$
		with
		$U_{\mathrm{comb}}^{A \prec B} = U_2 U_1 U_0$ (identity matrices on auxiliary systems are abbreviated) representing a quantum comb of $A \prec B$ and $V_{\mathrm{comb}}^{B \prec A} = V_2 V_1 V_0$ representing a quantum comb of $B \prec A$
		where $P = P^{A \prec B} \oplus P^{B \prec A}$ and $F = F^{A \prec B} \oplus F^{B \prec A}$.
	}
	\label{fig:main_result_2}
\end{figure}


\section{Review on quantum superchannel} 	\label{sec:preliminary}


In this section, we review a mathematical formalism required to present our main results.
In this paper we consider only finite dimensional complex Hilbert spaces.
Notations in the rest of this paper are set as follows.
We sometimes denote a Hilbert space $\mathcal{H}_X$ just by its index $X$ ($X = P, A_I, A_O, B_I, B_O, F$).
When we refer to an operator, map, and supermap, we mean a linear operator, a linear map, and a linear supermap, respectively.
$\mathcal{L}(\mathcal{H})$ denotes the set of operators on a Hilbert space $\mathcal{H}$ and $[\mathcal{L}({\mathcal{H}_I}) \to \mathcal{L}({\mathcal{H}_O})]$ denotes the set of maps from $\mathcal{L}({\mathcal{H}_I})$ to $\mathcal{L}({\mathcal{H}_O})$.
We often add superscripts on vectors, operators, maps, and supermaps such as
\begin{itemize}
	\item a vector $\ket{\phi}^{\mathcal{H}} \in \mathcal{H}$,
	\item an operator $\rho^{\mathcal{H}} \in \mathcal{L}(\mathcal{H})$,
	\item a map $\map{E}^{\mathcal{H}_I \to \mathcal{H}_O} \colon \mathcal{L}(\mathcal{H_I}) \to \mathcal{L}(\mathcal{H_O})$, and
	\item a supermap $\smap{S}{}^{(A_I \to A_O) \to (P \to F)} \colon [\mathcal{L}({A_I}) \to \mathcal{L}({A_O})] \to [\mathcal{L}({P}) \to \mathcal{L}({F})]$,
\end{itemize}
to specify their attributions.
The domain and the range of a unitary operator are allowed to differ as long as their dimensions are equal.


\subsection{Choi-Jamio{\l}kowski isomorphism and linear supermaps}
\label{sec:CJisomorphism}


We use the Choi-Jamio{\l}kowski (CJ) isomorphism \cite{jamiolkowski1972linear, choi1975completely} to represent maps, which act on operators, as operators themselves on a larger Hilbert space. The Choi operator of a map $\map{E} \colon \mathcal{L}(\mathcal{H}_I) \to \mathcal{L}(\mathcal{H}_O)$ is defined as
\begin{equation}
	E := \sum_{ii^\prime} \ketbra{i}{i^\prime} \otimes \map{E}(\ketbra{i}{i^\prime}) \in \mathcal{L}(\mathcal{H}_I \otimes \mathcal{H}_O),
\end{equation}
where $\{ \ket{i} \}_i$ is the computational basis \ReviseFurther{of} $\mathcal{H}_I$.
The action of $\map{E}$ on $\rho \in \mathcal{L}(\mathcal{H}_I)$ can be obtained via the relation
\begin{equation}
	\map{E}(\rho) = \rho * E,
\end{equation}
where $*$ is an operation called the link product \cite{chiribella07} defined as follows:
\begin{defi} (\textit{Link product})
	Let $\mathcal{H}_0$, $\mathcal{H}_1$, $\mathcal{H}_2$ be different Hilbert spaces.
	The \textit{link product} $E_{01} * F_{12}$ between two operators $E_{01} \in \mathcal{L}(\mathcal{H}_0 \otimes \mathcal{H}_1)$ and $F_{12} \in \mathcal{L}(\mathcal{H}_1 \otimes \mathcal{H}_2)$ is
	\begin{equation}
		E_{01} * F_{12} := \Tr_1[ (E_{01}^{T_1} \otimes \1_{2}) (\1_{0} \otimes F_{12})]
	\end{equation}
	where $\Tr_1$ denotes the partial trace on $\mathcal{L}(\mathcal{H}_1)$ and $\bullet^{T_1}$ denotes the partial transposition on $\mathcal{L}(\mathcal{H}_1)$ in terms of the computational basis.
\end{defi}
We can consider a pure version of the Choi isomorphism. The Choi vector of a operator $A \colon \mathcal{H}_I \to \mathcal{H}_O$ is
\begin{equation}
	\dket{A} := \sum_{i} \ket{i} \otimes (A\ket{i}) \in \mathcal{H}_I \otimes \mathcal{H}_O.
\end{equation}
Then, the Choi operator of the operation $\map{E}_A \colon \mathcal{L}(\mathcal{H}_I) \to \mathcal{L}(\mathcal{H}_O)$ defined as $\map{E}_A(\rho) := A \rho A^\dagger$ is equal to $\dket{A} \dbra{A}$.
The action of $A$ on $\ket{\phi} \in \mathcal{H}_I$ can be obtained from
\begin{equation}
	A \ket{\phi} = \bradket{\phi^*|A}
\end{equation}
where $\ket{\phi^*}$ denotes the complex conjugation of $\ket{\phi}$ in the computational basis, that is, if $\ket{\phi} = \sum_{i} \phi_i \ket{i}$, then $\ket{\phi^*} = \sum_{i} \phi_i^* \ket{i}$.

A supermap is a linear transformation from a set of maps to a set of maps.
Similarly to maps, supermaps can also be represented by operators.
This isomorphism between supermaps and operators can be shown by first
showing that supermaps can be represented by maps \cite{oreshkov11, chiribella09}, and
then, by using the CJ isomorphism on this map
representation. In this paper, our map representation of supermaps
is slightly different from those used in previous references
\cite{oreshkov11, chiribella09, gour2019comparison}. Our convention for representing supermaps as maps provides
good insights for pure superchannels, which will be evident after
considering Thm.\,\ref{thm:W=UWUW} or Figure \ref{fig:pureprocess_unitary}.

Let
$\smap{S} \colon [\mathcal{L}({A_I}) \to \mathcal{L}({A_O})] \to [\mathcal{L}({P}) \to \mathcal{L}({F})]$ be a supermap from the set of
maps $\mathcal{L}({A_I}) \to \mathcal{L}({A_O})$ to the set of maps $\mathcal{L}({P}) \to \mathcal{L}({F})$.
Define a map $\map{S} \colon \mathcal{L}(P \otimes A_O) \to \mathcal{L}(A_I \otimes F)$ as follows:
\begin{multline}
	\label{eq:f8a98f9afaaa}
	\map{S}^{P A_O \to A_I F}(\rho^{P} \otimes \sigma^{A_O} ) := \\
	\left( \left(\smap{I}^{(\mathbb{C} \to A_I) \to (\mathbb{C} \to A_I)} \otimes \smap{S}^{(A_I \to A_O) \to (P \to F)} \right) \left( \map{I}^{A_I \to A_I} \otimes \tilde{\sigma}^{\mathbb{C} \to A_O} \right) \right)
	\left( \ketbra{0}{0}^{\mathbb{C}} \otimes \rho^{P} \right)
\end{multline}
for all $\rho^{P} \in \mathcal{L}(P)$ and $\sigma^{A_O} \in \mathcal{L}(A_O)$
where $\ket{0}^{\mathbb{C}}$ is a unit vector in $\mathbb{C}$, $\smap{I}^{(\mathbb{C} \to A_I) \to (\mathbb{C} \to A_I)} \colon [\mathcal{L}(\mathbb{C}) \to \mathcal{L}(A_I)] \to [\mathcal{L}(\mathbb{C}) \to \mathcal{L}(A_I)]$ is the identity supermap on the set of maps $\mathcal{L}(\mathbb{C}) \to \mathcal{L}(A_I)$,
$\map{I}^{A_I \to A_I} \colon \mathcal{L}(A_I) \to \mathcal{L}(A_I)$ is the identity map on $\mathcal{L}(A_I)$ and
$\tilde{\sigma}^{\mathbb{C} \to A_O} \colon \mathcal{L}(\mathbb{C}) \to \mathcal{L}(A_O)$ is an isometry which creates $\sigma^{A_O}$ defined as
\begin{equation}
	\tilde{\sigma}^{\mathbb{C} \to A_O}(\tau^\mathbb{C}) := \Tr[\tau^\mathbb{C}] \sigma^{A_O} \quad \forall \tau^\mathbb{C} \in \mathcal{L}(\mathbb{C})
\end{equation}
(see Figure \ref{fig:mapFromSupermap}).
Note that there is only one density operator on $\mathbb{C}$, which is $\ketbra{0}{0}^{\mathbb{C}}$.
\Add{Also note that the type of the tensor product of the channels and the argument type of the tensor product of the superchannels on the right hand side in Eq.\,\eqref{eq:f8a98f9afaaa} are equal due to a relation $(A_I \to A_I) \otimes (\mathbb{C} \to A_O) = (\mathbb{C} \to A_I) \otimes (A_I \to A_O)$.}
\begin{figure}
	\centering \includegraphics[keepaspectratio, scale=0.30]{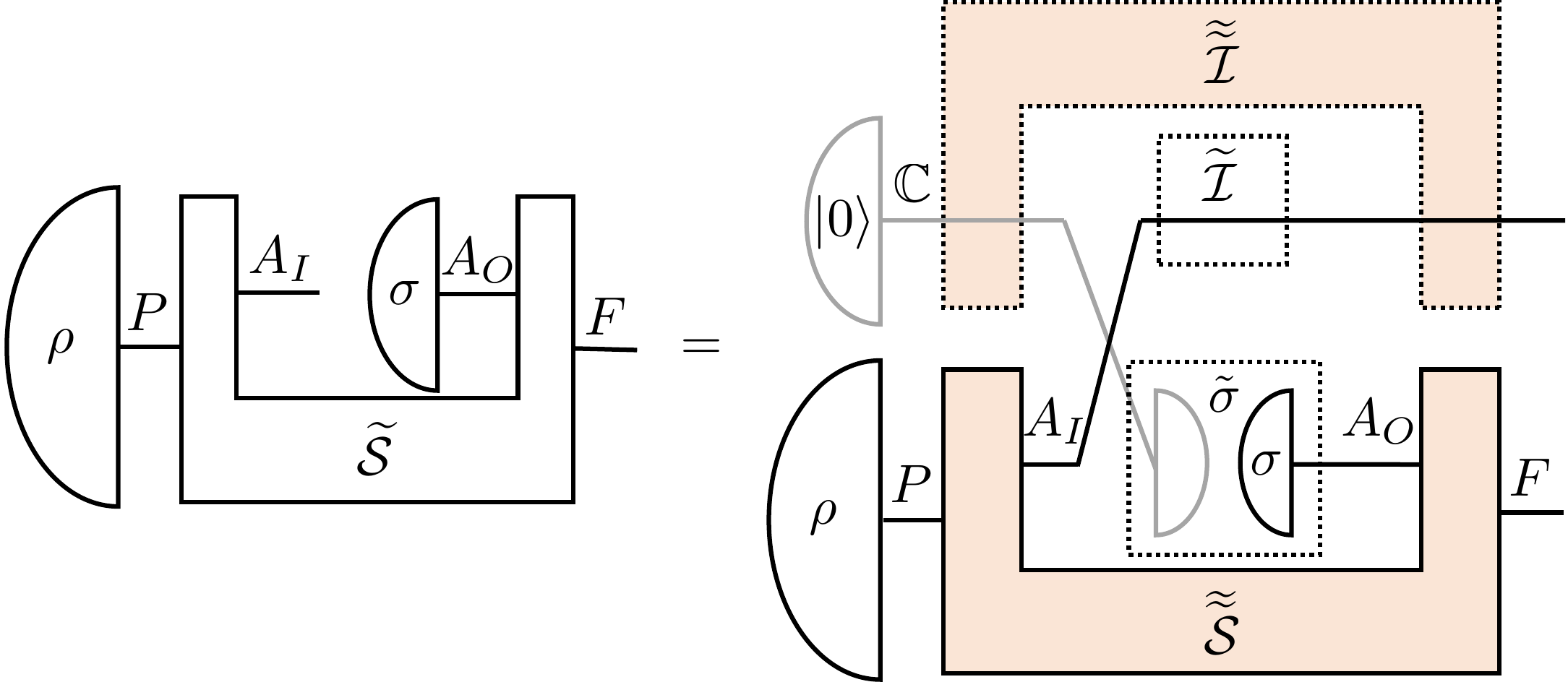}
	\caption{A supermap $\smap{S}$ can be seen as a map $\map{S}$ defined in Eq.\,\eqref{eq:f8a98f9afaaa}.
		The gray line represents $\mathbb{C}$, which is a one-dimensional quantum system.}
	\label{fig:mapFromSupermap}
\end{figure}
Let $S$ be the Choi operator of $\map{S}$.
The action of $\smap{S}$ on $\map{E} : \mathcal{L}(A_I) \to \mathcal{L}(A_O)$ can be retrieved via the relation in the CJ isomorphism
\begin{equation}
	E^\prime = S * E
\end{equation}
where $E$, $E^\prime$ are the Choi operators of $\map{E}$, $\smap{S}(\map{E})$, respectively. Now that $\smap{S}$, $\map{S}$ and $S$ are isomorphic, we can define $S$ as the Choi operator of $\smap{S}$.

In Refs.\,\cite{oreshkov11, chiribella09},
the authors represent a supermap $\smap{S} \colon [\mathcal{L}({A_I}) \to \mathcal{L}({A_O})] \to [\mathcal{L}({P}) \to \mathcal{L}({F})]$ as a map
$\map{S}_{\mathrm{prev}} := \choi^\prime \circ \smap{S} \circ \choi^{-1} \colon \mathcal{L}(A_I \otimes A_O) \to \mathcal{L}(P \otimes F)$, where
$\choi \colon [\mathcal{L}({A_I}) \to \mathcal{L}({A_O})] \to \mathcal{L}(A_I \otimes A_O)$ and $\choi^\prime \colon [\mathcal{L}({P}) \to \mathcal{L}({F})] \to \mathcal{L}(P \otimes F)$ transform maps into their Choi operators and they obtain the Choi operator of $\map{S}_{\mathrm{prev}}$ as the Choi operator of $\smap{S}$.
Although our map representation $\map{S}$ of the supermap $\smap{S}$ is different from previous one $\map{S}_{\mathrm{prev}}$, the Choi operators of $\map{S}_{\mathrm{prev}}$ and $\map{S}$ coincide. In Ref.\,\cite{gour2019comparison}, there is another representation of supermap as a map, but the Choi operator of that map also coincides with ours.


\subsection{Quantum superchannel}  	\label{sec:PureProcess}


We review the notions related to quantum superchannels.
We recall that quantum channels are completely positive and trace-preserving maps (CPTP maps).
Quantum superchannels with one slot are supermaps which transform quantum channels into quantum channels even if the supermaps act on a subsystem of the input quantum channels (see Figure \ref{fig:superchan}).

\begin{defi}
	(\textit{Quantum superchannel with one slot})
	Let $W$ be the Choi operator of a supermap $\smap{W} \colon [\mathcal{L}(A_I) \to \mathcal{L}(A_O)] \to [\mathcal{L}(P) \to \mathcal{L}(F)]$.
	We say that $\smap{W}$ is a \textit{quantum superchannel with one slot}
	if
	for all auxiliary Hilbert spaces $A_I^\prime$, $A_O^\prime$ and
	all quantum channels $\map{E} \colon \mathcal{L}(A_I \otimes A_I^\prime) \to \mathcal{L}(A_O \otimes A_O^\prime)$, $G = W \ast E$ is the Choi operator of
	a quantum channel $\map{G} \colon \mathcal{L}(P \otimes A_I^\prime) \to \mathcal{L}(F \otimes A_O^\prime)$,
	where $E$ is the Choi operator of $\map{E}$.
\end{defi}

In this paper we also consider superchannels with more than a single slot. \ReviseFurther{Multiple-slot} superchannels are supermaps transforming \Add{independent} quantum channels, which are represented by \ReviseFurther{the} tensor product of channels, into a quantum channel\footnote{Note that due to linearity, the action of superchannels with multiple slots is also defined for linear combinations of \ReviseFurther{the} tensor product of quantum channels. The set of multipartite quantum \ReviseFurther{channels} which can be written as a linear combination of \ReviseFurther{the} tensor product of single-party quantum channels is the set of \ReviseFurther{non-signaling} channels \cite{{chiribella09}}.} (see Figure \ref{fig:processt4qg}).

\begin{defi}
	(\textit{Quantum superchannel with multiple slots})
	\label{def:ProcessMatrix}
	Let $N$ be a positive integer, $\mathcal{H}_{m}$ be Hilbert spaces, $m \in \{ 0, \cdots, 2N+1\}$, and $W$ be
	the Choi operator of a supermap $\smap{W} \colon \bigotimes_{n=1}^{N} [ \mathcal{L}(\mathcal{H}_{2n-1}) \to \mathcal{L}(\mathcal{H}_{2n})] \to [\mathcal{L}(\mathcal{H}_{0}) \to \mathcal{L}(\mathcal{H}_{2N+1})]$.
	We say that $\smap{W}$ is a \textit{quantum superchannel with $N$ slots}
	if
	for all auxiliary Hilbert spaces $\mathcal{H}_{m}^\prime$, $m \in \{ 1, \cdots, 2N\}$, and
	all quantum channels $\map{E}_n \colon \mathcal{L}(\mathcal{H}_{2n-1} \otimes \mathcal{H}_{2n-1}^\prime) \to \mathcal{L}(\mathcal{H}_{2n} \otimes \mathcal{H}_{2n}^\prime)$, $n \in \{ 1, \cdots, N\}$,		$G = W \ast (\bigotimes_{n=1}^N E_n)$ is the Choi operator of
	a quantum channel $\map{G} \colon \mathcal{L}(\mathcal{H}_{0} \otimes \bigotimes_{n=1}^N \mathcal{H}_{2n-1}^\prime) \to \mathcal{L}(\mathcal{H}_{2N+1} \otimes \bigotimes_{n=1}^N \mathcal{H}_{2n}^\prime)$,
	where $E_n$ are the Choi operators of $\map{E}_n.$
\end{defi}

Quantum superchannels with a single slot are introduced in Ref.\,\cite{chiribella08} as deterministic quantum supermaps.
Quantum superchannels with $N$ slots appeared in Ref.\,\cite{chiribella09} as quantum supermaps on product channels.
The Choi operators of superchannels with $N$ slot are equal to $N$-partite process matrices \cite{oreshkov11} with a global past and global future \cite{araujo15, araujo16}.
$N$-partite process matrices with no global past and no global future are equivalent to the Choi operators of quantum superchannels with $N$ slots, where the first input space $\mathcal{H}_0$ and the last output space $\mathcal{H}_{2N+1}$ is the one-dimensional linear space $\mathbb{C}$.
In this paper, we mainly focus on quantum superchannels with two slots $[ \mathcal{L}(A_I) \to \mathcal{L}(A_O)] \otimes [\mathcal{L}(B_I) \to \mathcal{L}(B_O)] \to [\mathcal{L}(P) \to \mathcal{L}(F)]$ (see Figure \ref{fig:prsm}).
We name slots for maps $\mathcal{L}(A_I) \to \mathcal{L}(A_O)$, $\mathcal{L}(B_I) \to \mathcal{L}(B_O)$ as $A$, $B$, respectively. Slots like $A$ and $B$ in a quantum superchannel can be identified as parties of Alice and Bob in a process matrix.

\begin{figure}
	\centering \includegraphics[keepaspectratio, scale=0.30]{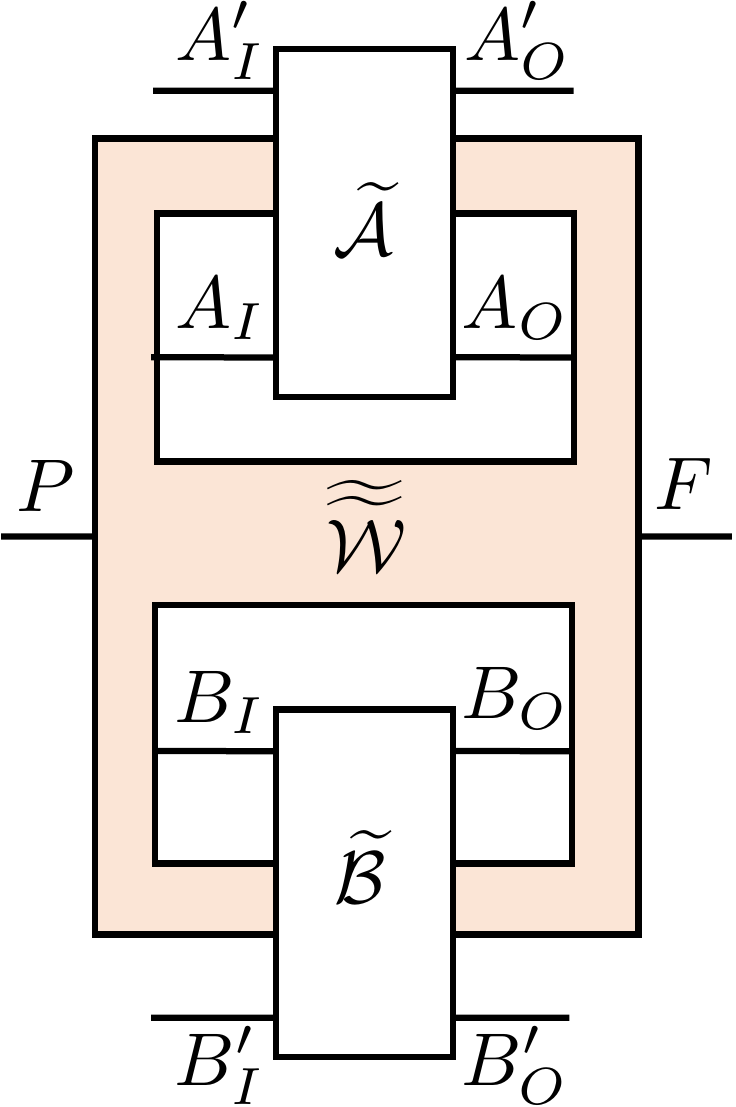}
	\caption{
		\label{fig:prsm}
		A quantum superchannel $\smap{W} \colon [\mathcal{L}(A_I \otimes B_I) \to \mathcal{L}(A_O \otimes B_O) ] \to [\mathcal{L}(P) \to \mathcal{L} (F)]$
		with two slots $A$ and $B$.
	}
\end{figure}

Next, we review the concept of \textit{pure superchannels}, which \ReviseFurther{are} first defined in Ref.\,\cite{araujo16} as \ReviseFurther{pure processes}\footnote{Note that Ref.\,\cite{araujo15} considers a different definition of pure process, where the authors define the pure process matrix to be a rank-one process matrix.}. We define unitary channels as maps $\map{U}(\rho) := U \rho U^\dagger$ for all $\rho \in \mathcal{L} (\mathcal{H}_I)$ with a unitary operator $U \colon \mathcal{H}_I \to \mathcal{H}_O$. Pure superchannels are quantum \ReviseFurther{superchannels} which transform a group of \Add{independent} unitary channels into \ReviseFurther{a unitary channel}.
\begin{defi}
	\label{def:PureProcessMatrix}
	(\textit{Pure} superchannel)
	Let $W$ be the Choi operator of
	a quantum superchannel $\smap{W} \colon \bigotimes_{n=1}^{N} [ \mathcal{L}(\mathcal{H}_{2n-1}) \to \mathcal{L}(\mathcal{H}_{2n})] \to [\mathcal{L}(\mathcal{H}_{0}) \to \mathcal{L}(\mathcal{H}_{2N+1})]$
	with $N$ slots.
	We say that $W$ (and $\smap{W}$) is \textit{pure} if for all auxiliary Hilbert spaces $\mathcal{H}_{m}^\prime$, $m \in \{ 1, \cdots, 2N\}$, and
	unitary operators $U_n : \mathcal{H}_{2n-1} \otimes \mathcal{H}_{2n-1}^\prime \to \mathcal{H}_{2n} \otimes \mathcal{H}_{2n}^\prime$, $n \in \{ 1, \cdots, N\}$,
	there exists a unitary operator $U_G : \mathcal{H}_{0} \otimes \bigotimes_{n=1}^N \mathcal{H}_{2n-1}^\prime \to \mathcal{H}_{2N+1} \otimes \bigotimes_{n=1}^N \mathcal{H}_{2n}^\prime$ such that
	$\dket{U_G} \dbra{U_G} = W \ast \bigotimes_{n=1}^N \dket{U_n} \dbra{U_n}$.
\end{defi}

Any pure superchannel can be represented as a unitary channel (see Figure \ref{fig:pureprocess_unitary}) from the following theorem:

\begin{theo} \
	\label{thm:W=UWUW}
	Let $W$ be the Choi operator of
	a quantum superchannel $\smap{W} \colon \bigotimes_{n=1}^{N} [ \mathcal{L}(\mathcal{H}_{2n-1}) \to \mathcal{L}(\mathcal{H}_{2n})] \to [\mathcal{L}(\mathcal{H}_{0}) \to \mathcal{L}(\mathcal{H}_{2N+1})]$.
	$W$ is pure if and only if there exists a unitary operator $U_W \colon \bigotimes_{n=0}^N \mathcal{H}_{2n} \to \bigotimes_{n=0}^N \mathcal{H}_{2n+1}$ such that
	$W = \dket{U_W} \dbra{U_W}$.
\end{theo}

In this paper, we say that $U_W$ represents the superchannel $\smap{W}$ in such cases.
A proof of this theorem follows the same steps of Theorem 2 in Ref.\,\cite{araujo16} which considers the case $N = 2$.
Note that the unitary operator representing a quantum superchannel is uniquely determined up to a global phase factor.

\begin{lemm}
	Let $U_W \colon P \otimes A_O \to A_I \otimes F$ and $U_A \colon A_I \otimes A_I^\prime \to A_O \otimes A_O^\prime$ be operators, not necessarily unitary. There exists an operator $U_G \colon P \otimes A_I^\prime \to F \otimes A_O^\prime$ such that
	\begin{equation}
		\dket{U_G} \dbra{U_G} = \dket{U_W} \dbra{U_W} \ast \dket{U_A} \dbra{U_A}.
	\end{equation}
	In addition,
	\begin{equation}
		U_G = \Tr_{A_O} \left[ \left( \1^{{F}} \otimes U_A \right) \left( U_W \otimes \1^{A_I^\prime} \right) \right]
		.
	\end{equation}
	A proof of this lemma follows from direct calculations and the fact that a global phase factor is irrelevant.
\end{lemm}

From this lemma,
$U_G$ in Def.\,\ref{def:PureProcessMatrix} can be calculated using $U_W$ in Thm.\,\ref{thm:W=UWUW}:
\begin{equation}
	\label{eq:ug8d0a0a0a0g9fg9}
	U_G = \Tr_{\mathcal{H}_{2} \mathcal{H}_{4} \cdots \mathcal{H}_{2N}} \left[ \left( \1^{\mathcal{H}_{2N+1}} \otimes \bigotimes_{n=1}^N U_n \right) \left( U_W \otimes \bigotimes_{n=1}^N \1^{\mathcal{H}_{2n-1}^\prime} \right) \right]
	.
\end{equation}


\subsection{Quantum comb}


Quantum combs (causally ordered superchannels) \Add{forms a special class of superchannels that can be realized by a quantum circuit (see Figure \ref{fig:quantum_comb}),} which are first introduced in Ref.\,\cite{chiribella07} and are mathematically equivalent to quantum channels with memory \cite{kretschmann05} and to quantum strategies \cite{gutoski07}.

\begin{defi}
	(\textit{Quantum comb}, \Add{see Figure \ref{fig:comb_maps_comb_to_comb}})
	\label{def:ChoiOpOfQuantumComb}

	We say that $\smap{R}{}^{(2)}$ is a \textit{quantum comb with one slot} if $\smap{R}{}^{(2)}$ is a quantum superchannel with one slot.
	Let $N \geq 2$ be a positive integer and let $\mathcal{H}_{m}$ be Hilbert spaces, $m \in \{ 0, \cdots, 2N+1\}$, and $R^{(N+1)}$ be the Choi operator of
	a supermap $\smap{R}{}^{(N+1)} \colon \bigotimes_{n=1}^{N} [\mathcal{L}(\mathcal{H}_{2n-1}) \to \mathcal{L}(\mathcal{H}_{2n})] \to [\mathcal{L}(\mathcal{H}_{0}) \to \mathcal{L}(\mathcal{H}_{2N+1})]$.
	We say that $\smap{R}{}^{(N+1)}$ is a \textit{quantum comb with $N$ slots} if for all auxiliary Hilbert spaces $\mathcal{H}_{1}^\prime$ and $\mathcal{H}_{2N}^\prime$ and
	all quantum combs with $N-1$ slots
	$\smap{S}{}^{(N)} \colon \bigotimes_{n=1}^{N-1}[ \mathcal{L}(\mathcal{H}_{2n}) \to \mathcal{L}(\mathcal{H}_{2n+1})] \to [\mathcal{L}(\mathcal{H}_{1} \otimes \mathcal{H}_{1}^\prime) \to \mathcal{L}(\mathcal{H}_{2N} \otimes \mathcal{H}_{2N}^\prime)]$,
	$G = R^{(N+1)} * S^{(N)}$ is the Choi operator of a quantum channel
	$\map{G} \colon \mathcal{L}(\mathcal{H}_{0} \otimes \mathcal{H}_{1}^\prime) \to \mathcal{L}(\mathcal{H}_{2N+1} \otimes \mathcal{H}_{2N}^\prime)$ where
	$S^{(N)}$ is the Choi operator of $\smap{S}{}^{(N)}$.
\end{defi}

Note that the slots of a quantum comb have an order. In the above definition, the $n$-th slot is a slot whose input space is $\mathcal{H}_{2n-1}$ and whose output space is $\mathcal{H}_{2n}$
($n=1 \cdots N$). The orders of slots of quantum combs are suppressed throughout the remaining text unless otherwise misunderstood.

Similarly to pure superchannels, we can define the pure combs.
\begin{defi}
	(\textit{Pure} comb)

	A \textit{pure} comb is a quantum comb which is a pure superchannel.
\end{defi}

Any quantum superchannel with one slot is a quantum comb with one slot \cite{chiribella08}.
Suppose a quantum comb with two slots has
the first slot $A$ and the second slot $B$, then we call it a quantum comb of $A \prec B$ (read as ``$A$ to $B$'')\footnote{\ReviseFurther{We stress that $A \prec B$ does not necessarily imply that $A$ signals to $B$ but that $B$ is forbidden to signal to $A$. In Supplementary Information of Ref.\,\cite{oreshkov11}, signaling from $A$ to $B$ is formally defined as ``the existence of statistical correlations between a random variable at $A$ which can be chosen freely, and another random variable at $B$'', denoted by ``$A \preceq B$''. $A \prec B$ in this paper is equal to ``$A \nsucceq B$'' in that reference. }}.
Similarly, we can consider the notion of a quantum comb of $B \prec A$.
If a quantum superchannel can be regarded as a quantum comb of $A \prec B$ and as a quantum comb of $B \prec A$, it is parallel and we call it a quantum comb of $A \parallel B$ (read as ``$A$ parallel to $B$'').

A necessary and sufficient condition for a supermap to be a quantum comb can be obtained from the following theorem \cite{chiribella09b} and any quantum comb is equivalent to a quantum circuit \cite{chiribella07} (see Figure \ref{fig:quantum_comb}):

\begin{theo}
	\label{thm:purecombcausalityforall}
	Let $N$ be a positive integer and let $\mathcal{H}_{m}$ be Hilbert spaces, $m \in \{ 0, \cdots, 2N+1\}$, and $R^{(N+1)}$ be the Choi operator of
	a supermap $\smap{R}{}^{(N+1)} \colon \bigotimes_{n=1}^{N} [ \mathcal{L}(\mathcal{H}_{2n-1}) \to \mathcal{L}(\mathcal{H}_{2n})] \to [\mathcal{L}(\mathcal{H}_{0}) \to \mathcal{L}(\mathcal{H}_{2N+1})]$.
	$\smap{R}{}^{(N+1)}$ is a quantum comb with $N$ slots
	if and only if $R^{(N+1)} \geq 0$ and for each $n \in \{ 0, \cdots, N\}$,
	\begin{equation}
		\Tr_{\mathcal{H}_{2n+1} \cdots \mathcal{H}_{2N-1} \mathcal{H}_{2N+1}} R^{(N+1)} = \1^{\mathcal{H}_{2n} \cdots \mathcal{H}_{2N-2} \mathcal{H}_{2N}} \otimes R^{(n)}
	\end{equation}
	holds
	where $R^{(n)}$ ($n \geq 2$) is the Choi operator of a quantum comb with $n-1$ slots, $R^{(1)}$ is the Choi operator of a quantum channel (``quantum comb with $0$ slot'') and  $R^{(0)} = 1$ (``quantum comb with $-1$ slot'').
\end{theo}


\section{Main results and their implications} \label{sec:main}


We are now ready to precisely state our main results and their implications.


\subsection{Main result 1: Pure comb}


Every pure comb with $N$ slots can be realized by $(N + 1)$ unitary channels and does not require discarding any system (see Figure \ref{fig:f89ga89hujap}).
\begin{theo} \
	\label{thm:PureCombDecomposition}
	Let $U^{(N+1)} : \bigotimes_{n=0}^N \mathcal{H}_{2n} \to \bigotimes_{n=0}^N \mathcal{H}_{2n+1}$ be a unitary operator representing a pure quantum comb
	$\smap{R}{}^{(N+1)} \colon \bigotimes_{n=1}^{N} [ \mathcal{L}(\mathcal{H}_{2n-1}) \to \mathcal{L}(\mathcal{H}_{2n})] \to [\mathcal{L}(\mathcal{H}_{0}) \to \mathcal{L}(\mathcal{H}_{2N+1})]$
	with $N$ slots.
	There exist integers $\{ k_n \}_{n=0}^{N+1}$ such that $k_0 = k_{N+1} = 1$ and
	\begin{equation}
		d_{2n} k_n = d_{2n+1} k_{n+1}, \quad n \in \{ 0, \cdots, N\},
	\end{equation}
	where $d_{m} := \dim \mathcal{H}_{m}$.
	$U^{(N+1)}$ can be decomposed into
	\begin{equation}
		\label{eq:UR=UN...U1U0}
		U^{(N+1)} = U_N \cdots U_1 U_0
	\end{equation}
	by using some unitary operators $ U_n:\mathcal{H}_{2n} \otimes \mathcal{A}_{n} \to \mathcal{H}_{2n+1} \otimes \mathcal{A}_{n+1}$ ($n=0 \cdots N+1$)
	for a Hilbert spaces $\mathcal{A}_{n} $ ($n=0 \cdots N+1$) such that $\dim \mathcal{A}_{n} = k_n$. (In Eq.\,\eqref{eq:UR=UN...U1U0}, it is implicitly assumed that $U_n$ trivially acts on all the spaces but $\mathcal{H}_{2n} \otimes \mathcal{A}_n$.)
\end{theo}
\begin{proof}
	See Appendix \ref{secApp:ProofOfPureCombDecomposition}.
\end{proof}
\Add{Proof techniques used in Appendix \ref{secApp:ProofOfPureCombDecomposition} are extensibly used to show the next main result.}

In Ref.\,\cite{Bisio_2011} the authors prove that every quantum comb with $N$ slots can be realized by $(N + 1)$ isometries followed by a partial trace which represents discarding part of the quantum system. Theorem \ref{thm:PureCombDecomposition} shows that a quantum comb is pure if and only if all these isometries can be replaced with unitaries and no quantum system is discarded. We refer to Figure \ref{fig:quantum_comb} for a pictorial illustration of general combs and Figure \ref{fig:f89ga89hujap} for a pictorial illustration of pure combs.


\subsection{Main result 2: Pure superchannel}


Before our second main statement, we introduce \ReviseFurther{the concept of a} \textit{direct sum of pure combs} as a generalization of the quantum switch (see Figure \ref{fig:qswitch_uaub}).
We observe that the quantum switch is represented by the following unitary operator (see Figure \ref{fig:qswitch} and Eq.\,\eqref{eq:U_QS}):
\footnotesize
\begin{align}
	\label{eq:unitaryofqs}
	U_{\mathrm{QS}} = & \ket{0}^{F_c} \bra{0}^{P_c} \otimes \1^{P_t \to A_I} \otimes \1^{A_O \to B_I} \otimes \1^{B_O \to F_t} + \ket{1}^{F_c} \bra{1}^{P_c} \otimes \1^{P_t \to B_I} \otimes \1^{B_O \to A_I} \otimes \1^{A_O \to F_t} \\
	\cong             & \1^{P^{A \prec B} \to A_I} \otimes \1^{A_O \to B_I} \otimes \1^{B_O \to F^{A \prec B}} \oplus \1^{P^{B \prec A} \to B_I} \otimes \1^{B_O \to A_I} \otimes \1^{A_O \to F^{B \prec A}} \nonumber,
\end{align}
\normalsize
where\footnote{$\oplus$ denotes the orthogonal direct sum of linear subspaces. See Appendix \ref{secApp:linearAlgebra}.}
$X = X_c \otimes X_t = X^{A \prec B} \oplus X^{B \prec A}$,
$\dim X_c = 2$, $\dim X_t = \dim X^{A \prec B} = \dim X^{B \prec A} = \dim A_I = \dim A_O = \dim B_I = \dim B_O$ ($X = P, F$).
Clearly, $\1^{P^{A \prec B} \to A_I} \otimes \1^{A_O \to B_I} \otimes \1^{B_O \to F^{A \prec B}}$ is a pure quantum comb of $A \prec B$ and
$\1^{P^{B \prec A} \to B_I} \otimes \1^{B_O \to A_I} \otimes \1^{A_O \to F^{B \prec A}}$ is a pure quantum comb of $B \prec A$.
The Choi vector form of this operator is
\begin{equation}
	\dket{U_{\mathrm{QS}}} =
	\ket{0}^{P_c} \dket{\1}^{P_t A_I} \dket{\1}^{A_O B_I} \dket{\1}^{B_O F_t} \ket{0}^{F_c} +
	\ket{1}^{P_c} \dket{\1}^{P_t B_I} \dket{\1}^{B_O A_I} \dket{\1}^{A_O F_t} \ket{1}^{F_c}.
\end{equation}

Now we define the direct sum of pure combs.

\begin{defi} (\textit{Direct sum of pure combs})
	\label{def:qsoqc}
	Let $U^{A \prec B} : P^{A \prec B} \otimes A_O \otimes B_O \to A_I \otimes B_I \otimes F^{A \prec B}$ be a unitary operator representing a quantum comb of $A \prec B$ and $U^{B \prec A} : P^{B \prec A} \otimes A_O \otimes B_O \to A_I \otimes B_I \otimes F^{B \prec A}$ be a unitary operator representing a quantum comb of $B \prec A$.
	A supermap represented by a unitary operator $U \colon P \otimes A_O \otimes B_O \to A_I \otimes B_I \otimes F$ defined as
	\begin{equation}
		\label{eq:defUqsoqc}
		U = U^{A \prec B} \oplus U^{B \prec A}
	\end{equation}
	is a \textit{direct sum of pure combs with two slots $A$ and $B$}, where $P := P^{A \prec B} \oplus P^{B \prec A}$ and $F := F^{A \prec B} \oplus F^{B \prec A}$.

\end{defi}
This definition can also be written in terms of Choi vectors. Equation \eqref{eq:defUqsoqc} is equivalent to $\dket{U} = \dket{U^{A \prec B}} + \dket{U^{B \prec A}}$.
Note that $\dket{U^{A \prec B}} \in P^{A \prec B} \otimes A_O \otimes B_O \otimes A_I \otimes B_I \otimes F^{A \prec B}$ and $\dket{U^{B \prec A}} \in P^{B \prec A} \otimes A_O \otimes B_O \otimes A_I \otimes B_I \otimes F^{B \prec A}$
although the sum of the two Choi vectors is taken in $P \otimes A_O \otimes B_O \otimes A_I \otimes B_I \otimes F$, which is valid because $P^{A \prec B} \otimes A_O \otimes B_O \otimes A_I \otimes B_I \otimes F^{A \prec B}$ and $P^{B \prec A} \otimes A_O \otimes B_O \otimes A_I \otimes B_I \otimes F^{B \prec A}$ are subspaces of $P \otimes A_O \otimes B_O \otimes A_I \otimes B_I \otimes F$.
Moreover, $P = P^{A \prec B} \oplus P^{B \prec A}$ and $F = F^{A \prec B} \oplus F^{B \prec A}$ are satisfied. If not satisfied, then $U$ is not unitary.
The quantum switch (Eq.\,\eqref{eq:unitaryofqs})
is an example of a direct sum of pure combs such that
\begin{equation}
	U^{A \prec B} = \ket{0}^{F_c} \bra{0}^{P_c} \otimes \1^{P_t \to A_I} \otimes \1^{A_O \to B_I} \otimes \1^{B_O \to F_t},
\end{equation}
\begin{equation}
	U^{B \prec A} = \ket{1}^{F_c} \bra{1}^{P_c} \otimes \1^{P_t \to B_I} \otimes \1^{B_O \to A_I} \otimes \1^{A_O \to F_t}
\end{equation}
where
\begin{equation}
	\label{PABinQS}
	P^{A \prec B} = \text{SPAN}[\ket{0}^{P_c}] \otimes P_t,
\end{equation}
\begin{equation}
	P^{B \prec A} = \text{SPAN}[\ket{1}^{P_c}] \otimes P_t,
\end{equation}
\begin{equation}
	F^{A \prec B} = \text{SPAN}[\ket{0}^{F_c}] \otimes F_t,
\end{equation}
\begin{equation}
	\label{QBAinQS}
	F^{B \prec A} = \text{SPAN}[\ket{1}^{F_c}] \otimes F_t,
\end{equation}
and $\text{SPAN}[\ket{\psi}]$ is the one-dimensional subspace spanned by $\ket{\psi}$, \eg,
\begin{equation}
	\text{SPAN}[\ket{0}^{P_c}] := \{ c \ket{0}^{P_c} | c \in \mathbb{C} \}.
\end{equation}

Our second main result is that a pure superchannel with two slots is a direct sum of pure combs.

\begin{theo}
	\label{thm:newgoal}
	All pure superchannels with two slots are equivalent to direct sums of pure combs.
\end{theo}

More formally, we prove the following theorem.

\begin{theo} \label{thm:equivalent}
	Let $U : P \otimes A_O \otimes B_O \to A_I \otimes B_I \otimes F$ be a unitary operator. The following conditions are equivalent:
	\begin{description}
		\item[(1)] $U$ represents a pure \ReviseFurther{quantum} superchannel with two slots $A$ and $B$.
		\item[(2)] For all $\ket{\pi}^P,\, \ket{{\pi}^\prime}^P \in P$,
		$\ket{\alpha}^{A_O},\, \ket{\alpha^\prime}^{A_O},\, \ket{\alpha^\perp}^{A_O} \in A_O$ such that $\ket{\alpha}^{A_O} \perp \ket{\alpha^\perp}^{A_O}$,
		$\ket{\beta}^{B_O},\, \ket{\beta^\prime}^{B_O},\, \ket{\beta^\perp}^{B_O} \in B_O$ such that $\ket{\beta}^{B_O} \perp \ket{\beta^\perp}^{B_O}$,
		$\dket{\Gamma}^{A_I B_I},\,\\ \dket{{\Gamma}^\prime}^{A_I B_I} \in A_I \otimes B_I$,
		$\ket{\phi}^{A_I},\, \ket{{\phi}^\prime}^{A_I} \in A_I$ and $\ket{\psi}^{B_I},\, \ket{{\psi}^\prime}^{B_I} \in B_I$,
		\begin{equation}
			\label{a}
			\left(
			{}^{A_I B_I}\dbra{\Gamma} U \left( \ket{\pi}^P \ket{\alpha}^{A_O} \ket{\beta}^{B_O} \right), {}^{A_I B_I}\dbra{{\Gamma}^\prime}
			U \left( \ket{{\pi}^\prime}^P \ket{\alpha^\perp}^{A_O} \ket{\beta^\perp}^{B_O} \right)
			\right) = 0,
			\tag{a}
		\end{equation}
		\begin{equation}
			\label{b}
			\left(
			{}^{B_I}\bra{\psi} U \left( \ket{\pi}^P \ket{\alpha}^{A_O} \ket{\beta}^{B_O} \right), {}^{B_I}\bra{{\psi}^\prime}
			U \left( \ket{{\pi}^\prime}^P \ket{\alpha^\prime}^{A_O} \ket{\beta^\perp}^{B_O} \right)
			\right) = 0,
			\tag{b}
		\end{equation}
		\begin{equation}
			\label{c}
			\left(
			{}^{A_I}\bra{\phi} U \left( \ket{\pi}^P \ket{\alpha}^{A_O} \ket{\beta}^{B_O} \right), {}^{A_I}\bra{{\phi}^\prime}
			U \left( \ket{{\pi}^\prime}^P \ket{\alpha^\perp}^{A_O} \ket{\beta^\prime}^{B_O} \right)
			\right) = 0.
			\tag{c}
		\end{equation}
		\item[(3)] $U$ represents a direct sum of pure combs with two slots $A$ and $B$.
	\end{description}
\end{theo}
The implication (1) $\implies$ (2) is proved from straightforward calculations (Appendix \ref{secApp:ProofsOfpreparation}).
The condition (2) can be simply rephrased in terms of the reduced subspace, introduced in Appendix \ref{secApp:partsp}. After rephrasing, we provide the proof of (2) $\implies$ (3) in Appendix \ref{secApp:idea}.
The last implication (3) $\implies$ (1) is clear since if we insert two arbitrary unitary operators $U_A : A_I \otimes A_I^\prime \to A_O \otimes A_O^\prime$, $U_B : B_I \otimes B_I^\prime \to B_O \otimes B_O^\prime$, in the two slots $A$ and $B$, the output is also a unitary operation. This \Add{implication (3) $\implies$ (1)} can be rigorously verified via Lem.\,\ref{lem:Uisprocess<=>UGisunitary}.

\Add{Eq.\,\eqref{a} in the condition (2) can be interpreted as follows.
	First, we consider ``feeding'' a pure state $\ket{\alpha}^{A_O}$ into $A_O$ and a pure state $\ket{\beta}^{B_O}$ into $B_O$, that is, $\ket{\alpha}^{A_O}$ and $\ket{\beta}^{B_O}$ are used as input-states in $A_O$ and $B_O$ for $U$, respectively, and an arbitrary input-state $\ket{\pi}^{P}$ in $P$ for $U$ \ReviseFurther{is} used.
	Next, we ``feed'' another pure state $\ket{\alpha^\perp}^{A_O}$ orthogonal to $\ket{\alpha}^{A_O}$ into $A_O$ and another pure state $\ket{\beta^\perp}^{B_O}$ orthogonal to $\ket{\beta}^{B_O}$ into $B_O$.
	\ReviseFurther{Eq.\,\eqref{a} means that the reduced states on $F$ of the two corresponding output-states are orthogonal.}
	Briefly, \ReviseFurther{Eq.\,\eqref{a}} means that ``feeding'' orthogonal pure states into $A_O$ and $B_O$ causes orthogonal reduced states on $F$. Similarly, \ReviseFurther{Eq.\,\eqref{b} (Eq.\,\eqref{c})} means that ``feeding'' orthogonal pure states into \ReviseFurther{$B_O$, ($A_O$)} causes orthogonal reduced states on  \ReviseFurther{$A_I \otimes F$ ($B_I \otimes F$)}, respectively.
	In the proof of (2) $\implies$ (3), 
	we ``feed'' pure states into $A_O$ and $B_O$, analyze the actions of $U$ with the condition (2) changing the fed states.}
	\ReviseFurther{Then we specify the subspaces characterized by ``$A \prec B$ vectors'' and the subspaces characterized by ``$B \prec A$ vectors''  in $P$ and $F$ to decompose $P$ and $F$ into $P = P^{A \prec B} \oplus P^{B \prec A}$ and $F = F^{A \prec B} \oplus F^{B \prec A}$.}
	\ReviseFurther{More detailed explanations for the core ideas of the proof are given in Appendix \ref{secApp:idea}.}


\subsection{Implications of the main results}


\Add{
	Recall that the Choi operators of quantum superchannels with two slots $[[ \mathcal{L}(A_I) \to \mathcal{L}(A_O)] \otimes [\mathcal{L}(B_I) \to \mathcal{L}(B_O)]] \to [\mathcal{L}(P) \to \mathcal{L}(F)]$ are equivalent to bipartite process matrices \cite{oreshkov11,araujo16,araujo15} with two parties $A$ and $B$, global past $P$ and global future $F$.
	More precisely, a bipartite process matrix with global past and global future is the Choi operator $W\in \mathcal{L}(P\otimes A_I\otimes A_O \otimes B_I\otimes B_O \otimes F)$ of a two-slot superchannel presented in Def.\,\ref{def:ProcessMatrix}, that is, $W\geq0$ and for any pair of Choi operators of channels $E_A\in \mathcal{L}(A_I\otimes A_O)$ and $E_B\in \mathcal{L}(B_I\otimes B_O)$, the operator
	$G=W*(E_{\ReviseFurther{A}}\otimes E_{\ReviseFurther{B}})\in\mathcal{L}(P\otimes F)$ is the Choi operator of a quantum channel. Process matrices with no global past are equivalent to process matrices where the dimension of the linear space $P$ is one, process matrices with no global future are equivalent to process matrices where the dimension of the linear space $F$ is one.
}

\Add{
With the aid of Thm.\,\ref{thm:newgoal}, we have three implications about a bipartite process matrix, in particular, about \textit{purifiable} process matrices\,\cite{araujo16}.
A process matrix with a global past and global future is \textit{pure} if it is a pure superchannel.
A process matrix $W$ with a global past and global future is \textit{purifiable} if it can be obtained from a pure process matrix $W^\prime_{\mathrm{pure}}$ with global past $P$ and global future $F$ by
inserting a fixed state $\ket{0}^{P_{\mathrm{sub}}}$ in a subsystem $P_{\mathrm{sub}}$ of $P$ and tracing out a subsystem $F_{\mathrm{sub}}$ of $F$, that is,
$W = W^\prime_{\mathrm{pure}} * (\ketbra{0}{0}^{P_{\mathrm{sub}}} \otimes \1^{F_{\mathrm{sub}}})$.
The concept of  purifiable process matrix was introduced based on the principle of ``reversibility of the transformations between quantum states'' in Ref.\,\cite{araujo16} and the authors propose a postulate that a process matrix admits a ``fair'' physical realization only if it is purifiable.
}

\Add{
	The first implication of our main result connects \textit{causally separable process matrices}\,\cite{oreshkov11} without global past and without global future and purifiable process matrices.
	A bipartite process matrix $W\in \mathcal{L}(A_I\otimes A_O \otimes B_I\otimes B_O)$  without global past and without global future is said to be causally separable if it can be written as is a probabilistic mixture of causally ordered process matrices\footnote{A bipartite process matrix $W^{A \prec B} \in \mathcal{L}(P\otimes A_I\otimes A_O \otimes B_I\otimes B_O \otimes F)$ is causally ordered from $A$ to $B$ if it is the Choi operator of a quantum comb where the slot $A$ comes before the slot $B$. For the definition of causal separability presented here we only consider the case where the dimension of linear spaces $P$ and $F$ is one.}  $W^{A \prec B}\in \mathcal{L}(A_I\otimes A_O \otimes B_I\otimes B_O)$  and $W^{B\prec A} \in \mathcal{L}(A_I\otimes A_O \otimes B_I\otimes B_O)$, \textit{i.e.,} $W=q W^{A\prec B} + (1-q)W^{B\prec A}$ for some $q\in [0,1]$.
}

\Add{
	\begin{coro} \label{cor:purifiableImplyCausalsep}
		All purifiable bipartite process matrices without a global past and without global future are causally separable.
	\end{coro}
}
\begin{proof}
	\Add{
		From Thm.\,\ref{thm:newgoal}, all pure processes $W_{\mathrm{pure}}\in \mathcal{L}(P\otimes A_I\otimes A_O \otimes B_I\otimes B_O \otimes F)$ can be written as
		\begin{equation}
			W_{\mathrm{pure}}=\dket{U_A}\dbra{U_A} + \dket{U_A}\dbra{U_B} +\dket{U_B}\dbra{U_A} +\dket{U_B}\dbra{U_B}
		\end{equation}
		where the future space of $\dket{U_A}$ is orthogonal to the future space of $\dket{U_B}$. Due to this orthogonality relation, pure processes always satisfy
		\begin{align}
			\Tr_F[W_{\mathrm{pure}}] & = \Tr_F[ \dket{U_A}\dbra{U_A} + \dket{U_A}\dbra{U_B} +\dket{U_B}\dbra{U_A} +\dket{U_B}\dbra{U_B} ] \\
			                         & = \Tr_F[ \dket{U_A}\dbra{U_A} + \dket{U_B}\dbra{U_B} ],
		\end{align}
		where $\dket{U_A}\dbra{U_A} $ and $\dket{U_B}\dbra{U_B} $ are causally ordered process matrices from $A$ to $B$ and from $B$ to $A$ respectively.
	}

	\Add{
		By definition, if a bipartite process matrix $W\in \mathcal{L}(A_I\otimes A_O \otimes B_I\otimes B_O )$  without global past and future is purifiable, there exists a pure process $W_{\mathrm{pure}}$ such that
		\begin{align}
			W = & W^\prime_{\mathrm{pure}} * (\ketbra{0}{0}^P \otimes \1^F)                                                                       \\
			=   & \ketbra{0}{0}^P * \Tr_{F} \left[ W^\prime_{\mathrm{pure}}\right]                                                                \\
			=   & \ketbra{0}{0}^P * \Tr_{F} \left[ \dket{U_A}\dbra{U_A} \right]  +  \ketbra{0}{0}^P * \Tr_{F} \left[ \dket{U_B}\dbra{U_B} \right].
		\end{align}
		Since $\dket{U_A}\dbra{U_A} $ is an ordered process matrix from $A$ to $B$, the operator $\ketbra{0}{0}^P * \Tr_{F} \left[ \dket{U_A}\dbra{U_A} \right] $ is proportional to an ordered process matrix without global past and global future, and an analogous argument can be used for $\dket{U_B}\dbra{U_B} $. This ensures that $W$ is a probabilistic mixture of differently causally ordered processes, hence, causally separable.
	}
\end{proof}
\Add{
	The same result as Cor.\,\ref{cor:purifiableImplyCausalsep} is independently obtained in Ref.\,\cite{barrett20} by Barrett \etal They present it as a main result and use the same fact as Thm.\,\ref{thm:newgoal} in their proof. See ``\nameref{sec:NoteAdded}'' in detail.
}

\Add{
	When considering bipartite process matrices with a non-trivial global past, the definition of causal separability does not resume to a simple probabilistic mixture of ordered process matrices but have extra nuances \cite{wechs18} which will not be covered in this manuscript. We note however that with the aid of Proposition B7 of Ref.\,\cite{wechs18}, Cor.\,\ref{cor:purifiableImplyCausalsep} can be extended to ``All purifiable bipartite process matrices without global future are causally separable.''.
}

\ReviseFurther{The next implication considers non-causal process matrices, which are process matrices which lead to non-causal sets of probability distributions \cite{oreshkov11,araujo15}. In a nutshell, let $p(ab|xy)$ be the probability of Alice and Bob respectively obtain outcomes $a$, $b$ when choosing an input configuration $x$, $y$. A set of probabilities $p^{A\prec B}(ab|xy)$ respects the causal order $A$ before $B$ if $\sum_{b}p^{A\prec B}(ab|xy)$ does not depend on $y$. Analogously, a set of probabilities $p^{B\prec A}(ab|xy)$ respects the causal order $B$ before $A$ if $\sum_{a}p^{B\prec A}(ab|xy)$ does not depend on $x$. In this way, a set of probabilities $p_\text{causal}(ab|xy)$ is causal if it can be written as a convex combination $p_\text{causal}(ab|xy)=qp^{A\prec B}(ab|xy)+(1-q)p^{B\prec A}(ab|xy)$ with $0\leq q\leq1$. It can also be shown that a set of probabilities is non-causal if and only if it violates a causal inequality \cite{oreshkov11,branciard2015simplest,araujo15}.}

\ReviseFurther{From this perspective, a bipartite process matrix $W\in \mathcal{L}(A_I\otimes A_O \otimes B_I\otimes B_O)$  is  \textit{not} causal if there \ReviseFurtherM{exist} a set of \ReviseFurtherM{the Choi representations} of \ReviseFurtherM{quantum} instruments%
\footnote{\ReviseFurther{A set of positive semidefinite operators $\{A_a \}_a $, $A_a \in \mathcal{L} (A_I \otimes A_O)$ is \ReviseFurtherM{the Choi representation} of a quantum instrument if $\Tr_{A_O}\left[\sum_a A_a\right]=\1_{A_I}$. In this notation, $a$ stands for the label of the output of the instrument element $A_a$. When we write a set of instruments $\{A_{a|x}\}_{a,x}$, $x$ is a label for \ReviseFurtherM{an} instrument $\{A_{a|x}\}_a$. 
}}
$\{A_{a|x}\}_{a,x}$, $A_{a|x}\in \mathcal{L} (A_I \otimes A_O)$ \ReviseFurtherM{and} a set of \ReviseFurtherM{the Choi representations} of \ReviseFurtherM{quantum} instruments $\{B_{b|y}\}_{b,y}$, $B_{b|y}\in \mathcal{L} (B_I \otimes B_O)$ such that the probabilities 
$p(ab|xy)= \ReviseFurtherM{\Tr [W^T (A_{a|x}\otimes B_{b|y})]}$ are \ReviseFurtherM{not} causal. When considering process matrices with \ReviseFurtherM{global} past and \ReviseFurtherM{global} future, a bipartite process matrix $W\in \mathcal{L}(P\otimes A_I\otimes A_O \otimes B_I\otimes B_O \otimes F)$ is not causal if there \ReviseFurtherM{exist} a quantum state $\rho \in \mathcal{L}(P)$, a set of \ReviseFurtherM{quantum} instruments $\{A_{a|x}\}_{a,x}$, $A_{a|x}\in \mathcal{L} (A_I \otimes A_O)$, a set of \ReviseFurtherM{quantum} instruments $\{B_{b|y}\}_{b,y}$, $B_{b|y}\in \mathcal{L} (B_I \otimes B_O)$, and a set of \ReviseFurtherM{quantum} measurements\footnote{\ReviseFurther{A set of positive semidefinite operators $\{C_c\}_c$, $C_c \in  \mathcal{L} (F) $ is a quantum measurement if $\sum_c C_c = \1_F$.}}
$\{C_{c|z}\}_{c,z}$, $\ReviseFurtherM{C_{c|z}} \in \mathcal{L}(F)$ such that the probabilities $p(abc|xyz)=\ReviseFurtherM{ \Tr [W^T (\rho \otimes A_{a|x}\otimes B_{b|y}\otimes C_{c|z})]}$ \ReviseFurtherM{are not} causal\footnote{\ReviseFurther{More precisely, the set of probabilities cannot be written as the convex combination $p(abc|xyz)=qp^{A\prec B\prec C}(abc|xyz) + (1-q)p^{B\prec A\prec C}(abc|xyz)$ where $\sum_{c}p^{A\prec B\prec C}(abc|xyz)$ and $\sum_{c}p^{B\prec A\prec C}(abc|xyz)$ do not depend on $z$, $\sum_{bc}p^{A\prec B\prec C}(abc|xyz)$ does not depend on $y$ and $z$, and $\sum_{ac}p^{B\prec A\prec C}(abc|xyz)$ does not depend on $x$ and $z$.}}.}%

\Add{
	A process matrix is \textit{extensibly causal}\,\cite{oreshkov2016causal} if it cannot violate causal inequalities even when the parties have access to additional arbitrary entangled input state. The set of extensibly causal process matrix is known to be strictly smaller than the set of device-independent causal ones \cite{feix2016causally}.
	Reference\,\cite{oreshkov2016causal} shows that the quantum switch is extensibly causal. Here we expand this result by showing that every purifiable bipartite process matrix is extensibly causal. Note that in Ref.\,\cite{araujo16} the authors write ``In fact, we haven't found a purification for any bipartite process that was able to violate a causal inequality.'', here we show that one can never find a purification for bipartite process matrices which violate a causal inequality.
}

\begin{coro}\label{cor:DI}
	\Add{
		All purifiable bipartite process matrices (including the ones with a non-trivial global past and non-trivial global future) are extensibly causal.
	}
\end{coro}

\begin{proof}
\ReviseFurther{
We start the theorem pointing that if a bipartite process matrix $W$ is purifiable, $W'=W\otimes \rho_{A_I'B_I'}$ is also purifiable for any quantum state%
\footnote{\ReviseFurther{To see this, first note that the quantum state $\rho_{A_I'B_I'}$ is a process where the dimensions $A_O$ and $B_O$ are set to one. Now let  $W_\text{pure}$  be a purification of $W$ and \ReviseFurtherM{$\rho_\text{pure}$} be a process purification of the process represented by the state $\rho_{A_I'B_I'}$. 
Straightforward calculation shows that  $W_\text{pure}\otimes \rho_\text{pure}$  is a valid purification for $W'=W\otimes \rho_{A_I'B_I'}$}} %
 $\rho_{A_I'B_I'}$.} 
	\ReviseFurther{If $W'_\text{pure}$ is the purification of $W'$, Thm.\,\ref{thm:newgoal} ensures that 
$W'_\text{pure}$ is a direct sum of pure processes. The same argument presented in the proof of Cor.\,\ref{cor:purifiableImplyCausalsep} \ReviseFurtherM{ensures} that the marginal process $\Tr_F[W'_\text{pure}]$ is causally separable. Theorem 4 of Ref.\,\cite{araujo15} states that if the marginal process $\Tr_F[W'_\text{pure}]$ of a bipartite process matrix $W'$ is causally separable, $W'$ is necessarily causal. Now, since $W'=W\otimes \rho_{A_I'B_I'}$ is causal for any quantum state $ \rho_{A_I'B_I'}$, the process $W$ is extensibly causal by definition.}
\end{proof}

\Add{
	If a process matrix is not causal, one can argue that its causal nonseparability property may be revealed in a \textit{device-independent} manner, since the causal property only depends on the probabilities $p(abc|xyz)$ and does not require the instruments involved to be characterized.
	From this perspective, when certifying non causal separability in a device-dependent manner (say, by a causal nonseparability witness), all instruments involved are assumed to be characterized, and when certifying non causal separability in a device-independent manner (say, by a causal inequality), none of the instruments involved are assumed to be characterized.
	The \textit{semi-device-independent} scenario  \cite{bavaresco19} is the intermediary case where some of the instruments performed are characterized and some instruments are not characterized.
	Reference  \cite{bavaresco19} shows that, although the quantum switch is extensibly causal in a device-independent scenario, its indefinite causally order property can be certified in several semi-device independent cases.
}

\Add{
	We now prove that all purifiable bipartite process matrices are extensibly causal in a semi-device-independent scenario where the instruments performed by Alice and Bob are not characterized but the measurement performed in the global future may be characterized. Note that this claim is stronger than Cor.\,\ref{cor:DI}, which states that all purifiable \ReviseFurther{bipartite} process matrices are (device-independent) extensibly causal.
}

\begin{coro}
	\Add{
		All purifiable bipartite process matrices are semi-device-independent extensibly causal in a scenario where the instruments performed by Alice and Bob are uncharacterized but the measurement performed in the global future may be characterized.
	}
\end{coro}

\begin{proof}
	\Add{
		The proof for the semi-device-independent case for the scenario where measurements performed in the future part are trusted follows the same steps used when proving Cor.\,\ref{cor:DI}. Since the partial trace over the total global future of $W'$ leads to a causally separable process matrix, Theorem 7 of Ref.\,\cite{bavaresco19} ensures that $W$ is extensibly causal \ReviseFurther{even if the measurements performed in the global future are characterized}.
	}
\end{proof}


\subsection{Properties of pure superchannels}


In this subsection, we analyze general properties of pure superchannels with two slots and discuss a few examples.
We start by presenting necessary conditions on the dimensions of quantum systems for a two-slot superchannel to be pure.
\begin{coro}
	Let $U : P \otimes A_O \otimes B_O \to A_I \otimes B_I \otimes F$ be the unitary operator representing a direct sum of pure combs. Suppose $\dim  A_I = \dim  A_O =: d_A$, $\dim  B_I = \dim  B_O =: d_B$ and $\dim  P = \dim  F =: D$. Then, both $d_A$ and $d_B$ have to be a divisor of $D$.
\end{coro}
This corollary can be proved from Thm.\,\ref{thm:PureCombDecomposition}. Namely, if $d_A = d_B = d$, then $D = k d$ should hold for some integer $k \geq 1$ due to Thm.\,\ref{thm:PureCombDecomposition}.
Next, we look at simple examples of direct sums of pure combs satisfying the dimensional conditions $\dim  A_I = \dim  A_O = \dim  B_I = \dim  B_O$ and $\dim  P = \dim  F$.

\begin{coro}
	\label{cor:dddanddd2d}
	Let $U : P \otimes A_O \otimes B_O \to A_I \otimes B_I \otimes F$ be the unitary operator representing a direct sum of pure combs
	and suppose $\dim  A_I = \dim  A_O = \dim  B_I = \dim  B_O = d$ and $\dim  P = \dim  F = D$. \\
	(1) If $D = d$, then $U$ \Add{represents a quantum comb} and $U$ can be decomposed into
	\begin{equation}
		U^{P A_O B_O \to A_I B_I F} = U_0^{P \to A_I} \otimes U_1^{A_O \to B_I} \otimes U_2^{B_O \to F}
	\end{equation}
	for unitary operators $U_0 : P \to A_I$, $U_1 :A_O \to B_I$, $U_2 : B_O \to F$ or $U$ can be decomposed into
	\begin{equation}
		U^{P A_O B_O \to A_I B_I F} = V_0^{P \to B_I} \otimes V_1^{B_O \to A_I} \otimes V_2^{A_O \to F}
	\end{equation}
	for unitary operators $V_0 : P \to B_I$, $V_1 :B_O \to A_I$, $V_2 : A_O \to F$. \\
	(2) If $D = 2d$, then $U$ \Add{represents a quantum comb} or $U$ is called switch-like, that is, $U$ is decomposed into
	\Add{
		\begin{equation}
			\label{eq:QSwitchlike}
			U = \ket{0}^{F_c} \bra{0}^{P_c} \otimes U^{A \prec B}	+ \ket{1}^{F_c} \bra{1}^{P_c} \otimes V^{B \prec A},
		\end{equation}
		for a unitary operator $U^{A \prec B} : P^t \otimes A_O \otimes B_O \to F^t \otimes A_I \otimes B_I$ representing a quantum comb of $A \prec B$ and a unitary operator $V^{B \prec A} : P^t \otimes A_O \otimes B_O \to F^t \otimes A_I \otimes B_I$ representing a quantum comb of $B \prec A$ where $P = P_c \otimes P_t$, $F = F_c \otimes F_t$, $\dim P_c = \dim F_c = 2$.
		In this case, since $U^{A \prec B}$ and $V^{A \prec B}$ also can be decomposed like the case $D = d$, $U$ is decomposed into
		\begin{multline}
			\label{eq:QSlikeasjfdiajig88}
			U^{P A_O B_O \to A_I B_I F} = \ket{0}^{F_c} \bra{0}^{P_c} \otimes U_0^{P_t \to A_I} \otimes U_1^{A_O \to B_I} \otimes U_2^{B_O \to F_t} \\
			+ \ket{1}^{F_c} \bra{1}^{P_c} \otimes V_0^{P_t \to B_I} \otimes V_1^{B_O \to A_I} \otimes V_2^{A_O \to F_t},
		\end{multline}
		for unitary operators $U_0 : P_t \to A_I$, $U_1 :A_O \to B_I$, $U_2 : B_O \to F_t$, $V_0 : P_t \to B_I$, $V_1 :B_O \to A_I$, $V_2 : A_O \to F_t$ where $\dim P_t = \dim F_t = d$.
	}
	The Choi vector form of Eq.\,\eqref{eq:QSlikeasjfdiajig88} is given by
	\begin{multline}
		\dket{U} =
		\ket{0}^{P_c} \dket{U_0}^{P_t A_I} \dket{U_1}^{A_O B_I} \dket{U_2}^{B_O F_t} \ket{0}^{F_c} \\
		+ \ket{1}^{P_c} \dket{V_0}^{P_t B_I} \dket{V_1}^{B_O A_I} \dket{V_2}^{A_O F_t} \ket{1}^{F_c}.
	\end{multline}
\end{coro}
We refer to pure superchannels of the form of \ReviseFurther{Eq.\,\eqref{eq:QSwitchlike}} as switch-like because, similarly to the quantum switch, the linear spaces $P$ and $F$ admit the decomposition $P=P_c\otimes P_t$ and $F=F_c\otimes F_t$ where $P_c$ and $F_c$ play the role of a control system, and $P_t$ and $F_t$ play the role of a target system.

Finally, we show that not all pure superchannels are ``switch-like''. That is, we may not be able to decompose $P$ and $F$ as $P=P_c\otimes P_t$ and $F=F_c\otimes F_t$.
In Ref.\,\cite{taddei2019quantum} the authors ask whether any \ReviseFurther{non-causally} ordered two-slot superchannels with a rank-1 Choi operator\footnote{In Ref.\,\cite{taddei2019quantum} the authors define pure process as superchannels where its Choi representation is a rank-1 projector. Although this is not equivalent to the definition of pure superchannels used here, it holds true that all pure superchannels are represented by rank-1 Choi operators.} has a control subsystem of causal orders in $P$ and $F$.
While  Cor.\,\ref{cor:dddanddd2d} ensures this is true in the case of $D = 2d$, a general pure superchannel with two slots may not admit a separate control system.
For instance, when $D = 3d$, there exists a unitary operator representing a quantum superchannel with two slots which cannot be written in such a form with the control system.
For example, suppose a direct sum of pure combs $U = U^{A \prec B} \oplus U^{B \prec A}$ such that
\begin{equation}
	\label{eq:da8g9ag}
	U^{A \prec B} \left( \ket{c,t}^P \ket{a}^{A_O} \ket{b}^{B_O} \right) = \ket{t}^{A_I} \ket{c \oplus a}^{B_I} \ket{a, b}^F,
\end{equation}
\begin{equation}
	\label{eq:da8g9ag2}
	U^{B \prec A} \left( \ket{2, t}^P \ket{a}^{A_O} \ket{b}^{B_O} \right) = \ket{t}^{B_I} \ket{b}^{A_I} \ket{2, a}^F
\end{equation}
(note the order of the superscripts) for each $c = 0, 1$, $t = 0, 1$, $a = 0, 1$, $b = 0, 1$
where $\{ \ket{c, t}^P \}_{c = 0,1,\, t=0,1,2 }$ and $\{ \ket{c, t}^F \}_{c = 0,1,\, t=0,1,2 }$ are orthonormal bases in $P$ and $F$, respectively, and
\begin{align}
	P^{A \prec B} & = \text{SPAN} \left[ \ket{c, t}^P \relmiddle| c=0,1, t=0,1 \right], \\
	P^{B \prec A} & = \text{SPAN} \left[ \ket{2, t}^P \relmiddle| t=0,1 \right],        \\
	F^{A \prec B} & = \text{SPAN} \left[ \ket{c, t}^F \relmiddle| c=0,1, t=0,1 \right], \\
	F^{B \prec A} & = \text{SPAN} \left[ \ket{2, t}^F \relmiddle| t=0,1 \right].
\end{align}
In Figure \ref{fig:D3dprocess} and Figure \ref{fig:D3dprocessBtoA}, the corresponding quantum circuits for $U^{A \prec B}$ and $U^{B \prec A}$ are shown, respectively.
In this case, $U$ cannot have a separate control subsystem of causal orders in $P$ and $F$, since $U$ is not causally ordered and the ``control parameter'' $c$ of $U^{A \prec B}$ depends on the output $a$ from $A$.
Note however that we can always embed a pure superchannel in a larger space where this superchannel becomes switch-like. For instance, we can embed $U^{B \prec A} $ of Eq.\,\eqref{eq:da8g9ag2} in a larger linear space by defining the unitary $U'^{B \prec A}:= U^{B \prec A} \oplus 1$. With this embedding, we have ${ U^{A \prec B} \oplus U'^{B \prec A} = \ket{0}^{F_c} \bra{0}^{P_c} \otimes U^{A \prec B} +  \ket{1}^{F_c} \bra{1}^{P_c} \otimes U'^{B \prec A}} $.

\begin{figure}
	\centering \includegraphics[keepaspectratio, scale=0.30]{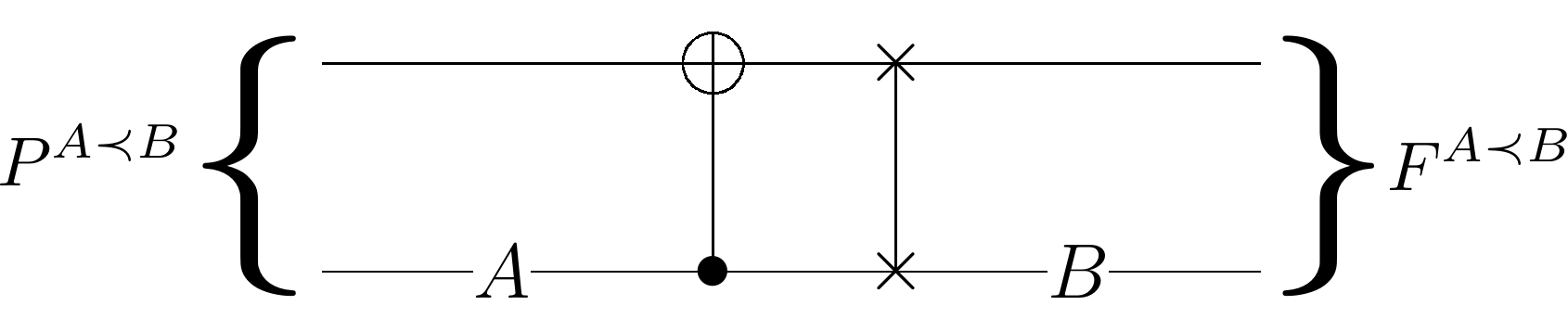}
	\caption{The circuit of $U^{A \prec B}$ in Eq.\,\eqref{eq:da8g9ag}.}
	\label{fig:D3dprocess}
\end{figure}
\begin{figure}
	\centering \includegraphics[keepaspectratio, scale=0.30]{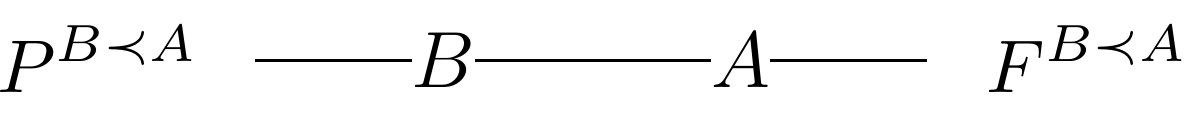}
	\caption{The circuit of $U^{B \prec A}$ in Eq.\,\eqref{eq:da8g9ag2}.}
	\label{fig:D3dprocessBtoA}
\end{figure}


\section{Conclusions} \label{sec:discussions}


Quantum superchannels with indefinite causal order have revealed fundamental properties of quantum theory and have proved useful to enhance performance in several \ReviseFurther{information-theoretic} tasks.
Nevertheless, when compared to quantum channels, we are still taking our first steps towards understanding superchannels.
For instance, despite superchannels without a definite causal order being allowed by quantum theory, it is still \ReviseFurther{unclear} if there is a fair universal procedure to \ReviseFurther{implement such objects physically}.
In addition, apart from the quantum switch and some straightforward generalizations, we still lack a simple interpretation for indefinite causal order in quantum theory.

In this paper, we have presented a simple mathematical decomposition for two-slot pure superchannels, that is, two-slot superchannels which preserve the reversibility of \ReviseFurther{input quantum operations}. This characterization allows us to understand every two-slot pure \ReviseFurther{superchannel} as a coherent superposition of ordered quantum \ReviseFurther{circuits}, \ReviseFurther{though} this does not hold true for general superchannels. Roughly speaking, our results show that all two-slot pure superchannels are similar to the quantum switch and admit a similar physical interpretation. We believe that our findings may contribute to the debate \ReviseFurther{on} interpretations and physical \ReviseFurther{realizations} of quantum superchannels.

We have also shown that purifiable two-slot superchannels, also referred \ReviseFurther{to} as purifiable bipartite processes, cannot violate causal inequalities, even if the parties share additional unlimited entanglement, and that purifiable bipartite processes with no global future are necessarily causally separable. These implications provide partial answers to the questions raised in Ref.\,\cite{araujo16}.

A clear future direction is to investigate when superchannels with more than two slots admit a simple decomposition and which methods/results can be generalized. We know however that our main results cannot be generalized in a straightforward way. There exist pure superchannels with three slots which violate causal inequalities \cite{baumeler2016space, araujo16, araujo2017quantum}, showing that pure superchannels with more than two slots are not always \ReviseFurther{direct sums} of pure combs. A more technical future direction is to investigate if condition (2) in Thm.\,\ref{thm:equivalent} can be proved without exploiting condition (3) of Thm.\,\ref{thm:equivalent}.
If that is the case, our methods may be applied to more general scenarios and superchannels with multiple slots.

\section*{Acknowledgements}
We acknowledge Alessandro Bisio and Michal Sedl\'ak for valuable discussions.
This work was supported by MEXT Quantum Leap Flagship Program (MEXT Q-LEAP) Grant Number JPMXS0118069605 \ReviseFurther{and JPMXS0120351339}, Japan Society for the Promotion of Science (JSPS) by KAKENHI grant No. 15H01677, 16F16769, 17H01694, 18H04286 and 18K13467.
\ReviseFurther{M.T. Quintino acknowledges the Austrian Science Fund (FWF) through the SFB project BeyondC (sub-project F7103), a grant from the Foundational Questions Institute (FQXi) as part of the  Quantum Information Structure of Spacetime (QISS) Project (qiss.fr). The opinions expressed in this publication are those of the authors and do not necessarily reflect the views of the John Templeton Foundation. This project has received funding from the European Union's Horizon 2020 research and innovation programme under the Marie Sk\l odowska-Curie grant agreement No 801110 and the Austrian Federal Ministry of Education, Science and Research (BMBWF). It reflects only the authors' view, the EU Agency is not responsible for any use that may be made of the information it contains.}

\section*{Note added} \label{sec:NoteAdded}
\Add{Upon completion of our work, we became aware of two independent works, Ref.\,\cite{lorenz2020causal} proving the decomposition of the reversibility preserving superchannels and Ref.\,\cite{barrett20} on one of our implications on the equivalence between causal separability and unitarily extendibility of bipartite processes.  Reference \cite{lorenz2020causal} exploits the decomposition derived in Ref.\,\cite{allen2016quantum}, which first proposes a generalization of Reichenbach's common cause principle to quantum theory and proves the existence of a decomposition for a quantum channel from input A to output B and C, such that A serves as a quantum common cause of B and C.  The proof of Ref.\,\cite{allen2016quantum} is based on Ref.\,\cite{hayden2003structure}.  The latter shows that a class of tripartite quantum \ReviseFurther{states}	 called short quantum Markov chain is the only tripartite states that saturate the strong subadditivity inequality for the von Neumann entropy.  Reference \cite{allen2016quantum} adapts the result of Ref.\,\cite{hayden2003structure} to the Choi operator of a quantum channel, despite the difference in the normalization between Choi operators and short quantum Markov chains.  Reference \cite{hayden2003structure}, on the other hand, relies on the existence of a particular mathematical property of states saturating the inequality. It invokes Ref.\,\cite{koashi2001possible}, which aims to identify \ReviseFurther{a} kind of information that can or cannot be extracted from a quantum state without changing it under the additional knowledge that it belongs to a known set of states.  In doing so, Ref.\,\cite{koashi2001possible} obtains a particular decomposition of the Hilbert space on which these states are defined.  Reference \cite{hayden2003structure} applies Ref.\,\cite{koashi2001possible} on a set of states induced from any tripartite state saturating the strong subadditivity.}
We also became aware of a related result by David Trillo, forthcoming.

\clearpage

\providecommand{\href}[2]{#2}\begingroup\raggedright\endgroup

\clearpage
\appendix


\section{Relevant properties of composite Hilbert space}
\label{secApp:notation}


In our proofs of the main theorems, we perform various operations on composite Hilbert spaces.
In Appendix \ref{secApp:linearAlgebra} and Appendix \ref{secApp:propertiesForElem}, we respectively provide notations and properties for elementary linear algebra relevant to the operations.
In Appendix \ref{secApp:NotationfVectorsincompositeHilbertspace} and Appendix \ref{secApp:Function from vectors to a subspace}, we provide notations particular to this paper.
In Appendix \ref{secApp:partsp}, we introduce the notion of ``reduced subspace'' generated from composite Hilbert spaces.
In Appendix \ref{secApp:unitaryRepresentingPureProcess} we derive properties of unitary operators representing quantum superchannels.
In Appendix \ref{secApp:BipartiteProcessMatrix}, we focus on such properties about pure superchannels with two slots.
\subsection{Notations for linear algebra}
\label{secApp:linearAlgebra}

We consider only finite-dimensional Hilbert spaces.
Let $\mathcal{X}, \mathcal{Y}$ be subspaces of a finite-dimensional Hilbert space $\mathcal{H}$ and let $\ket{x}, \ket{y}$  be vectors (which are not necessarily normalized) in $\mathcal{H}$.
\begin{itemize}
	\item $\left( \ket{x}, \ket{y} \right)$ and $\braket{x|y}$ denote the inner product of $\ket{x}$ and $\ket{y}$.
	\item $\| \ket{x} \|$ denotes the norm $\sqrt{\braket{x|x}} $ of $\ket{x}$.
	\item $\ket{x}$ and $\ket{y}$ are \textit{orthogonal} (denoted by $\ket{x} \perp \ket{y}$) if $\braket{x|y} = 0$.
	\item $\ket{x}$ and $\mathcal{X}$ are \textit{orthogonal} (denoted by $\ket{x} \perp \mathcal{X}$) if $\ket{x} \perp \ket{y}$ for all $\ket{y} \in \mathcal{X}$.
	\item $\mathcal{X}$ and $\mathcal{Y}$ are \textit{orthogonal} (denoted by $\mathcal{X} \perp \mathcal{Y}$) if $\ket{x} \perp \ket{y}$ for all $\ket{x} \in \mathcal{X}$ and for all $\ket{y} \in \mathcal{Y}$.

	\item $\mathcal{X} + \mathcal{Y}$ denotes the \textit{sum} of $\mathcal{X}$ and $\mathcal{Y}$, which is the subspace of $\mathcal{H}$ defined as follows:
	      \begin{equation}
		      \mathcal{X} + \mathcal{Y} := \left\{ \ket{x} + \ket{y} \relmiddle| \ket{x} \in \mathcal{X}, \ket{y} \in \mathcal{Y} \right\}.
	      \end{equation}
	      Since the operation $+$ is associative and commutative, this definition has a straightforward extension to the case of more than two subspaces.
	\item When the linear subspaces $\mathcal{X}$ and $\mathcal{Y}$ are orthogonal, it holds that  $\mathcal{X} + \mathcal{Y} =\mathcal{X} \oplus \mathcal{Y}$, where $\oplus$ stands for the \textit{orthogonal direct sum}\footnote{The orthogonal direct sum is also referred as internal direct sum  and it is isomorphic to the external orthogonal direct sum, which constructs a Hilbert space from two or more Hilbert spaces \cite{fletcher1973elementary, halmos2017finite}.} \cite{blyth2013further}. In order to emphasize the orthogonal relations of subspaces, when  $\mathcal{X} \perp \mathcal{Y}$ we will often write $\mathcal{X} \oplus \mathcal{Y}$ instead of $\mathcal{X} + \mathcal{Y}$.
	\item $\mathcal{X}^{\perp}$ denotes the \textit{orthogonal complement} of $\mathcal{X}$, which is the subspace of $\mathcal{H}$ defined as follows:
	      \begin{equation}
		      \mathcal{X}^{\perp} \coloneqq \left\{ \ket{v} \in \mathcal{H} \relmiddle| \ket{v} \perp \mathcal{X} \right\}.
	      \end{equation}
\end{itemize}

Let $(\mathcal{X}_\lambda)_{\lambda \in \Lambda}$ be a family of subspaces of $ \mathcal{H} $ indexed by a set $\Lambda$, which may be continuous.
\begin{itemize}
	\item $+_{\lambda \in \Lambda} \mathcal{X}_\lambda$ denotes the \textit{sum} of $(\mathcal{X}_\lambda)_{\lambda \in \Lambda}$. In this work we will also consider the case where the sum runs over continuous indices, \ie, $\Lambda$ is not countable. Note however that, since $+_{\lambda \in \Lambda} \mathcal{X}_\lambda$ is, by definition, a subset of the finite linear space $\mathcal{H}$, this sum can always be restricted to have only finite terms\footnote{In order to overcome this technicality, when $\Lambda$ has infinity elements we define $+_{\lambda \in \Lambda} \mathcal{X}_\lambda$ as follows:
	      \begin{align}
		      \bigplus_{\lambda \in \Lambda} \mathcal{X}_\lambda & := \left\{ \sum_{i=0}^{n-1} \ket{x_i} \relmiddle| n \in \mathbb{N},\  (\lambda_{i})_{i=0}^{n-1} \in \Lambda^n \ ,  \ket{x_i}\in \mathcal{X}_{\lambda_{i}}\right\},
	      \end{align}
	      where $\Lambda^n$ is the set of ordered n-tuples of elements of $\Lambda$ or the set of functions from the set $\{ 0,1,\ldots,n-1\}$ to the set $\Lambda$.}.

\end{itemize}
We have some properties about this $\bigplus$. The subset inclusion
$\mathcal{X}_{\lambda^\prime} \subset \bigplus_{\lambda \in \Lambda} \mathcal{X}_\lambda$
holds true for every  $\lambda^\prime \in \Lambda$. Also, if $\mathcal{X}_{\lambda} \subset \mathcal{Y}$ for all $\lambda \in \Lambda$, then $\bigplus_{\lambda \in \Lambda} \mathcal{X}_\lambda \subset \mathcal{Y}$.


\subsection{Useful properties}
\label{secApp:propertiesForElem}


\begin{lemm}
	\label{lem:g7g9ga7g0}
	Let $\mathcal{X}, \mathcal{Y}$ be subspaces of a finite-dimensional Hilbert space $\mathcal{H}$ and suppose $\mathcal{X}_0, \mathcal{X}_1$ are subspaces of $\mathcal{H}$ such that $\mathcal{X} = \mathcal{X}_0 \oplus \mathcal{X}_1$.
	If $\mathcal{X}_0 \subset \mathcal{Y}$ and $\mathcal{X}_1 \perp \mathcal{Y}$, then $\mathcal{X}_0 = \mathcal{X} \cap \mathcal{Y}$ and $\mathcal{X}_1 = \mathcal{X} \cap \mathcal{Y}^{\perp}$.
\end{lemm}
\begin{proof}
	Assume $\mathcal{X}_0 \subset \mathcal{Y}$ and $\mathcal{X}_1 \perp \mathcal{Y}$.
	From $\mathcal{X} \cap \mathcal{Y} \subset \mathcal{X}$ and $\mathcal{X} \cap \mathcal{Y} \perp \mathcal{X}_1$,
	$\mathcal{X} \cap \mathcal{Y} \subset \mathcal{X}_0$. $\mathcal{X}_0 \subset \mathcal{X} \cap \mathcal{Y}$ is trivial so $\mathcal{X}_0 = \mathcal{X} \cap \mathcal{Y}$.
	From $\mathcal{X} \cap \mathcal{Y}^\perp \subset \mathcal{X}$ and $\mathcal{X} \cap \mathcal{Y}^\perp \perp \mathcal{X}_0$,
	$\mathcal{X} \cap \mathcal{Y}^\perp \subset \mathcal{X}_1$. $\mathcal{X}_1 \subset \mathcal{X} \cap \mathcal{Y}^\perp$ is trivial so $\mathcal{X}_1 = \mathcal{X} \cap \mathcal{Y}^\perp$.
\end{proof}

\begin{lemm} \label{lem:UperpV->Uop(U^perpcapV^perp)opV}
	Let $\mathcal{X}$, $\mathcal{Y}$ be subspaces of a finite-dimensional Hilbert space $\mathcal{H}$.
	If $\mathcal{X} \perp \mathcal{Y}$,
	then $\mathcal{H} = \mathcal{X} \oplus (\mathcal{X}^{\perp} \cap \mathcal{Y}^{\perp}) \oplus \mathcal{Y}$.
\end{lemm}
\begin{proof} \
	Let $\ket{x^\perp} \in \mathcal{X}^{\perp}$. $\ket{x^\perp}$ can be decomposed into $\ket{x^\perp} =: \ket{y} + \ket{y^\perp}$ with some $\ket{y} \in \mathcal{Y}$ and some $\ket{y^\perp} \in \mathcal{Y}^{\perp}$.
	Since $\mathcal{Y} \subset \mathcal{X}^{\perp}$, $\ket{y} \in \mathcal{X}^{\perp}$ and then we obtain $\ket{y^\perp} = \ket{x^\perp} - \ket{y} \in \mathcal{X}^{\perp}$.
	Therefore, $\ket{x^\perp} \in (\mathcal{X}^{\perp} \cap \mathcal{Y}^{\perp}) \oplus \mathcal{Y}$, which implies $\mathcal{X}^{\perp} \subset (\mathcal{X}^{\perp} \cap \mathcal{Y}^{\perp}) \oplus \mathcal{Y}$.
	Clearly $\mathcal{X}^{\perp} \supset (\mathcal{X}^{\perp} \cap \mathcal{Y}^{\perp}) \oplus \mathcal{Y}$ and then
	we obtain $\mathcal{H} = \mathcal{X} \oplus \mathcal{X}^{\perp} = \mathcal{X} \oplus (\mathcal{X}^{\perp} \cap \mathcal{Y}^{\perp}) \oplus \mathcal{Y}$.
\end{proof}

\begin{lemm}
	\label{lem:sesquiperp}
	Let $f(\ket{x},\, \ket{x^\prime})$ be a sesquilinear form on a Hilbert space $\mathcal{H}$.
	If for all $\ket{x},\, \ket{x^\perp} \in \mathcal{H}$ such that $\ket{x} \perp \ket{x^\perp}$,
	\begin{equation}
		\label{eq:xxperp}
		f(\ket{x},\, \ket{x^\perp}) = 0,
	\end{equation}
	then for all $\ket{x_0},\, \ket{x_1} \in \mathcal{H}$ such that $\| \ket{x_0} \| = \| \ket{x_1} \|$,
	\begin{equation}
		f(\ket{x_0},\, \ket{x_0}) = f(\ket{x_1},\, \ket{x_1}).
	\end{equation}
\end{lemm}
\begin{proof}
	When $\dim \mathcal{H} = 1$, it is trivial. In the following, we assume $\dim \mathcal{H} \geq 2$.
	First, we show the case of $\ket{x_0} \perp \ket{x_1}$. Take the computational basis in $\mathcal{H}$ and assume $\ket{x_0} = \ket{0}$ and $\ket{x_1} = \ket{1}$ without loss of generality.
	In Eq.\,\eqref{eq:xxperp}, substitute $\ket{0} + \ket{1}$ for $\ket{x}$ and $\ket{0} - \ket{1}$ for $\ket{x^\perp}$
	and then using Eq.\,\eqref{eq:xxperp} we obtain
	\begin{equation}
		f(\ket{0},\, \ket{0}) - f(\ket{0},\, \ket{1}) + f(\ket{1},\, \ket{0}) - f(\ket{1},\, \ket{1}) = 0.
	\end{equation}
	\begin{equation}
		\label{eq:x0x0=x1x1}
		\therefore f(\ket{0},\, \ket{0}) = f(\ket{1},\, \ket{1}). \quad (\because \mathrm{Eq.\,\eqref{eq:xxperp}})
	\end{equation}
	Next, we show the general case where $\ket{x_0}$ and $\ket{x_1}$ may not be orthogonal. Without loss of generality, assume $\ket{x_0} = \ket{0}$ and $\ket{x_1} = a \ket{0} + b \ket{1}$ where $a$, $b$ are arbitrary complex numbers such that $|a|^2 + |b|^2 = 1$.
	Then,
	\begin{alignat}{2}
		  &   &       & f(\ket{x_1},\, \ket{x_1}) \notag                                                                                       \\
		  & = & |a|^2 & f(\ket{0},\, \ket{0}) + a^* b f(\ket{0},\, \ket{1}) + b^* a f(\ket{1},\, \ket{0}) + |b|^2 f(\ket{1},\, \ket{1}) \notag \\
		  & = & |a|^2 & f(\ket{0},\, \ket{0}) + |b|^2 f(\ket{1},\, \ket{1}) \quad (\because \text{Eq.\,\eqref{eq:xxperp}}) \notag              \\
		  & = & |a|^2 & f(\ket{0},\, \ket{0}) + |b|^2 f(\ket{0},\, \ket{0}) \quad (\because \text{Eq.\,\eqref{eq:x0x0=x1x1}}) \notag           \\
		  & = &       & f(\ket{0},\, \ket{0}).
	\end{alignat}
\end{proof}


\subsection{Notation of multiple-ket}
\label{secApp:NotationfVectorsincompositeHilbertspace}


$\ket{\phi}^{\mathcal{X}}$ denotes a vector in a Hilbert space $\mathcal{X}$. The superscript $\mathcal{X}$ of $\ket{\phi}^{\mathcal{X}}$ represents the space where it belongs.
$\dket{\Phi}^{\mathcal{X} \mathcal{Y}}$ denotes a vector in the tensor product of Hilbert spaces $\mathcal{X}$ and $\mathcal{Y}$. We use $\dket{\Phi}^{\mathcal{X} \mathcal{Y}}$, not with a single-ket, to emphasize the fact that it belongs to the tensor product of two Hilbert spaces.
Similarly, we use $\tket{\eta}^{\mathcal{X} \mathcal{Y} \mathcal{Z}}$, not with a single-ket, to emphasize on the fact that it belongs to the tensor product of {three} Hilbert spaces.
Please note that $\Phi$ in $\dket{\Phi}^{\mathcal{X} \mathcal{Y}}$ is just a label and $\dket{\Phi}^{\mathcal{X} \mathcal{Y}}$ is not a Choi vector of some operator $\Phi$, that is, $\dket{\Phi}^{\mathcal{X} \mathcal{Y}} \neq I \otimes \Phi \dket{I}$.


\subsection{Function from subspaces to a subspace}
\label{secApp:Function from vectors to a subspace}


In this paper, we deal with functions mapping linear subspaces to linear subspaces where some particular spaces are often used. We define two useful subspaces.

Let $\ket{\alpha}^{A_O} \in A_O$, we  define the one-dimensional subspace spanned by $\ket{\alpha}$ as
\begin{equation}
	\alpha:=\vspan{\ket{\alpha}^{A_O}}
\end{equation}
and its orthogonal complement as
\begin{equation}
	\overline{\alpha}:=\vspan{\ket{\alpha}^{A_O}}^{\perp}.
\end{equation}
Analogously, if $\ket{\beta}^{B_O} \in B_O$, we define
\begin{equation}
	\beta:=\vspan{\ket{\beta}^{B_O}}, \quad \overline{\beta}:=\vspan{\ket{\beta}^{B_O}}^{\perp}.
\end{equation}

Let $\lnsp{V}$ be a function which takes a pair of linear space as arguments such that $\lnsp{V} \left(A_O, B_O \right)$ is a linear space.
For simplicity of notation, we will often write the \ReviseFurther{arguments} of such functions by sub-indexes, for example,
$\lnsp{V}_{A_O B_O}$ for $\lnsp{V} \left(A_O, B_O \right)$,
$\lnsp{V}_{A_O \beta} $ for $ \lnsp{V} (A_O, \beta)$,
$\lnsp{V}_{\alpha \overline{\beta}}$ for
$ \lnsp{V} \left(\alpha ,\overline{\beta}\right)$, \textit{etc}.

We call $\lnsp{V}$ bi-additive if
$\lnsp{V}(A_{0} + A_{1}, B_{0} + B_{1}) = \lnsp{V}(A_{0}, B_{0}) + \lnsp{V}(A_{0}, B_{1}) +
	\lnsp{V}(A_{1}, B_{0}) + \lnsp{V}(A_{1}, B_{1})$ for all subspaces $A_{0}, A_{1}$ in $A_O$ and $B_{0}, B_{1}$ in $B_O$.


\subsection{\Partsp}
\label{secApp:partsp}


Let $\mathcal{E}, \mathcal{F}$ be two finite-dimensional Hilbert spaces and $\mathcal{W}$ be a subspace of $\mathcal{E} \otimes \mathcal{F}$. Define the following notation $\redspn{\bullet}{\mathcal{E}}{\mathcal{F}}$ which \ReviseFurther{generates} some subspace of $\mathcal{F}$:
\begin{equation}
	\redspn{\mathcal{W}}{\mathcal{E}}{\mathcal{F}} := \vspan{ {}^{\mathcal{E}}\bradket{\epsilon | \eta}^{\mathcal{EF}} \relmiddle|
	\dket{\eta}^{\mathcal{EF}} \in \mathcal{W}, \ket{\epsilon}^{\mathcal{E}} \in \mathcal{E} }.
\end{equation}
For example, \ReviseFurther{if}
\begin{equation}
	\mathcal{W} = \vspan{\ket{0}^{\mathcal{E}} \ket{0}^{\mathcal{F}} + 2 \ket{1}^{\mathcal{E}} \ket{1}^{\mathcal{F}}},
\end{equation}
then
\begin{equation}
	\redspn{\mathcal{W}}{\mathcal{E}}{\mathcal{F}} = \vspan{\ket{0}^{\mathcal{F}},\, \ket{1}^{\mathcal{F}}}.
\end{equation}
$\redspn{\mathcal{W}}{\mathcal{E}}{\mathcal{F}}$ consists of ``$\mathcal{F}$ parts'' of $c$ in some sense. Note that the subspace $\redspn{\mathcal{W}}{\mathcal{E}}{\mathcal{F}}$ is equal to
the sum of the supports of the reduced density operators $\rho_{\eta}^{\mathcal{F}} := \Tr_{\mathcal{E}} \dket{\eta}\dbra{\eta}$ for system $\mathcal{F}$ corresponding to each state $\dket{\eta}^{\mathcal{EF}} \in \ReviseFurther{\mathcal{W}}$:
\begin{equation}
	\redspn{\mathcal{W}}{\mathcal{E}}{\mathcal{F}} = \bigplus_{\dket{\eta}^{\mathcal{EF}} \in \mathcal{W}} \suppa \left( \rho_{\eta}^{\mathcal{F}} \right).
\end{equation}
It is also interesting to remark that  the simple union $\bigcup_{\dket{\eta}^{\mathcal{EF}} \in \mathcal{E} \otimes \mathcal{F}} \suppa \left( \rho_{\eta}^{\mathcal{F}} \right)$ is not a linear space in general. One of the counter examples is in the case $\mathcal{W} = \vspan{\ket{0}^{\mathcal{E}} \ket{0}^{\mathcal{F}}, \ket{0}^{\mathcal{E}} \ket{1}^{\mathcal{F}} + \ket{1}^{\mathcal{E}} \ket{2}^{\mathcal{F}}}$.
Since this operation $\redspn{\bullet}{\mathcal{E}}{\mathcal{F}}$, which transforms \ReviseFurther{subspaces of product spaces} into \ReviseFurther{subspaces of component spaces}, is like the partial trace, which transforms density operators \ReviseFurther{in product spaces} into reduced density operators \ReviseFurther{in component spaces}, we give the name ``\textit{\partsp}'' to its product $\redspn{\mathcal{W}}{\mathcal{E}}{\mathcal{F}}$.

Some useful properties are following where $\lnsp{W}$, $\lnsp{W}^\prime$ are subspaces of $\mathcal{E} \otimes \mathcal{F}$ and $\mathcal{X}$ is a subspace of $\mathcal{F}$:
\begin{itemize}
	\item $\redsp{\mathcal{E} \otimes \mathcal{X}}{\mathcal{E}}{\mathcal{F}} = \mathcal{X}$.
	\item $\mathcal{W} \subset \mathcal{E} \otimes \redspn{\mathcal{W}}{\mathcal{E}}{\mathcal{F}}$.
	\item If $\lnsp{W} \subset \lnsp{W}^\prime$, then $\redspn{\lnsp{W}}{\mathcal{E}}{\mathcal{F}} \subset \redspn{\lnsp{W}^\prime}{\mathcal{E}}{\mathcal{F}}$.
	\item $\redspn{\mathcal{W}}{\mathcal{E}}{\mathcal{F}} \subset \mathcal{X}$ if and only if $\mathcal{W} \subset \mathcal{E} \otimes \mathcal{X}$ .
	\item If $\redspn{\lnsp{W}}{\mathcal{E}}{\mathcal{F}} \perp \redspn{\lnsp{W}^\prime}{\mathcal{E}}{\mathcal{F}}$, then $\lnsp{W} \perp \lnsp{W}^\prime$.
	\item $\redspn{\mathcal{W}}{\mathcal{E}}{\mathcal{F}} \perp \mathcal{X}$ if and only if $\mathcal{W} \perp \mathcal{E} \otimes \mathcal{X}$.
	\item $\redsp{\lnsp{W} + \lnsp{W}^\prime}{\mathcal{E}}{\mathcal{F}} = \redspn{\lnsp{W}}{\mathcal{E}}{\mathcal{F}} + \redspn{\lnsp{W}^\prime}{\mathcal{E}}{\mathcal{F}}$.
\end{itemize}

\begin{lemm}
	\label{lem:UperpV0FeqEotUsubV1}
	Let $\mathcal{E}$, $\mathcal{F}$ be two finite-dimensional Hilbert spaces and $\lnsp{W}_0$, $\lnsp{W}_1$ be two subspaces of $\mathcal{E} \otimes \mathcal{F}$ such that $\mathcal{E} \otimes \mathcal{F} = \lnsp{W}_0 \oplus \lnsp{W}_1$.
	Suppose $\mathcal{X}$ is a subspace of $\mathcal{F}$. Then,
	$\mathcal{X} \perp \redspn{\lnsp{W}_1}{\mathcal{E}}{\mathcal{F}}$ if and only if $\mathcal{E} \otimes \mathcal{X} \subset \lnsp{W}_0$.
\end{lemm}

\begin{proof}
	$\mathcal{X} \perp \redspn{\lnsp{W}_1}{\mathcal{E}}{\mathcal{F}}\iff \mathcal{E} \otimes \mathcal{X} \perp \lnsp{W}_1 \iff \mathcal{E} \otimes \mathcal{X} \subset \lnsp{W}_0$.
\end{proof}

\begin{coro}
	\label{cor:F0perpF1=>V0=EtensorF0}
	Let $\mathcal{E}$, $\mathcal{F}$ be two finite-dimensional Hilbert spaces and $\lnsp{W}_0$, $\lnsp{W}_1$ be two subspaces of $\mathcal{E} \otimes \mathcal{F}$ such that $\mathcal{E} \otimes \mathcal{F} = \lnsp{W}_0 \oplus \lnsp{W}_1$.
	$\redspn{\lnsp{W}_0}{\mathcal{E}}{\mathcal{F}} \perp \redspn{\lnsp{W}_1}{\mathcal{E}}{\mathcal{F}}$ if and only if $\mathcal{E} \otimes \redspn{\lnsp{W}_0}{\mathcal{E}}{\mathcal{F}} = \lnsp{W}_0$.
\end{coro}

\begin{lemm}
	\label{lem:5-1-2-1}
	Let $\mathcal{E}$, $\mathcal{F}$ be two finite-dimensional Hilbert spaces and $\lnsp{W}_0$, $\lnsp{W}_1$ be subspaces of $\mathcal{E} \otimes \mathcal{F}$ such that $\mathcal{E} \otimes \mathcal{F} = \lnsp{W}_0 \oplus \lnsp{W}_1$. Then
	$\mathcal{F} = (\redspn{\lnsp{W}_1}{\mathcal{E}}{\mathcal{F}})^{\perp}
		\oplus (\redspn{\lnsp{W}_0}{\mathcal{E}}{\mathcal{F}} \cap \redspn{\lnsp{W}_1}{\mathcal{E}}{\mathcal{F}})
		\oplus (\redspn{\lnsp{W}_0}{\mathcal{E}}{\mathcal{F}})^{\perp}$ holds.
\end{lemm}

\begin{proof} \
	First, we show $(\redspn{\lnsp{W}_0}{\mathcal{E}}{\mathcal{F}})^{\perp} \perp (\redspn{\lnsp{W}_1}{\mathcal{E}}{\mathcal{F}})^{\perp}$.
	From $\lnsp{W}_0 \oplus \lnsp{W}_1 = \mathcal{E} \otimes \mathcal{F}$,
	$(\redspn{\lnsp{W}_1}{\mathcal{E}}{\mathcal{F}})^{\perp} \perp \redspn{\lnsp{W}_1}{\mathcal{E}}{\mathcal{F}}$ and
	Lem.\,\ref{lem:UperpV0FeqEotUsubV1},
	we obtain $\mathcal{E} \otimes (\redspn{\lnsp{W}_1}{\mathcal{E}}{\mathcal{F}})^{\perp} \subset \lnsp{W}_0$.
	Similarly, we can obtain $\mathcal{E} \otimes (\redspn{\lnsp{W}_0}{\mathcal{E}}{\mathcal{F}})^{\perp} \subset \lnsp{W}_1$.
	Thus, since $\lnsp{W}_0 \perp \lnsp{W}_1$, we obtain
	$\mathcal{E} \otimes (\redspn{\lnsp{W}_1}{\mathcal{E}}{\mathcal{F}})^{\perp} \perp \mathcal{E} \otimes (\redspn{\lnsp{W}_0}{\mathcal{E}}{\mathcal{F}})^{\perp}$,
	which implies $(\redspn{\lnsp{W}_1}{\mathcal{E}}{\mathcal{F}})^{\perp} \perp (\redspn{\lnsp{W}_0}{\mathcal{E}}{\mathcal{F}})^{\perp}$.
	Therefore, from Lem.\,\ref{lem:UperpV->Uop(U^perpcapV^perp)opV}, we obtain
	\begin{equation}
		\mathcal{F} = (\redspn{\lnsp{W}_1}{\mathcal{E}}{\mathcal{F}})^{\perp} \oplus (\redspn{\lnsp{W}_0}{\mathcal{E}}{\mathcal{F}} \cap \redspn{\lnsp{W}_1}{\mathcal{E}}{\mathcal{F}})
		\oplus (\redspn{\lnsp{W}_0}{\mathcal{E}}{\mathcal{F}})^{\perp}.
	\end{equation}
\end{proof}

\begin{coro}
	\label{cor:EF=V0opV1->F0=V1perpFop(F0capF1)}
	Let $\mathcal{E}$, $\mathcal{F}$ be two finite-dimensional Hilbert spaces and $\lnsp{W}_0$, $\lnsp{W}_1$ be subspaces of $\mathcal{E} \otimes \mathcal{F}$ such that $\mathcal{E} \otimes \mathcal{F} = \lnsp{W}_0 \oplus \lnsp{W}_1$.
	$\redspn{\lnsp{W}_0}{\mathcal{E}}{\mathcal{F}}$ is decomposed into
	\begin{equation}
		\redspn{\lnsp{W}_0}{\mathcal{E}}{\mathcal{F}} = \mathcal{F}_0^{\mathrm{share}} \oplus \mathcal{F}_0^{\mathrm{perp}}
	\end{equation}
	by using the subspaces $\mathcal{F}_0^{\mathrm{share}}$, $\mathcal{F}_0^{\mathrm{perp}}$ of 	$\redspn{\lnsp{W}_0}{\mathcal{E}}{\mathcal{F}}$ satisfying
	\begin{equation}
		\mathcal{F}_0^{\mathrm{share}} \subset \redspn{\lnsp{W}_1}{\mathcal{E}}{\mathcal{F}},
	\end{equation}
	\begin{equation}
		\mathcal{F}_0^{\mathrm{perp}} \perp \redspn{\lnsp{W}_1}{\mathcal{E}}{\mathcal{F}}.
	\end{equation}
	Note that $\mathcal{F}_0^{\mathrm{share}}, \mathcal{F}_0^{\mathrm{perp}}$ are uniquely determined from Lem.\,\ref{lem:g7g9ga7g0}.
\end{coro}

Let $\mathcal{A}$, $\mathcal{B}$ be two finite-dimensional Hilbert spaces and $\lnsp{W}$ be a subspace of $\mathcal{A} \otimes \mathcal{B} \otimes \mathcal{F}$.
\begin{itemize}
	\item $\redspn{\lnsp{W}}{\mathcal{A} \otimes \mathcal{B}}{\mathcal{F}} = \redspn{\lnsp{W}}{\mathcal{AB}}{\mathcal{F}} = \redspn{\lnsp{W}}{\mathcal{BA}}{\mathcal{F}}$.
	\item $\redsp{B \otimes \lnsp{W}}{\mathcal{A}}{\mathcal{BF}} = B \otimes \redspn{\lnsp{W}}{\mathcal{A}}{\mathcal{F}}$.
\end{itemize}

\subsection{Unitary operators representing pure superchannels}
\label{secApp:unitaryRepresentingPureProcess}
In this section, we construct from Def.\,\ref{def:ProcessMatrix} and Def.\,\ref{def:ChoiOpOfQuantumComb} the conditions that unitary operators represent a \ReviseFurther{quantum} superchannel and a quantum comb, respectively.

\begin{lemm}
	\label{lem:Uisprocess<=>UGisunitary}
	Let $U : \bigotimes_{n=0}^N \mathcal{H}_{2n} \to \bigotimes_{n=0}^N \mathcal{H}_{2n+1}$ be a unitary operator.
	$U$ represents a quantum superchannel if and only if
	for all auxiliary Hilbert spaces $\mathcal{H}_{m}^\prime$, $m \in \{ 1, \cdots, 2N\}$, and unitary operators $U_n : \mathcal{H}_{2n-1} \otimes \mathcal{H}_{2n-1}^\prime \to \mathcal{H}_{2n} \otimes \mathcal{H}_{2n}^\prime$, $n \in \{ 1, \cdots, N\}$,
	$ \dket{U} \dbra{U} \ast (\bigotimes_{n=1}^N \dket{U_n} \dbra{U_n})$ is the Choi operator of a unitary channel, that is,
	\begin{equation}
		U_G := \Tr_{\mathcal{H}_{2} \mathcal{H}_{4} \cdots \mathcal{H}_{2N}} \left[ \left( \1^{\mathcal{H}_{2N+1}} \otimes \bigotimes_{n=1}^N U_n \right) \left( U \otimes \bigotimes_{n=1}^N \1^{\mathcal{H}_{2n-1}^\prime} \right) \right]
	\end{equation}
	is a unitary operator.
\end{lemm}

\begin{proof}

	Sufficiency is \Add{implied} from Thm.\,\ref{thm:W=UWUW}. To show necessity,
	take $N$ arbitrary CPTP maps $\tilde{E_n} : \mathcal{L}(\mathcal{H}_{2n-1} \otimes \mathcal{H}_{2n-1}^\prime) \to \mathcal{L}(\mathcal{H}_{2n} \otimes \mathcal{H}_{2n}^\prime)$, $n \in \{ 1, \cdots, N\}$, as the input of $W_U := \dket{U} \dbra{U}$.
	Next, extend these CPTP maps into a unitary transformation with suitable auxiliary systems and then we obtain a unitary operation with the auxiliary systems as the output of $W_U$ since $W_U$ is pure. Finally, reduce the auxiliary systems then we obtain a CPTP map as the output of $W_U$ from the arbitrary CPTP maps $\tilde{E_n}$, $n \in \{ 1, \cdots, N\}$.
\end{proof}
\begin{lemm}
	\label{lem:2mperp->2m+1perp}
	Let $U^{(N+1)} : \bigotimes_{n=0}^N \mathcal{H}_{2n} \to \bigotimes_{n=0}^N \mathcal{H}_{2n+1}$ be a unitary operator.
	$U^{(N+1)}$ represents a quantum comb
	with $N$ slots
	if and only if
	$U^{(N+1)}$ satisfies for each $n=1 \cdots N$,
	for all $\ket{\alpha}^{2n} \in \mathcal{H}_{2n}$,
	\begin{equation}
		\label{eq:2mperp->2m+1perp}
		\redspn{\lnsp{V}^{(n)}_{\alpha}}{1 3 \cdots 2n-1}{2n+1 \cdots 2N+1} \perp
		\redspn{\lnsp{V}^{(n)}_{\overline{\alpha}}}{1 3 \cdots 2n-1}{2n+1 \cdots 2N+1}
	\end{equation}
	where
	\begin{equation}
		\label{eq:defVnphi}
		\lnsp{V}^{(n)} ( \mathcal{H}_{2n,\, \mathrm{sub}} ) := U^{(N+1)} \left( \left( \bigotimes_{k \neq n} \mathcal{H}_{2k} \right) \otimes \mathcal{H}_{2n,\, \mathrm{sub}} \right)
	\end{equation}
	for a given subspace $\mathcal{H}_{2n,\, \mathrm{sub}}$ of $\mathcal{H}_{2n}$. We have used abbreviations similarly to the one introduced in Appendix \ref{secApp:Function from vectors to a subspace} as follows:
	\begin{alignat}{2}
		  & \lnsp{V}^{(n)}_{\alpha}            &   & := \lnsp{V}^{(n)} \left(\vspan{\ket{\alpha}^{2n}} \right)                         \\
		  & \lnsp{V}^{(n)}_{\overline{\alpha}} &   & := \lnsp{V}^{(n)} \left( \left( \vspan{\ket{\alpha}^{2n}} \right)^{\perp} \right)
	\end{alignat}
	(Notations $\redsp{\bullet}{m}{}$ and $\ket{\bullet}^{m}$ are shorthands respectively for $\redsp{\bullet}{\mathcal{H}_m}{}$ and $\ket{\bullet}^{\mathcal{H}_m}$.)
\end{lemm}

\begin{proof}
	\Add{We show the only-if part.}
	Let $n$ be any integer in $1...N$.
	\Add{
		For simplicity, define
		\begin{equation}
			\mathcal{H}_{\widehat{2n}} := \bigotimes_{k \neq n} \mathcal{H}_{2k}.
		\end{equation}
	}
	From \Add{Thm.\,\ref{thm:purecombcausalityforall}}, we have
	\begin{align}
		\label{eq:8dfa0g}
		  & \Tr_{2N+1 \ 2N-1 \ \cdots \ 2n+1} \dket{U^{(N+1)}} \dbra{U^{(N+1)}} = \1^{2N \ 2N-2 \ \cdots \ 2n} \otimes R^{(n)}
	\end{align}
	for some Choi operator $R^{(n)}$.
	Let $\ket{\pi}^{\widehat{2n}}, \ket{\pi^\prime}^{\widehat{2n}} \in \mathcal{H}_{\widehat{2n}}$, $\ket{\alpha^\perp}^{2n} \in \mathcal{H}_{2n}$
	such that $\ket{\alpha^\perp}^{2n} \perp \ket{\alpha}^{2n}$ and
	$\ket{\Phi}^{13\cdots2n-1}, \ket{\Phi^\prime}^{13\cdots2n-1} \in \otimes_{k = 0}^{n-1} \mathcal{H}_{2k+1}$, then
	\footnotesize
	\begin{align}
		  & \Tr_{2N+1 \ 2N-1 \ \cdots \ 2n+1} \left[ \left( {}^{\widehat{2n}}\bra{\pi^*} {}^{2n}\bra{\alpha^*} {}^{13\cdots2n-1}\bra{\Phi} \right) \dket{U^{(N+1)}} \dbra{U^{(N+1)}} \left(
			\ket{{\pi^\prime}^*}^{\widehat{2n}} \ket{{\alpha^\perp}^*}^{2n} \ket{{\Phi^\prime}}^{13\cdots2n-1} \right) \right] = 0. \notag
	\end{align}
	\small
	\begin{equation}
		\therefore \left( {}^{13\cdots2n-1}\bra{\Phi} U^{(N+1)} \left( \ket{\pi}^{\widehat{2n}} \ket{\alpha}^{2n} \right),
		{}^{13\cdots2n-1}\bra{\Phi^\prime} U^{(N+1)} \left( \ket{\pi^\prime}^{\widehat{2n}} \ket{\alpha^\perp}^{2n} \right) \right) = 0.
	\end{equation}
	\normalsize
	Therefore, we obtain
	\begin{align}
		\label{eq:8g0g0a}
		\redspn{\lnsp{V}^{(n)}_{\alpha}}{1 3 \cdots 2n-1}{} & \perp \redspn{\lnsp{V}^{(n)}_{\alpha^\perp}}{13\cdots 2n-1}{}.
	\end{align}
	Thus we have
	\begin{equation}
		\redspn{\lnsp{V}^{(n)}_{\alpha}}{1 3 \cdots 2n-1}{} \perp \redspn{\lnsp{V}^{(n)}_{\overline{\alpha}}}{13\cdots 2n-1}{}.
	\end{equation}

	\Add{
		The if part can be obtained by calculating backwards. Note that the case $n=0$ makes Eq.\,\eqref{eq:8dfa0g} hold from unitarity of $U^{(N+1)}$.
	}
\end{proof}

In Eq.\,\eqref{eq:2mperp->2m+1perp}, $\redsp{\bullet}{13\cdots2n-1}{2n+1 2n+3\cdots2N+1}$ means reducing the subsystems of $\mathcal{H}_1$, $\mathcal{H}_3$, $\cdots$, $\mathcal{H}_{2n-1}$ and
extracting the states on the subsystem of $\mathcal{H}_{2n+1}$, $\mathcal{H}_{2n+3}$, $\cdots$, $\mathcal{H}_{2N+1}$.
This lemma just rephrases the condition that a Choi operator satisfies the causality condition in terms of unitary operators.
\Add{In Appendix \ref{secApp:BipartiteProcessMatrix}, the special case of this lemma where $N = 2$ is given and we consider the interpretation of this lemma through the case.}

From this lemma, the condition that $U^{(N+1)}$ represents a quantum comb is formulated by using the reduced subspaces $\redspn{\lnsp{V}^{(n)}_{\alpha}}{13\cdots2n-1}{2n+1\cdots2N+1}$ generated from $U^{(N+1)}$.
In the following, we present our results in terms of unitary operators, not in terms of Choi matrices as the proof of our main statement (\ReviseFurther{Thm.}\,\ref{thm:newgoal}) is written in terms of unitary operators.


\subsection{Quantum combs with two slots}
\label{secApp:BipartiteProcessMatrix}


In the case of a quantum comb of $A \prec B$, Lem.\,\ref{lem:2mperp->2m+1perp} is in the following form:
\begin{coro}
	\label{cor:N=2purecomb}
	Let $U : P \otimes A_O \otimes B_O \to A_I \otimes B_I \otimes F$ be a unitary operator.
	$U$ represents a quantum comb of $A \prec B$
	if and only if
	for all $\ket{\alpha}^{A_O} \in A_O$ and $\ket{\beta}^{B_O} \in B_O$,
	\begin{equation}
		\label{eq:bperp->fperp}
		\redspn{\lnsp{V}_{A_O \beta}}{A_I B_I}{F} \perp \redspn{\lnsp{V}_{A_O \overline{\beta}} }{A_I B_I}{F},
	\end{equation}
	\begin{equation}
		\label{eq:abperp->bfperp}
		\redspn{\lnsp{V}_{\alpha B_O}}{A_I}{B_I F} \perp \redspn{\lnsp{V}_{\overline{\alpha} B_O}}{A_I}{B_I F}
	\end{equation}
	where
	\begin{equation}
		\label{eq:defVabincor:N=2purecomb}
		\lnsp{V}(A_{\mathrm{sub}}, B_{\mathrm{sub}}) := \left\{ U \left( \ket{\pi}^{P} \ket{\alpha}^{A_O} \ket{\beta}^{B_O} \right) \relmiddle|
		\ket{\pi}^{P} \in P,\, \ket{\alpha}^{A_O} \in A_{\mathrm{sub}},\, \ket{\beta}^{B_O} \in B_{\mathrm{sub}} \right\}.
	\end{equation}
	for a subspace $A_{\mathrm{sub}}$, $B_{\mathrm{sub}}$ of $A_O$, $B_O$, respectively.
\end{coro}

In Eq.\,\eqref{eq:bperp->fperp} and Eq.\,\eqref{eq:abperp->bfperp}, $\redspn{\bullet}{A_I B_I}{F}$, $\redspn{\bullet}{A_I}{B_I F}$ means reducing the subsystem of $A_I \otimes B_I$, $A_I$ and extracting the states on the subsystem of $F$, $B_I \otimes F$, respectively.
\Add{
Eq.\,\eqref{eq:bperp->fperp} means that if two orthogonal pure states are ``fed'' into $B_O$, that is, two orthogonal pure states are used as input-states in $B_O$ for $U$ and arbitrary input-states for $U$ in \ReviseFurther{$P$ and $A_O$} are used, the corresponding output-states are orthogonal on the reduced states on the subsystems of $F$.
Eq.\,\eqref{eq:abperp->bfperp} means that if two orthogonal pure states are ``fed'' into $A_O$, the corresponding output-states are orthogonal on the reduced states on the subsystems of $B_I \otimes F$.
From this corollary, the condition that $U$ represents a quantum comb with two slots is formulated by using the reduced subspaces of subspaces $\lnsp{V}_{A_{\mathrm{sub}} B_{\mathrm{sub}}}$ generated from $U$.
}

\Add{Consider why this corollary holds.
	Although Thm.\,\ref{thm:PureCombDecomposition} has not been shown here yet and is proved later in Appendix.\,\ref{secApp:ProofOfPureCombDecomposition}, we assume Thm.\,\ref{thm:PureCombDecomposition} in the case of two-slot combs in advance, \ie, assume that two-slot combs of $A \prec B$ are in the form of concatenation of unitary channels as in Fig.\,\ref{fig:f89ga89hujap}.
	First, we consider ``feeding'' orthogonal pure states into $B_O$.
	From unitarity of $U_2$, ``feeding'' orthogonal pure states into $B_O$ causes orthogonal reduced states in $F$.
	This can be simply formulated by using reduced subspaces into Eq.\,\eqref{eq:bperp->fperp}.
	Next, from unitarity of $U_1$ and $U_2$, ``feeding'' orthogonal pure states into $A_O$ causes orthogonal reduced states in $B_I \otimes F$. This corresponds to Eq.\,\eqref{eq:abperp->bfperp}.
	Of course, we have not yet proved that pure combs are concatenations of unitary channels, but we can use these properties in the proof of that.}


\section{Proof of Thm.\,\ref{thm:PureCombDecomposition}}
\label{secApp:ProofOfPureCombDecomposition}


We prove Thm.\,\ref{thm:PureCombDecomposition} inductively by showing \\
(i) there exists an integer $k_N$ such that $d_{2N} k_N = d_{2N+1}$, \label{appendix a i} \\
(ii) $U^{(N+1)}$ can be decomposed as:
\begin{equation}
	\label{eq:UR=V_(N-1)UN}
	U^{(N+1)} = U_N U^{({N})}
\end{equation}
for some unitary operators $U^{(N)} \colon \bigotimes_{n=0}^{N-1} \mathcal{H}_{2n} \to \bigotimes_{n=0}^{N-1} \mathcal{H}_{2n+1} \otimes \mathcal{A}_{N}$ and $U_N \colon \mathcal{H}_{2N} \otimes \mathcal{A}_{N} \to \mathcal{H}_{2N+1}$
for some Hilbert space $\mathcal{A}_{N}$ such that $\dim \mathcal{A}_{N} = k_N$ and \\
(iii) $U^{(N)}$ represents a quantum comb
with $N-1$ slots ($N \geq 2$). \\
When $N = 1$, we define $k_0 := 1$ and $U_0 := U^{(1)}$ then the proof completes.

Define
\begin{alignat}{4}
	  & P   & := & \bigotimes_{n=0}^{N-1} &   & \mathcal{H}_{2n},   \\
	  & A_I & := & \bigotimes_{n=0}^{N-1} &   & \mathcal{H}_{2n+1}, \\
	  & A_O & := &                        &   & \mathcal{H}_{2N},   \\
	  & F   & := &                        &   & \mathcal{H}_{2N+1}.
\end{alignat}
Then we can regard $U^{(N+1)} \colon P \otimes A_O \to A_I \otimes F$ as a unitary operator
representing a quantum comb
with only one slot $A$ (see Figure \ref{fig:unitaryCombOneSlotA}).
\begin{figure}
	\centering \includegraphics[keepaspectratio, scale=0.30]{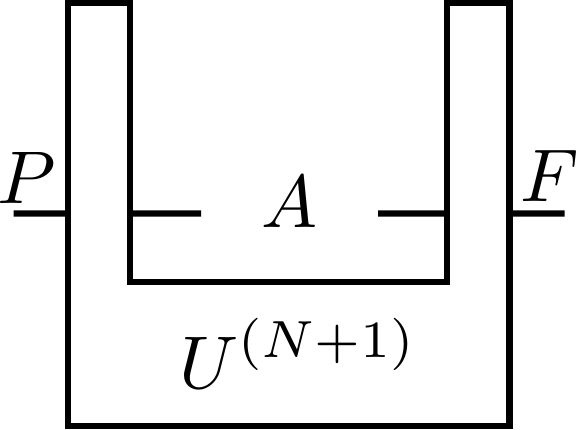}
	\caption{$U^{(N+1)}$ can be seen as a unitary operator representing a quantum comb $[\mathcal{L}({A_I}) \to \mathcal{L}({A_O})] \to [\mathcal{L}({P}) \to \mathcal{L}({F})]$
		with only one slot $A$.}
	\label{fig:unitaryCombOneSlotA}
\end{figure}
We denote the dimension of a Hilbert space $X$ by $d_X$. From unitarity of $U^{(N+1)}$, we have
\begin{equation}
	\label{eq:f98f0aa0}
	d_{P} d_{A_O} = d_{A_I} d_{F}.
\end{equation}
First, we show (i), that is, there exists a positive integer $k$ such that $d_{A_O} k = d_F$. Let $\ket{\alpha}^{A_O} \in A_O$. Using $\lnsp{V} := \lnsp{V}^{(N)}$ defined in Eq.\,\eqref{eq:defVnphi}, we have
\begin{equation}
	\label{eq:f8h0afh8ah9}
	A_I \otimes F = \lnsp{V}_{\alpha} \oplus \lnsp{V}_{\overline{\alpha}}.
\end{equation}
From Lem.\,\ref{lem:2mperp->2m+1perp}, we have
\begin{equation}
	\label{eq:f8h0afh8ah92}
	\redspn{\lnsp{V}_{\alpha}}{A_I}{F} \perp \redspn{\lnsp{V}_{\overline{\alpha}}}{A_I}{F}.
\end{equation}
In Cor.\,\ref{cor:F0perpF1=>V0=EtensorF0}, take $A_I$ as $\mathcal{E}$, $F$ as $\mathcal{F}$, $\lnsp{V}_{\alpha}$ as $\lnsp{W}_0$ and
$\lnsp{V}_{\overline{\alpha}}$ as $\lnsp{W}_1$. Then from Eq.\,\eqref{eq:f8h0afh8ah9} and Eq.\,\eqref{eq:f8h0afh8ah92}, we obtain
\begin{equation}
	\lnsp{V}_{\alpha} = A_I \otimes \redspn{\lnsp{V}_{\alpha}}{A_I}{F}.
\end{equation}
\begin{equation}
	\label{eq:f98f0aa02}
	\therefore d_P = d_{A_I} k
\end{equation}
holds for $k := \dim \redspn{\lnsp{V}_{\alpha}}{A_I}{F}$. From Eq.\,\eqref{eq:f98f0aa0} and Eq.\,\eqref{eq:f98f0aa02}, we obtain
\begin{equation}
	\label{eq:f98f0aa03}
	k d_{A_O} = d_{F}.
\end{equation}

Next, we prove (ii). Take the computational basis $\{ \ket{a}^{A_O} \}_{a=0}^{d_{A_O}-1}$ in $A_O$, the computational basis $\{ \ket{i}^{A_I} \}_{i=0}^{d_{A_I}-1}$ in $A_I$, and
an orthonormal basis $\{ \ket{x,\, 0}^{F} \}_{x=0}^{k-1}$ in $\redspn{\lnsp{V}_{0}}{A_I}{F}$.
By definition of $\lnsp{V}_{0}$, there exists $\ket{i,\, x}^P \in P$ such that
\begin{equation}
	\label{eq:r8qg8a0e9}
	U^{(N+1)} \left( \ket{i,\, x}^P \ket{0}^{A_O} \right) = \ket{i}^{A_I} \ket{x,\, 0}^{F}.
\end{equation}
Since $\{ \ket{i,\, x}^{P} \}_{i=0,\, x=0}^{d_{A_I}-1,\, k-1}$ is an orthonormal set from Eq.\,\eqref{eq:r8qg8a0e9} and a complete basis in $P$ from Eq.\,\eqref{eq:f98f0aa02},
\begin{equation}
	\label{eq:g8a9h0ejj3kd}
	\{ \ket{i,\, x}^{P} \}_{i=0,\, x=0}^{d_{A_I}-1,\, k-1} \text{ is an orthonormal basis in $P$}.
\end{equation}
For each $a$,
\begin{align}
	  & \left( {}^{A_I}\bra{i=0} U^{(N+1)} \left( \ket{i=0,\, x}^P \ket{a}^{A_O} \right), {}^{A_I}\bra{i^\prime} U^{(N+1)} \left( \ket{i^\prime,\, x^\prime}^P \ket{a}^{A_O} \right) \right) \notag \\
	= & \left( {}^{A_I}\bra{i=0} U^{(N+1)} \left( \ket{i=0,\, x}^P \ket{0}^{A_O} \right), {}^{A_I}\bra{i^\prime} U^{(N+1)} \left( \ket{i^\prime,\, x^\prime}^P \ket{0}^{A_O} \right) \right) \notag \\
	= & \delta_{x x^\prime} \label{eq:fad8u8gfa9a}
\end{align}
where the first equality follows from Lem.\,\ref{lem:sesquiperp} and Eq.\,\eqref{eq:f8h0afh8ah92} (cf. the proof of Prop.\,\ref{prop:fgijoiaojhb}).
Equation \eqref{eq:fad8u8gfa9a} for $x = x^\prime$ gives
\begin{equation}
	\label{eq:8fda98g8a}
	U^{(N+1)} \left( \ket{i,\, x}^P \ket{a}^{A_O} \right) = \ket{i}^{A_I} \ket{x,\, a}^{F},
\end{equation}
where
\begin{equation}
	\label{eq:8g9gf9h0s9}
	\ket{x,\, a}^{F} := {}^{A_I}\bra{i=0} U^{(N+1)} \left( \ket{i=0,\, x}^P \ket{a}^{A_O} \right).
\end{equation}
Note that Eq.\,\eqref{eq:8g9gf9h0s9} is defined using only $\ket{i=0,\, x}^P$, not with general $\ket{i, x}^P$ for $i \neq 0$, and $\ket{x,\, a}^{F}$ defined here for $a = 0$ is equal to $\ket{x,\, 0}^{F}$ defined previously.
For each $a$, $\{ \ket{x,\, a}^{F} \}_{x=0}^{k-1}$ is an orthonormal set as seen from Eq.\,\eqref{eq:g8a9h0ejj3kd} and Eq.\,\eqref{eq:8fda98g8a}.
Thus, $\{ \ket{x,\, a}^{F} \}_{x=0,\, a=0}^{k-1,\, d_{A_O}-1}$ is an orthonormal set from Eq.\,\eqref{eq:f8h0afh8ah92} and
is a complete basis in $F$ from Eq.\,\eqref{eq:f98f0aa03}. Therefore,	$\{ \ket{x,\, a}^{F} \}_{x=0,\, a=0}^{k-1,\, d_{A_O}-1}$ is an orthonormal basis in $F$.

Take a Hilbert space $\mathcal{A}$ such that $\dim \mathcal{A} = k$ and orthonormal basis $\{ \ket{x}^{\mathcal{A}} \}_{x=0}^{k-1}$.
Define a unitary operator $U^{(N)} \colon P \to A_I \otimes \mathcal{A}$ as	$U^{(N)} := \sum_{i,\, x} \ket{i}^{A_I} \ket{x}^{\mathcal{A}} {}^P \bra{i,\, x}$
and define a unitary operator $U_N \colon \mathcal{A} \otimes A_O \to F$ as	$U_N := \sum_{x,\, a} \ket{x,\, a}^{F} {}^{\mathcal{A}}\bra{x} {}^{A_O}\bra{a}$,
then
\begin{align}
	\label{eq:UNU(N)|ix>|a>=|i>|xa>}
	U_N U^{(N)} \ket{i,\, x}^P \ket{a}^{A_O} =  \ket{i}^{A_I} U_N \left( \ket{x}^{\mathcal{A}} \ket{a}^{A_O} \right) = \ket{i}^{A_I} \ket{x,\, a}^{F}.
\end{align}
Comparing Eqs.\,\eqref{eq:UNU(N)|ix>|a>=|i>|xa>} and \eqref{eq:8fda98g8a}, we conclude that $U^{(N+1)} = U_N U^{(N)}$.

Finally, we prove (iii).
Let $n$ be an integer in $\{1...N-1\}$. Define a subspace $\mathcal{X}^{(n)}(\ket{\alpha}^{2n}) \subset \otimes_{k} \mathcal{H}_{2k}$ as
\begin{equation}
	\mathcal{X}^{(n)}( \ket{\alpha}^{2n} ) := 	 \left\{ U^{(N+1)} \left( \ket{\tilde{\pi}}^{\widehat{2n,\, 2N}} \ket{\alpha}^{2n} \ket{0}^{2N} \right)
	\relmiddle| \ket{\tilde{\pi}}^{\widehat{2n,\, 2N}} \in \bigotimes_{k \neq n,\,N} \mathcal{H}_{2k} \right\}.
\end{equation}
From $\mathcal{X}^{(n)}_{\alpha} \subset \lnsp{V}^{(n)}_{\alpha}$, Lem.\,\ref{lem:2mperp->2m+1perp} leads to
$\redspn{\mathcal{X}^{(n)}_{\alpha}}{1 3 \cdots 2n-1}{2n+1 \cdots 2N+1}  \perp \redspn{\mathcal{X}^{(n)}_{\overline{\alpha}}}{1 3 \cdots 2n-1}{2n+1 \cdots 2N+1}$.
This implies
\begin{equation}
	\label{eq:fd8a9gu}
	\redspn{\mathcal{W}^{(n)}_{\alpha} }{1 3 \cdots 2n-1}{2n+1 \cdots 2N+1}
	\perp \redspn{\mathcal{W}^{(n)}_{\overline{\alpha}}}{1 3 \cdots 2n-1}{2n+1 \cdots 2N+1},
\end{equation}
because
\begin{alignat}{2}
	  & \left( {}^{13\cdots2n-1}\bra{\Phi} U^{(N+1)} \left( \ket{\pi}^{\widehat{2N}} \ket{0}^{2N} \right),\,
	{}^{13\cdots2n-1}\bra{\Phi^\prime} U^{(N+1)} \left( \ket{\pi^\prime}^{\widehat{2N}} \ket{0}^{2N} \right) \right) \notag \\
	= & \left( U_N \left( {}^{13\cdots2n-1}\bra{\Phi} U^{(N)} \ket{\pi}^{\widehat{2N}} \right)  \ket{0}^{2N},\,
	U_N \left( {}^{13\cdots2n-1}\bra{\Phi^\prime} U^{(N)} \ket{\pi^\prime}^{\widehat{2N}} \right) \ket{0}^{2N} \right) \notag \\
	= & \left( {}^{13\cdots2n-1}\bra{\Phi} U^{(N)} \ket{\pi}^{\widehat{2N}},\,
	{}^{13\cdots2n-1}\bra{\Phi^\prime} U^{(N)} \ket{\pi^\prime}^{\widehat{2N}}\right),
\end{alignat}
holds for all $\ket{\alpha}^{2n} \in \mathcal{H}_{2n}$, where
\begin{equation}
	\mathcal{W}^{(n)}( \ket{\alpha}^{2n} ) := \left\{ U^{(N)} \left( \ket{\tilde{\pi}}^{\widehat{2n,\, 2N}} \ket{\alpha}^{2n} \right)
	\relmiddle| \ket{\tilde{\pi}}^{\widehat{2n,\, 2N}} \in \bigotimes_{k \neq n,\,N} \mathcal{H}_{2k} \right\}.
\end{equation}
From Lem.\,\ref{lem:2mperp->2m+1perp}, Eq.\,\eqref{eq:fd8a9gu} indicates that  $U^{(N)}$ represents a quantum comb with $N-1$ slots.
\qed


\section{The proof of (1) $\implies$ (2) in Thm.\,\ref{thm:equivalent}}
\label{secApp:ProofsOfpreparation}


We start by pointing that when we insert any unitary operator $\dket{U_B}\dbra{U_B}$ into the slot $B$ of a \ReviseFurther{two-slot pure} superchannel described by $\dket{U}\dbra{U}$, the link product $\dket{U}\dbra{U} * \dket{U_B}\dbra{U_B}$ is a \ReviseFurther{one-slot pure} superchannel. This holds true because if we insert any unitary operation $\dket{U_A}\dbra{U_A}$ into the slot $A$ of the one-slot supermap described by \ReviseFurther{$\dket{U}\dbra{U} * \dket{U_B}\dbra{U_B}$}, the output operation must be unitary.
Thus, for all $U_B \colon B_I \otimes B_I^\prime \to B_O \otimes B_O^\prime$,
the product $\dket{U}\dbra{U} * \dket{U_B}\dbra{U_B}$
is the Choi operator of a unitary operation.

Hence, Lem.\,\ref{lem:Uisprocess<=>UGisunitary} shows that $\dket{U}\dbra{U}$ is a process matrix with one slot $B$.
Thus, $\dket{U}\dbra{U}$ is the Choi operator of a quantum comb whose global past is $P \otimes A_O$, whose global future is $A_I \otimes F$ and whose slot has $B_I$ ($B_O$) as an input (output) space, respectively. Thus, Lem.\,\ref{lem:2mperp->2m+1perp} shows
$\redspn{\lnsp{V}_{A_O \beta}}{B_I}{A_I F} \perp \redspn{\lnsp{V}_{A_O \overline{\beta}}}{B_I}{A_I F}$
for all $\ket{\beta}^{B_O} \in B_O$, which implies Eq.\,\eqref{b}. Similarly, we can obtain
$\redspn{\lnsp{V}_{\alpha B_O}}{A_I}{B_I F} \perp \redspn{\lnsp{V}_{\overline{\alpha} B_O}}{A_I}{B_I F}$
for all $\ket{\alpha}^{A_O} \in A_O$, which implies Eq.\,\eqref{c}.

To show Eq.\,\eqref{a}, we start with the case where $\dim A_I = \dim A_O =: d_A$ and $\dim B_I = \dim B_O =: d_B$. And since the dimensions of the input-channels are greater than one, we only consider the cases $d_A \geq 2$ and $d_B \geq 2$. We now start with the following lemma:
\begin{lemm} \
	\label{lem:3-1}
	Let $U : P \otimes A_O \otimes B_O \to A_I \otimes B_I \otimes F$ be a unitary operator representing a quantum superchannel where $\dim A_I = \dim A_O =: d_A$ and $\dim B_I = \dim B_O =: d_B$.
	For all $\ket{\pi}^P,\, \ket{\pi^\prime}^P \in P$,
	$\ket{\alpha}^{A_O},\, \ket{\alpha^\perp}^{A_O} \in A_O$ such that $\ket{\alpha}^{A_O} \perp \ket{\alpha^\perp}^{A_O}$,
	$\ket{\beta}^{B_O},\, \ket{\beta^\perp}^{B_O} \in B_O$ such that $\ket{\beta}^{B_O} \perp \ket{\beta^\perp}^{B_O}$,
	$\ket{\phi}^{A_I},\, \ket{\phi^\perp}^{A_I} \in A_I$ such that
	$\ket{\phi}^{A_I} \perp \ket{\phi^\perp}^{A_I}$,
	$\ket{\psi}^{B_I},\, \text{ and } \ket{\psi^\perp}^{B_I} \in B_I$ such that
	$\ket{\psi}^{B_I} \perp \ket{\psi^\perp}^{B_I}$, we have
	\small
	\begin{multline}
		\label{eq:3_3-1}
		\left(
		\left( {}^{A_I}\bra{\phi} {}^{B_I}\bra{\psi} \right)
		U \left( \ket{\pi}^P \ket{\alpha}^{A_O} \ket{\beta}^{B_O} \right), \right. \\
		\left. \left( {}^{A_I}\bra{\phi^\perp} {}^{B_I}\bra{\psi^\perp} \right)
		U \left( \ket{\pi^\prime}^P \ket{\alpha^\perp}^{A_O} \ket{\beta^\perp}^{B_O} \right)
		\right) = 0.
	\end{multline}
\end{lemm}
\begin{proof}

	We show only the case where $d_A \geq 2$ and $d_B \geq 2$.
	Take the computational bases $\{ \ket{j}^{A_O} \}_j$, $\{ \ket{k}^{B_O} \}_k$ in $A_O$, $B_O$, respectively.
	Without loss of generality, assume
	$\ket{\alpha}^{A_O} = \ket{0}^{A_O}$,
	$\ket{\alpha^\perp}^{A_O} = \ket{1}^{A_O}$, $\ket{\beta}^{B_O} = \ket{0}^{B_O}$ and  $\ket{\beta^\perp}^{B_O} = \ket{1}^{B_O}$.

	Let $\{ \ket{\phi_j}^{A_I} \}_{j}$ be an orthonormal basis in $A_I$ such that $\ket{\phi_0}^{A_I} = \ket{\phi}^{A_I}$, $\ket{\phi_1}^{A_I} = \ket{\phi^\perp}^{A_I}$ and
	let $\{ \ket{\psi_k}^{B_I} \}_{k}$ an orthonormal basis in $B_I$ such that $\ket{\psi_0}^{B_I} = \ket{\psi}^{B_I}$, $\ket{\psi_1}^{B_I} = \ket{\psi^\perp}^{B_I}$.
	Since $ U_A^{A_I \to A_O} \coloneqq \sum_{j} \ket{j}^{A_O} \bra{\phi_j}^{A_I}$ and $U_B^{B_I \to B_O} \coloneqq \sum_{k} \ket{k}^{B_O} \bra{\psi_k}^{B_I}$ are unitary,
	the unitary operator $U_G$ representing the output $\dket{U} \dbra{U} * \dket{U_A} \dbra{U_A} * \dket{U_B} \dbra{U_B}$ is obtained by using Eq.\,\eqref{eq:ug8d0a0a0a0g9fg9} as:
	\begin{equation}
		U_G^{P \to F}
		= \sum_{j,k} \left( {}^{A_I}\bra{\phi_j} {}^{B_I}\bra{\psi_k} \right) U^{P A_O B_O \to A_I B_I F} \left( \ket{j}^{A_O} \ket{k}^{B_O} \right)
		= \sum_{j,k} U_{j k}
	\end{equation}
	where $U_{j k} := \left( {}^{A_I}\bra{\phi_{j}} {}^{B_I}\bra{\psi_{k}} \right) U \left( \ket{j}^{A_O} \ket{k}^{B_O} \right)$.
	Thus we have,
	\begin{equation}
		I^P = U_G^\dagger U_G
		= \sum_{j,\, k,\, j^\prime,\, k^\prime} U_{jk}^\dagger U_{j^\prime k^\prime}.
		\label{eq:3-1}
	\end{equation}
	If we substitute $\ket{\phi_0}$ with $-\ket{\phi_0}$ in $U_A$, then $U_A$ remains unitary. Then Eq.\,\eqref{eq:3-1} becomes
	\begin{align}
		\label{eq:3-2}
		\sum_{j,\, k,\, j^\prime,\, k^\prime} (-1)^{\delta_{0 j} + \delta_{0 j^\prime}}
		U_{j k}^\dagger U_{j^\prime k^\prime}
		= I^P.
	\end{align}
	By subtracting Eq.\,\eqref{eq:3-2} from Eq.\,\eqref{eq:3-1}, we obtain
	\begin{align}
		\label{eq:3-3}
		\sum_{k,\, k^\prime} \biggl\{
		\sum_{j^\prime \neq 0} U_{0 k}^\dagger U_{j^\prime k^\prime} + \sum_{j \neq 0} U_{j k}^\dagger U_{0 k^\prime} \biggr\} = O.
	\end{align}
	In Eq.\,\eqref{eq:3-3}, we substitute $\ket{\phi_0}$ with $i \ket{\phi_0}$ then
	\begin{align}
		\label{eq:3-3-1}
		\sum_{k,\, k^\prime} \biggl\{ i
		\sum_{j^\prime \neq 0} U_{0 k}^\dagger U_{j^\prime k^\prime} - i \sum_{j \neq 0} U_{j k}^\dagger U_{0 k^\prime} \biggr\} = O.
	\end{align}
	By multiplying $- i$ on Eq.\,\eqref{eq:3-3-1} and subtracting from Eq.\,\eqref{eq:3-3},
	\begin{equation}
		\label{eq:3-3-2}
		\sum_{k,\, k^\prime}
		\sum_{j^\prime \neq 0} U_{0 k}^\dagger U_{j^\prime k^\prime} = O.
	\end{equation}
	We substitute $\ket{\phi_1}$ with $-\ket{\phi_1}$ then Eq.\,\eqref{eq:3-3-2} becomes
	\begin{equation}
		\label{eq:3-3-3}
		\sum_{k,\, k^\prime}
		\sum_{j^\prime \neq 0} (-1)^{\delta_{1 j^\prime} + 1} U_{0 k}^\dagger U_{j^\prime k^\prime} = O.
	\end{equation}
	By adding Eq.\,\eqref{eq:3-3-2} and Eq.\,\eqref{eq:3-3-3},
	\begin{equation}
		\label{eq:3-3-4}
		\sum_{k,\, k^\prime} U_{0 k}^\dagger U_{1 k^\prime} = O.
	\end{equation}
	We apply the same argument to $k$ and $k^\prime$
	in Eq.\,\eqref{eq:3-3-4} then acquire
	\begin{equation}
		U_{0 0}^\dagger U_{1 1} = O.
	\end{equation}
	Therefore,
	\begin{align}
		  & \left(
		\left( {}^{A_I}\bra{\phi} {}^{B_I}\bra{\psi} \right)
		U \left( \ket{\pi}^P \ket{\alpha}^{A_O} \ket{\beta}^{B_O} \right),\,
		\left( {}^{A_I}\bra{\phi^\perp} {}^{B_I}\bra{\psi^\perp} \right)
		U \left( \ket{\pi^\prime}^P \ket{\alpha^\perp}^{A_O} \ket{\beta^\perp}^{B_O} \right)
		\right) \notag \\
		= & \left(
		\left( {}^{A_I}\bra{\phi_0} {}^{B_I}\bra{\psi_0} \right)
		U \left( \ket{\pi}^P \ket{0}^{A_O} \ket{0}^{B_O} \right),
		\left( {}^{A_I}\bra{\phi_1} {}^{B_I}\bra{\psi_1} \right)
		U \left( \ket{\pi^\prime}^P \ket{1}^{A_O} \ket{1}^{B_O} \right)
		\right) \notag \\
		= & \left( U_{00} \ket{\pi}^P,\, U_{11} \ket{\pi^\prime}^P \right) \notag
		= {}^{P}\bra{\pi} U_{0 0}^{\dagger} U_{1 1} \ket{\pi^\prime}^P = 0.
	\end{align}
\end{proof}

Next, we show
\begin{equation}
	\left(
	\left( {}^{A_I}\bra{\phi} {}^{B_I}\bra{\psi} \right)
	U \left( \ket{\pi}^P \ket{\alpha}^{A_O} \ket{\beta}^{B_O} \right),
	\left( {}^{A_I}\bra{\phi^\prime} {}^{B_I}\bra{\psi^\perp} \right)
	U \left( \ket{\pi^\prime}^P \ket{\alpha^\perp}^{A_O} \ket{\beta^\perp}^{B_O} \right) \right) = 0
\end{equation}
for all $\ket{\phi}^{A_I},\, \ket{\phi^\prime}^{A_I} \in A_I$ and for all $\ket{\psi}^{B_I},\, \ket{\psi^\perp}^{B_I} \in B_I$ such that $\ket{\psi}^{B_I} \perp \ket{\psi^\perp}^{B_I}$. Define the sesquilinear function
\small
\begin{align}
	  & f(\ket{\phi^\prime}^{A_I},\, \ket{\phi}^{A_I}) := \\
	  & \quad \quad \left(
	\left( {}^{A_I}\bra{\phi} {}^{B_I}\bra{\psi} \right)
	U \left( \ket{\pi}^P \ket{\alpha}^{A_O} \ket{\beta}^{B_O} \right),
	\left( {}^{A_I}\bra{\phi^\prime} {}^{B_I}\bra{\psi^\perp} \right)
	U \left( \ket{\pi^\prime}^P \ket{\alpha^\perp}^{A_O} \ket{\beta^\perp}^{B_O} \right) \right). \notag
\end{align}
\normalsize
From Lem.\,\ref{lem:3-1}, we have
\begin{equation}
	\label{eq:fa98fa-1}
	f(\ket{\phi^\perp}^{A_I},\, \ket{\phi}^{A_I}) = 0
\end{equation}
for all $\ket{\phi}^{A_I},\, \ket{\phi^\perp}^{A_I} \in A_I$ such that
$\ket{\phi}^{A_I} \perp \ket{\phi^\perp}^{A_I}$. Lem.\,\ref{lem:sesquiperp} and
this show $f(\ket{\phi}^{A_I},\, \ket{\phi}^{A_I}) = f(\ket{\phi^\prime}^{A_I},\, \ket{\phi^\prime}^{A_I})$
for all $\ket{\phi}^{A_I},\, \ket{\phi^\prime}^{A_I} \in A_I$ such that $\| \ket{\phi}^{A_I} \| = \| \ket{\phi^\prime}^{A_I} \|$. Thus, for all $\ket{\phi}^{A_I} \in A_I,\, \neq 0$, we have
\begin{equation}
	\label{eq:f(phi,phi)InLemma3-1}
	d_A \frac{f(\ket{\phi}^{A_I},\, \ket{\phi}^{A_I}) }{\| \ket{\phi}^{A_I} \|^2} = \sum_{j = 0}^{d_A - 1} f \left( \frac{\ket{\phi}^{A_I}}{\| \ket{\phi}^{A_I} \|},\, \frac{\ket{\phi}^{A_I}}{\| \ket{\phi}^{A_I} \|} \right) = \sum_{j = 0}^{d_A - 1} f(\ket{j}^{A_O},\, \ket{j}^{A_O}),
\end{equation}
where $\{ \ket{j}^{A_I} \}_{j=0}^{d_A-1}$ is the computational basis in $A_I$.
Recalling the definition of $f$ and completeness relation of $\{ \ket{j}^{A_I} \}_{j}$,
\begin{align}
	  & \text{Eq.\,\eqref{eq:f(phi,phi)InLemma3-1}} \notag                                                                                                   \\
	= & \sum_{j = 0}^{d_A - 1} \left( {}^P\bra{\pi} {}^{A_O}\bra{\alpha} {}^{B_O}\bra{\beta} \right) U^\dagger \left( \ket{j}^{A_I} \ket{\psi}^{B_I} \right)
	\left( {}^{A_I}\bra{j} {}^{B_I}\bra{\psi^\perp} \right)
	U \left( \ket{\pi^\prime}^P \ket{\alpha^\perp}^{A_O} \ket{\beta^\perp}^{B_O} \right) \notag \\
	= & \left( {}^P\bra{\pi} {}^{A_O}\bra{\alpha} {}^{B_O}\bra{\beta} \right) U^\dagger  \ket{\psi}^{B_I}
	{}^{B_I}\bra{\psi^\perp}
	U \left( \ket{\pi^\prime}^P \ket{\alpha^\perp}^{A_O} \ket{\beta^\perp}^{B_O} \right) \notag \\
	= & \left(
	{}^{B_I}\bra{\psi}
	U \left( \ket{\pi}^P \ket{\alpha}^{A_O} \ket{\beta}^{B_O} \right),
	{}^{B_I}\bra{\psi^\perp}
	U \left( \ket{\pi^\prime}^P \ket{\alpha^\perp}^{A_O} \ket{\beta^\perp}^{B_O} \right) \right) = 0,
\end{align}
where the last equality follows from Eq.\,\eqref{b}.
Therefore, we obtain
\begin{equation}
	\label{eq:fa98fa}
	f(\ket{\phi}^{A_I},\, \ket{\phi}^{A_I}) = 0
\end{equation}
for all $\ket{\phi}^{A_I} \in A_I$. Eq.\,\eqref{eq:fa98fa-1} and Eq.\,\eqref{eq:fa98fa} implies $f(\ket{\phi^\prime}^{A_I},\, \ket{\phi}^{A_I}) = 0$, that is,
\small
\begin{equation}
	\label{eq:f8a9fa9f}
	\left(
	\left( {}^{A_I}\bra{\phi} {}^{B_I}\bra{\psi} \right)
	U \left( \ket{\pi}^P \ket{\alpha}^{A_O} \ket{\beta}^{B_O} \right),
	\left( {}^{A_I}\bra{\phi^\prime} {}^{B_I}\bra{\psi^\perp} \right)
	U \left( \ket{\pi^\prime}^P \ket{\alpha^\perp}^{A_O} \ket{\beta^\perp}^{B_O} \right) \right) = 0
\end{equation}
\normalsize
for all $\ket{\phi}^{A_I},\, \ket{\phi^\prime}^{A_I} \in A_I$ and for all $\ket{\psi}^{B_I},\, \ket{\psi^\perp}^{B_I} \in B_I$ such that $\ket{\psi}^{B_I} \perp \ket{\psi^\perp}^{B_I}$. Similarly to the case of Eq.\,\eqref{eq:f8a9fa9f}, we can show
\small
\begin{equation}
	\label{eq:resultOf1to2WhereDimensionRestricted}
	\left(
	\left( {}^{A_I}\bra{\phi} {}^{B_I}\bra{\psi} \right)
	U \left( \ket{\pi}^P \ket{\alpha}^{A_O} \ket{\beta}^{B_O} \right),
	\left( {}^{A_I}\bra{\phi^\prime} {}^{B_I}\bra{\psi^\prime} \right)
	U \left( \ket{\pi^\prime}^P \ket{\alpha^\perp}^{A_O} \ket{\beta^\perp}^{B_O} \right) \right) = 0
\end{equation}
\normalsize
using Lem.\,\ref{lem:sesquiperp},\, Eq.\,\eqref{eq:f8a9fa9f} and Eq.\,\eqref{c}.

Now we investigate the case where $\dim A_I \neq \dim A_O$ or $\dim B_I \neq \dim B_O$. Define a unitary operator
\begin{equation}
	\tilde{U} := U \otimes I^{A_P} \otimes I^{A_F} \otimes I^{B_P} \otimes I^{B_F}
\end{equation}
where $A_P$, $A_F$, $B_P$, $B_F$ are Hilbert spaces such that $A_P \cong A_O$, $A_F \cong A_I$, $B_P \cong B_O$, $B_F \cong B_I$, respectively.
Then, $\tilde{U} \colon \tilde{P} \otimes \tilde{A}_O \otimes \tilde{B}_O \to \tilde{A}_I \otimes \tilde{B}_I \otimes \tilde{F}$ represents a two-slot superchannel whose global past is $\tilde{P} := P \otimes A_P \otimes B_P$, whose global future is $\tilde{F} := F \otimes A_F \otimes B_F$, one of whose slots has $\tilde{A}_I :=A_I \otimes A_P$
($\tilde{A}_O := A_O \otimes A_F$) as an input (output) space and the other of whose slots has $\tilde{B}_I := B_I \otimes B_P$ ($\tilde{B}_O := B_O \otimes B_F$) as an input (output) space.
Now we have $\dim \tilde{A}_I = \dim A_I \dim A_P = \dim A_F \dim A_O = \dim \tilde{A}_O$ and similarly we have $\dim \tilde{B}_I = \dim \tilde{B}_O$. Thus, Eq.\,\eqref{eq:resultOf1to2WhereDimensionRestricted} can be replaced by
\small
\begin{equation}
	\label{eq:fajoib}
	\left(
	\left( {}^{\tilde{A}_I}\bra{\tilde{\phi}} {}^{\tilde{B}_I}\bra{\tilde{\psi}} \right)
	\tilde{U} \left( \ket{\tilde{\pi}}^{\tilde{P}} \ket{\tilde{\alpha}}^{\tilde{A}_O} \ket{\tilde{\beta}}^{\tilde{B}_O} \right),
	\left( {}^{\tilde{A}_I}\bra{\tilde{\phi}^\prime} {}^{\tilde{B}_I}\bra{\tilde{\psi}^\prime} \right)
	\tilde{U} \left( \ket{\tilde{\pi}^\prime}^{\tilde{P}} \ket{\tilde{\alpha}^\perp}^{\tilde{A}_O} \ket{\tilde{\beta}^\perp}^{\tilde{B}_O} \right) \right) = 0
\end{equation}
\normalsize
for all $\ket{\tilde{\pi}}^{\tilde{P}},\, \ket{\tilde{\pi}^\prime}^{\tilde{P}} \in {\tilde{P}}$,
$\ket{\tilde{\alpha}}^{\tilde{A}_O},\, \ket{\tilde{\alpha}^\perp}^{\tilde{A}_O} \in \tilde{A}_O$ such that $\ket{\tilde{\alpha}}^{\tilde{A}_O} \perp \ket{\tilde{\alpha}^\perp}^{\tilde{A}_O}$,
$\ket{\tilde{\beta}}^{\tilde{B}_O},\, \ket{\tilde{\beta}^\perp}^{\tilde{B}_O} \in \tilde{B}_O$ such that $\ket{\tilde{\beta}}^{\tilde{B}_O} \perp \ket{\tilde{\beta}^\perp}^{\tilde{B}_O}$,
$\ket{\tilde{\phi}}^{\tilde{A}_I},\, \ket{\tilde{\phi}^\perp}^{\tilde{A}_I} \in \tilde{A}_I$ and
$\ket{\tilde{\psi}}^{\tilde{B}_I},\, \ket{\tilde{\psi}^\perp}^{\tilde{B}_I} \in \tilde{B}_I$.
Substitute
$\ket{\tilde{\pi}}^{\tilde{P}} = \ket{\pi}^{P} \ket{0}^{A_P} \ket{0}^{B_P}$,
$\ket{\tilde{\pi}^\prime}^{\tilde{P}} = \ket{\pi^\prime}^{P} \ket{0}^{A_P} \ket{0}^{B_P}$,
$\ket{\tilde{\alpha}}^{\tilde{A}_O} = \ket{\alpha}^{A_O} \ket{0}^{A_F}$,
$\ket{\tilde{\alpha}^\perp}^{\tilde{A}_O} = \ket{\alpha^\perp}^{A_O} \ket{0}^{A_F}$,
$\ket{\tilde{\beta}}^{\tilde{B}_O} = \ket{\beta}^{B_O} \ket{0}^{B_F}$,
$\ket{\tilde{\beta}^\perp}^{\tilde{B}_O} = \ket{\beta^\perp}^{B_O} \ket{0}^{B_F}$,
$\ket{\tilde{\phi}}^{\tilde{A}_I} = \ket{\phi}^{A_I} \ket{0}^{A_P}$,
$\ket{\tilde{\phi}^\prime}^{\tilde{A}_I} = \ket{\phi^\prime}^{A_I} \ket{0}^{A_P}$,
$\ket{\tilde{\psi}}^{\tilde{B}_I} = \ket{\psi}^{B_I} \ket{0}^{B_P}$ and
$\ket{\tilde{\psi}^\prime}^{\tilde{B}_I} = \ket{\psi^\prime}^{B_I} \ket{0}^{B_P}$ where $\ket{0}^{\mathcal{X}}$ is a normalized vector in a Hilbert space $\mathcal{X}$. Then,
\begin{align}
	0 = & \text{ (LHS of Eq.\,\eqref{eq:fajoib})} \\
	=   & \begin{multlined}
		\left(
		\left( {}^{A_I}\bra{\phi} {}^{B_I}\bra{\psi} \right)
		U \left( \ket{\pi}^P \ket{\alpha}^{A_O} \ket{\beta}^{B_O} \right) \ket{0}^{A_F} \ket{0}^{B_F}, \right. \\
		\left. \left( {}^{A_I}\bra{\phi^\prime} {}^{B_I}\bra{\psi^\prime} \right)
		U \left( \ket{\pi^\prime}^P \ket{\alpha^\perp}^{A_O} \ket{\beta^\perp}^{B_O} \right) \ket{0}^{A_F} \ket{0}^{B_F} \right)
	\end{multlined}
	\\
	=   & \left(
	\left( {}^{A_I}\bra{\phi} {}^{B_I}\bra{\psi} \right)
	U \left( \ket{\pi}^P \ket{\alpha}^{A_O} \ket{\beta}^{B_O} \right),
	\left( {}^{A_I}\bra{\phi^\prime} {}^{B_I}\bra{\psi^\prime} \right)
	U \left( \ket{\pi^\prime}^P \ket{\alpha^\perp}^{A_O} \ket{\beta^\perp}^{B_O} \right) \right).
\end{align}
\qed

\section{The proof of (2) $\implies$ (3) in Thm.\,\ref{thm:equivalent}}
\label{secApp:idea}

First, we rephrase the condition (2) using the notation of reduced subspaces introduced in Appendix \ref{secApp:partsp} \Add{so that we can grasp the interpretation of the condition and we effectively utilize the condition in the proof}.
\begin{description}
	\item[(2$^\prime$)] For all $\ket{\alpha}^{A_O} \in A_O$ and $\ket{\beta}^{B_O} \in B_O$,
	\begin{equation}
		\label{A}
		\redspn{ \lnsp{V}_{\alpha \beta} }{A_I B_I}{F} \perp \redspn{ \lnsp{V}_{\overline{\alpha} \overline{\beta}} }{A_I B_I}{F},
		\tag{A}
	\end{equation}
	\begin{equation}
		\label{C}
		\redspn{ \lnsp{V} _{A_O \beta} }{B_I}{A_I F} \perp \redspn{ \lnsp{V} _{A_O \overline{\beta}} }{B_I}{A_I F},
		\tag{B}
	\end{equation}
	\begin{equation}
		\label{B}
		\redspn{ \lnsp{V}_{\alpha B_O} }{A_I}{B_I F} \perp \redspn{ \lnsp{V}_{\overline{\alpha} B_O} }{A_I}{B_I F},
		\tag{C}
	\end{equation}
	where $\alpha:=\vspan{\ket{\alpha}}$, \, $\overline{\alpha}:=\vspan{\ket{\alpha}}^\perp$,
	\begin{equation}
		\label{eq:defVWab}
		\lnsp{V}_{A_{\mathrm{sub}} B_{\mathrm{sub}}}:=\lnsp{V}(A_{\mathrm{sub}}, B_{\mathrm{sub}}) := U \left( P \otimes A_{\mathrm{sub}} \otimes B_{\mathrm{sub}} \right).
	\end{equation}
	for a subspace $A_{\mathrm{sub}}$, $B_{\mathrm{sub}}$ of $A_O$, $B_O$, respectively,
	and $U(\bullet)$ means the image of a subspace $\bullet$ under $U$.
\end{description}

In Eq.\,\eqref{A}, Eq.\,\eqref{C}, and Eq.\,\eqref{B},
$\redspn{\bullet}{A_I B_I}{F}$, $\redspn{\bullet}{A_I}{B_I F}$, and  $\redspn{\bullet}{B_I}{A_I F}$ mean reducing the subsystem of $A_I \otimes B_I$, $A_I$, $B_I$ and extracting the states on the subsystem of $F$, $B_I \otimes F$, $A_I \otimes F$, respectively.
\Add{
	Eq.\,\eqref{A} means that if two orthogonal pure states are ``fed'' into $A_O$ and $B_O$, that is, two orthogonal pure states are used as input-states in $A_O$ and $B_O$ for $U$, respectively, and arbitrary input-states in $P$ for $U$ are used, then \ReviseFurther{the reduced states on $F$ of the two corresponding output-states are orthogonal}.
	Briefly, Eq.\,\eqref{A} means that ``feeding'' orthogonal pure states into $A_O$ and $B_O$ causes orthogonal reduced states on $F$. Similarly, \ReviseFurther{Eq.\,\eqref{C} (Eq.\,\eqref{B})} means that ``feeding'' orthogonal pure states into \ReviseFurther{$B_O$ ($A_O$)} causes orthogonal reduced states on \ReviseFurther{$A_I \otimes F$ ($B_I \otimes F$)}, respectively.
}

\Add{Next, to figure out how to show (2)$^\prime$ $\implies$ (3), that is, how to decompose $U$ into $U^{A \prec B} \oplus U^{B \prec A}$ by using the condition (2)$^\prime$, we consider properties of unitary operators representing two-slot pure combs.
	We start from how to distinguish the case $U = U^{A \prec B}$ and the case $U = U^{B \prec A}$. As seen in Cor.\,\ref{cor:N=2purecomb}, in the case of $U = U^{B \prec A}$, ``feeding'' orthogonal pure states into $A_O$ causes orthogonal reduced states in $F$. Thus, if ``feeding'' orthogonal pure states into $A_O$ causes intersection in $F$, then we can identify $U = U^{A \prec B}$. }

\Add{To decompose $U$ into $U^{A \prec B} \oplus U^{B \prec A}$, we use similar strategies. We focus on $\redspn{ \lnsp{V}_{\alpha \beta} }{A_I B_I}{F} \in F$, which includes reduced states in $F$ by ``feeding'' $\ket{\alpha}^{A_O}$ into $A_O$ and $\ket{\beta}^{B_O}$ into $B_O$. We consider extracting intersecting vectors in $\redspn{ \lnsp{V}_{\alpha \beta} }{A_I B_I}{F}$ with $\redspn{ \lnsp{V}_{\overline{\alpha} \beta} }{A_I B_I}{F} \in F$, which includes reduced states in $F$ by ``feeding'' orthogonal pure states to $\ket{\alpha}^{A_O}$ into $A_O$ and $\ket{\beta}^{B_O}$ into $B_O$, and gather these vectors into a space named $F^|_{\alpha \beta}$, supposed to include only ``$A \prec B$ vectors.'' Similarly, we construct $F^-_{\alpha \beta}$, supposed to include only ``$B \prec A$ vectors.'' Besides $F^|_{\alpha \beta}$ and $F^-_{\alpha \beta}$, we have a remnant subspace $F^\square_{\alpha \beta}$ orthogonal to both $\redspn{ \lnsp{V}_{\overline{\alpha} \beta} }{A_I B_I}{F}$ and $\redspn{ \lnsp{V}_{\alpha \overline{\beta}} }{A_I B_I}{F}$, supposed to include ``$A \prec B$ vectors'' and ``$B \prec A$ vectors.'' We prove that we can decompose $\redspn{ \lnsp{V}_{\alpha \beta} }{A_I B_I}{F} = F^|_{\alpha \beta} \oplus F^\square_{\alpha \beta} \oplus F^-_{\alpha \beta}$ in Appendix \ref{secApp:divideFintoF|F-Fsq} (Fig.\,\ref{fig:csqctable}). }

\begin{figure}
	\centering \includegraphics[width=0.8	\columnwidth]{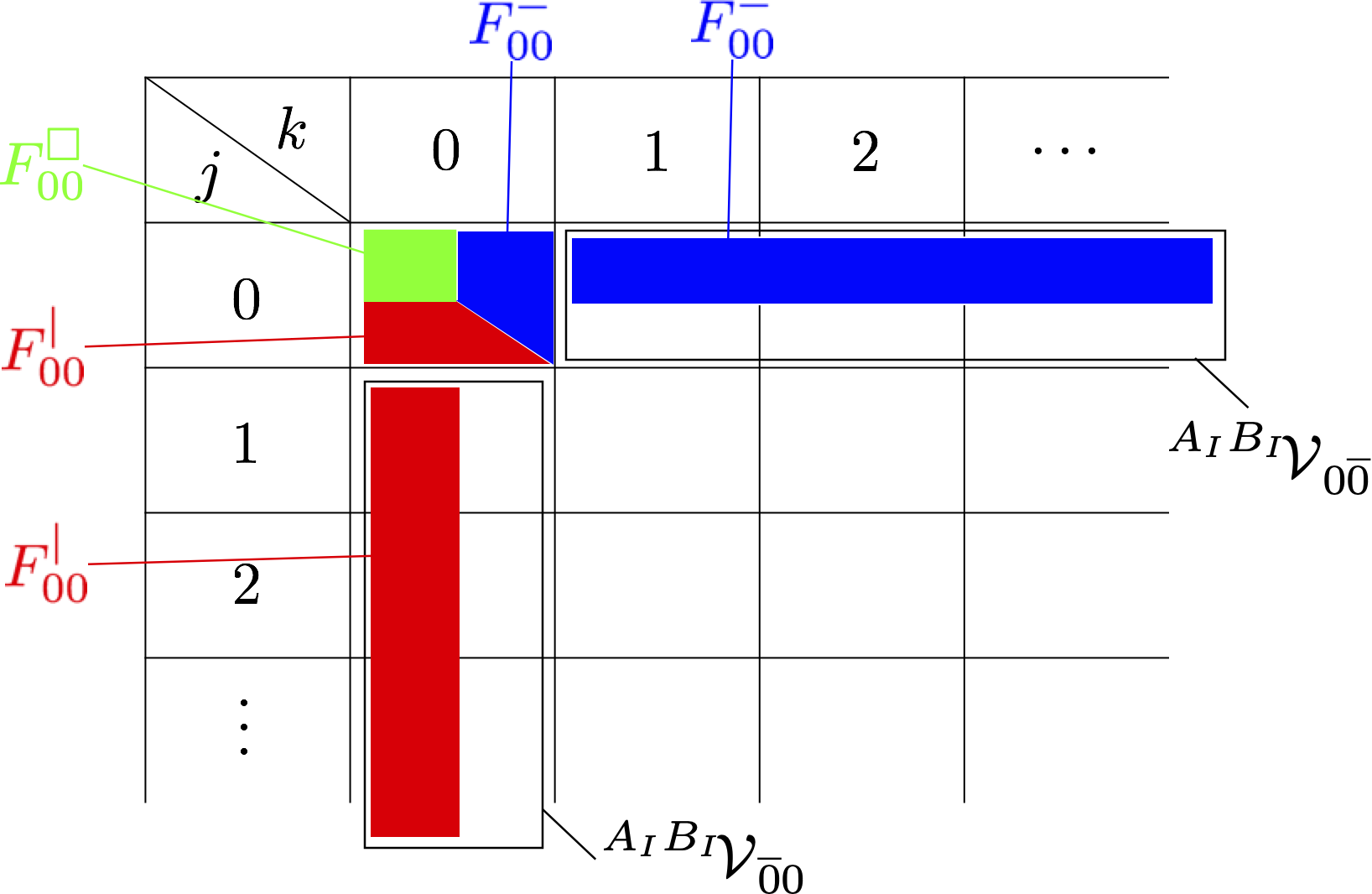}
	\caption{
	\label{fig:csqctable}
	Pictorial illustration of the main idea used in Appendix \ref{secApp:divideFintoF|F-Fsq} \Add{to show $\redspn{ \lnsp{V}_{\alpha \beta} }{A_I B_I}{F} = F^|_{\alpha \beta} \oplus F^\square_{\alpha \beta} \oplus F^-_{\alpha \beta}$}.
	The cell at $j$ row, $k$ column represents $\redspn{\lnsp{V}_{jk}}{A_I B_I}{F}$ where $\{ \ket{j}^{A_O} \}_j$, $\{ \ket{k}^{A_O} \}_k$ are computational bases in $A_O$, $B_O$, respectively.
	The parts colored with red, green and blue represent the subspaces $F^|_{00}$, $F^\square_{00}$ and $F^-_{00}$, \ReviseFurther{respectively}.
	$F^|_{00}$ is included in both $\redspn{\lnsp{V}_{00}}{A_I B_I}{F}$ and $\redspn{\lnsp{V}_{\overline{0} 0}}{A_I B_I}{F}$ and
	$F^-_{00}$ is included in both $\redspn{\lnsp{V}_{00}}{A_I B_I}{F}$ and $\redspn{\lnsp{V}_{0 \overline{0}}}{A_I B_I}{F}$.
	}
\end{figure}

\Add{In Appendix \ref{secApp:divideV}, we construct three subspaces $\lnsp{V}^{|}_{\alpha \beta}$, $\lnsp{V}^{\square}_{\alpha \beta}$, $\lnsp{V}^{-}_{\alpha \beta}$ of $\lnsp{V}_{\alpha \beta}$ corresponding with the three subspaces $F^|_{\alpha \beta}, F^\square_{\alpha \beta}, F^-_{\alpha \beta}$ of $F$ and show decompositions $\lnsp{V}_{\alpha \beta} = \lnsp{V}^|_{\alpha \beta} \oplus \lnsp{V}^\square_{\alpha \beta} \oplus \lnsp{V}^-_{\alpha \beta}$.}
\Add{In Appendix \ref{secApp:dividePintoPABopPBAopPparallel}, we construct three subspaces $P^|_{\alpha \beta}$, $P^\square_{\alpha \beta}$, $P^-_{\alpha \beta}$ of $P$ corresponding with the three subspaces $\lnsp{V}^{|}_{\alpha \beta}$, $\lnsp{V}^{\square}_{\alpha \beta}$, $\lnsp{V}^{-}_{\alpha \beta}$ of $\lnsp{V}_{\alpha \beta}$ and show decompositions $P = P^|_{\alpha \beta} \oplus P^\square_{\alpha \beta} \oplus P^-_{\alpha \beta}$.}

\Add{From what \ReviseFurther{we have} seen so far, we can guess that $P^|_{\alpha \beta}$, $P^-_{\alpha \beta}$ relate to $A \prec B$, $B \prec A$ while the two subspaces do not coincide with $P^{A \prec B}$, $P^{B \prec A}$, respectively. In Appendix \ref{secApp:dividePintoPABopPBAopPparallel}, we also decompose $P$ into three subspaces which are directly related to the subspaces $P^{A \prec B}$, $P^{B \prec A}$ we want. We construct three subspaces $\tilde{P}^| := +_{\alpha \beta} P^|_{\alpha \beta}$, $\tilde{P}^\square := \cap_{\alpha \beta} P^\square_{\alpha \beta}$, $\tilde{P}^- := +_{\alpha \beta} P^-_{\alpha \beta}$ and show a decomposition $P = \tilde{P}^| \oplus \tilde{P}^\square \oplus \tilde{P}^-$. These three subspaces $\tilde{P}^|$, $\tilde{P}^\square$, $\tilde{P}^-$ correspond to subspaces of $P$ with $A \prec B$, $A \parallel B$, $B \prec A$, respectively. Since $A \parallel B$ can be seen as $A \prec B$, we define $P^{A \prec B} := \tilde{P}^| \oplus \tilde{P}^\square$, $P^{B \prec A} := \tilde{P}^-$ later.
In Appendix \ref{secApp:divideFintoFABoplFBAoplFpal}, we construct three subspaces $\tilde{F}^|$, $\tilde{F}^\square$, $\tilde{F}^-$ of $F$ corresponding with the three subspaces $\tilde{P}^|$, $\tilde{P}^\square$, $\tilde{P}^-$ of $P$, respectively, and show a decomposition $F = \tilde{F}^| \oplus \tilde{F}^\square \oplus \tilde{F}^-$. We also define $F^{A \prec B} := \tilde{F}^| \oplus \tilde{F}^\square$, $F^{B \prec A} := \tilde{F}^-$ later.
In Appendix \ref{secApp:causality}, we define $P^{A \prec B}$, $P^{B \prec A}$, $F^{A \prec B}$, and $F^{B \prec A}$ properly as mentioned and show a decomposition $U = U^{A \prec B} \oplus U^{B \prec A}$ where $U^{A \prec B} \colon P^{A \prec B} \otimes A_O \otimes B_O \to A_I \otimes B_I \otimes F^{A \prec B}$, $U^{B \prec A}\colon P^{B \prec A} \otimes A_O \otimes B_O \to A_I \otimes B_I \otimes F^{B \prec A}$ are unitary operators representing quantum combs of $A \prec B$, $B \prec A$, respectively. It completes the proof.}
\Add{Appendix \ref{secApp:proofsInPDecomp}, Appendix \ref{secApp:proofFFFFdecmpdoijfaj}, and Appendix \ref{secApp:proofofFeqFABoplFBAoplFpal} we have not referred provide proofs of propositions used in other subsections.}

We show (2)$^\prime$ $\implies$ (3) in the following.
Figure \ref{fig:dependencyGraphOfProof} presents implication relations between propositions shown in the following.
\begin{figure}
	\centering \includegraphics{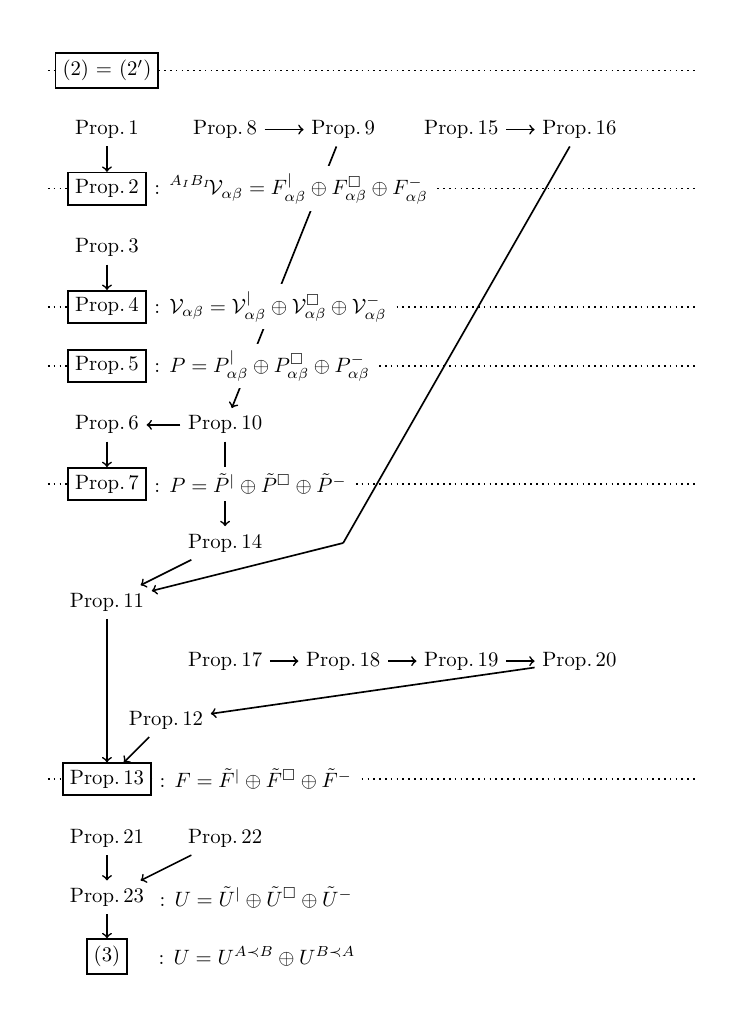}
	\caption{Implication relations between propositions for the proof of (2) $\implies$ (3) presented in Appendix \ref{secApp:idea}.
		Main propositions are boxed.
		The position of each proposition indicates that boxed propositions in a higher place are likely to be used. For details, see the main text.}
	\label{fig:dependencyGraphOfProof}
\end{figure}


\subsection{Decomposing $\redspn{\lnsp{V}_{\alpha \beta}}{A_I B_I}{F}$}
\label{secApp:divideFintoF|F-Fsq}

We decompose the reduced subspace $\redspn{\lnsp{V}_{\alpha \beta}}{A_I B_I}{F}$ into $F^|_{\alpha \beta} \oplus F^\square_{\alpha \beta} \oplus F^-_{\alpha \beta}$. First we show the following proposition.

\begin{prop}
	\label{prop:BF00=WopX}
	$\redspn{\lnsp{V}_{\alpha \beta}}{A_I B_I}{F}$ is decomposed into
	\begin{equation}
		\label{eq:BF00=WopXprp}
		\redspn{\lnsp{V}_{\alpha \beta}}{A_I B_I}{F} = F^|_{\alpha \beta} \oplus F^{\text{not}|}_{\alpha \beta}
	\end{equation}
	by using the subspaces $F^|_{\alpha \beta}$, $F^{\text{not}|}_{\alpha \beta}$ of $\redspn{\lnsp{V}_{\alpha \beta}}{A_I B_I}{F}$ satisfying
	\begin{equation}
		\label{eq:F|00subF0ex0}
		F^|_{\alpha \beta} \subset \redspn{\lnsp{V}_{\overline{\alpha} \beta}}{A_I B_I}{F},
	\end{equation}
	\begin{equation}
		\label{eq:Fex|00subF0ex0}
		F^{\text{not}|}_{\alpha \beta} \perp \redspn{\lnsp{V}_{\overline{\alpha} \beta}}{A_I B_I}{F}.
	\end{equation}
	Moreover, $\redspn{\lnsp{V}_{\alpha \beta}}{A_I B_I}{F}$ is decomposed into
	\begin{equation}
		\label{eq:F00=F-opFex-}
		\redspn{\lnsp{V}_{\alpha \beta}}{A_I B_I}{F} = F^-_{\alpha \beta} \oplus F^{\text{not}-}_{\alpha \beta}
	\end{equation}
	by using the subspaces $F^-_{\alpha \beta}$, $F^{\text{not}-}_{\alpha \beta}$ of $\redspn{\lnsp{V}_{\alpha \beta}}{A_I B_I}{F}$ satisfying
	\begin{equation}
		\label{eq:F-00subFex00}
		F^-_{\alpha \beta} \subset \redspn{\lnsp{V}_{\alpha \overline{\beta}}}{A_I B_I}{F},
	\end{equation}
	\begin{equation}
		\label{eq:Fex-00perpFex00}
		F^{\text{not}-}_{\alpha \beta} \perp \redspn{\lnsp{V}_{\alpha \overline{\beta}}}{A_I B_I}{F}.
	\end{equation}
\end{prop}
\begin{proof} \
	We show only the first part. Since $B_I \otimes (A_I \otimes F) = \lnsp{V}_{A_O \beta} \oplus \lnsp{V}_{A_O \overline{\beta}}$ \ReviseFurther{and $\redspn{\lnsp{V}_{A_O \beta}}{B_I}{A_I F} \perp \redspn{\lnsp{V}_{A_O \overline{\beta}}}{B_I}{A_I F}$ due to Eq.\,\eqref{C}},
	by taking $B_I$ as $\mathcal{E}$, $A_I \otimes F$ as $\mathcal{F}$, $\lnsp{V}_{A_O \beta}$ as $\lnsp{W}_0$  and $\lnsp{V}_{A_O \overline{\beta}}$ as $\lnsp{W}_1$ in Cor.\,\ref{cor:F0perpF1=>V0=EtensorF0}, we obtain
	\begin{equation}
		\label{eq:V0B=A(BF)0B}
		\lnsp{V}_{A_O \beta} = B_I \otimes \redspn{\lnsp{V}_{A_O \beta}}{B_I}{A_I F}.
	\end{equation}
	Since $B_I \otimes \redspn{\lnsp{V}_{A_O \beta}}{B_I}{A_I F} = \lnsp{V}_{A_O \beta} = \lnsp{V}_{\alpha \beta} \oplus \lnsp{V}_{\overline{\alpha} \beta}$ from Eq.\,\eqref{eq:V0B=A(BF)0B}, taking $B_I$ as $\mathcal{E}$, $\redspn{\lnsp{V}_{A_O \beta}}{B_I}{A_I F}$ as $\mathcal{F}$,
	$\lnsp{V}_{\alpha \beta}$ as $\lnsp{W}_0$ and $\lnsp{V}_{\overline{\alpha} \beta}$ as $\lnsp{W}_1$ in  in Cor.\,\ref{cor:EF=V0opV1->F0=V1perpFop(F0capF1)}, then $\redspn{\lnsp{V}_{\alpha \beta}}{B_I}{A_I F}$ can be decomposed into
	\begin{equation}
		\label{eq:BF00=WopX}
		\redspn{\lnsp{V}_{\alpha \beta}}{B_I}{A_I F} = \lnsp{Z}^| \oplus \lnsp{Z}^{\text{not}|}
	\end{equation}
	by using subspaces $\lnsp{Z}^|$ and $\lnsp{Z}^{\text{not}|}$ of $\redspn{\lnsp{V}_{\alpha \beta}}{B_I}{A_I F}$ satisfying
	\begin{equation}
		\label{eq:Wsub(BF)0ex0}
		\lnsp{Z}^| \subset \redspn{\lnsp{V}_{\overline{\alpha} \beta}}{B_I}{A_I F},
	\end{equation}
	\begin{equation}
		\label{eq:Xperp(BF)0ex0}
		\lnsp{Z}^{\text{not}|} \perp \redspn{\lnsp{V}_{\overline{\alpha} \beta}}{B_I}{A_I F}.
	\end{equation}
	From Eq.\,\eqref{eq:Wsub(BF)0ex0}, we have
	\begin{equation}
		\label{eq:F_WsubF0ex0}
		\ReviseFurther{(F^|_{\alpha \beta} :=)} \ \redspn{\lnsp{Z}^|}{A_I}{F} \subset \redspn{\lnsp{V}_{\overline{\alpha} \beta}}{A_I B_I}{F}.
	\end{equation}
	Since $B_I \otimes \redspn{\lnsp{V}_{A_O \beta}}{B_I}{A_I F} = \lnsp{V}_{A_O \beta} = \lnsp{V}_{\alpha \beta} \oplus \lnsp{V}_{\overline{\alpha} \beta}$ from Eq.\,\eqref{eq:V0B=A(BF)0B}, taking $B_I$ as $\mathcal{E}$\ReviseFurther{,}
	$\redspn{\lnsp{V}_{A_O \beta}}{B_I}{A_I F}$ as $\mathcal{F}$, $\lnsp{V}_{\alpha \beta}$ as $\lnsp{W}_0$, $\lnsp{V}_{\overline{\alpha} \beta}$ as $\lnsp{W}_1$ and $\lnsp{Z}^{\text{not}|}$ as $\mathcal{X}$
	in Lem.\,\ref{lem:UperpV0FeqEotUsubV1} and using Eq.\,\eqref{eq:Xperp(BF)0ex0} we obtain
	\begin{alignat}{2}
		  &            & B_I \otimes \lnsp{Z}^{\text{not}|}                                          & \subset \lnsp{V}_{\alpha \beta},
		\label{eq:A(B_X)subAF00} \\
		  & \therefore & B_I \otimes \redspn{\lnsp{Z}^{\text{not}|}}{A_I}{F}                         & \subset \redspn{\lnsp{V}_{\alpha \beta}}{A_I}{B_I F},                                                 \\
		  & \therefore & B_I \otimes \redspn{\lnsp{Z}^{\text{not}|}}{A_I}{F}                         & \perp \redspn{\lnsp{V}_{\overline{\alpha} \beta}}{A_I}{B_I F}, \quad (\because \text{Eq.\,\eqref{B}}) \\
		\label{eq:F_XperpF0ex0}
		  & \therefore & (F^{\text{not}|}_{\alpha \beta} :=) \ \redspn{\lnsp{Z}^{\text{not}|}}{A_I}{F} & \perp \redspn{\lnsp{V}_{\overline{\alpha} \beta}}{A_I B_I}{F}.
	\end{alignat}
	By combining Eq.\,\eqref{eq:BF00=WopX}, Eq.\,\eqref{eq:F_WsubF0ex0} and Eq.\,\eqref{eq:F_XperpF0ex0}, we complete the proof.
\end{proof}

Since $\redspn{\lnsp{V}_{\overline{\alpha} \beta}}{A_I B_I}{F} \perp \redspn{\lnsp{V}_{\alpha \overline{\beta}}}{A_I B_I}{F}$ due to Eq.\,\eqref{A}, from Eq.\,\eqref{eq:F|00subF0ex0} and  Eq.\,\eqref{eq:F-00subFex00} we obtain
\begin{equation}
	\label{eq:F|00perpF0not0Succeed}
	F^|_{\alpha \beta} \perp \redspn{\lnsp{V}_{\alpha \overline{\beta}}}{A_I B_I}{F}, \quad F^-_{\alpha \beta} \perp \redspn{\lnsp{V}_{\overline{\alpha} \beta}}{A_I B_I}{F},
\end{equation}
\begin{equation}
	\label{F|00perpF-00}
	F^|_{\alpha \beta} \perp F^-_{\alpha \beta}.
\end{equation}
In Lem.\,\ref{lem:UperpV->Uop(U^perpcapV^perp)opV}, take $F^|_{\alpha \beta}$ as $\mathcal{X}$ and $F^-_{\alpha \beta}$ as $\mathcal{Y}$ and then by Eqs.\,\eqref{eq:BF00=WopXprp}, \eqref{eq:F00=F-opFex-}, \eqref{F|00perpF-00}, we obtain
\begin{equation}
	\label{eq:V00decompSucceed}
	\redspn{\lnsp{V}_{\alpha \beta}}{A_I B_I}{F} = F^|_{\alpha \beta} \oplus F^\square_{\alpha \beta} \oplus F^-_{\alpha \beta},
\end{equation}
where
\begin{equation}
	F^\square_{\alpha \beta} := F^{\text{not}|}_{\alpha \beta} \cap F^{\text{not}-}_{\alpha \beta}.
\end{equation}
From Eq.\,\eqref{eq:Fex|00subF0ex0} and Eq.\,\eqref{eq:Fex-00perpFex00}, we obtain
\begin{equation}
	\label{eq:rigoWFsqcharactertized}
	F^{\square}_{\alpha \beta} \perp \redspn{\lnsp{V}_{\overline{\alpha} \beta}}{A_I B_I}{F}, \quad
	F^{\square}_{\alpha \beta} \perp \redspn{\lnsp{V}_{\alpha \overline{\beta}}}{A_I B_I}{F}.
\end{equation}
By combining Eq.\,\eqref{eq:F|00subF0ex0}, Eq.\,\eqref{eq:F-00subFex00},
Eq.\,\eqref{eq:F|00perpF0not0Succeed}, Eq.\,\eqref{eq:V00decompSucceed} and Eq.\,\eqref{eq:rigoWFsqcharactertized},
we obtain the following decomposition of $\redspn{\lnsp{V}_{\alpha \beta}}{A_I B_I}{F}$.

\begin{prop}
	\label{prop:FabDecomposition}
	For all $\ket{\alpha}^{A_O} \in A_O$ and for all $\ket{\beta}^{B_O} \in B_O$, $\redspn{\lnsp{V}_{\alpha \beta}}{A_I B_I}{F}$ is decomposed into
	\begin{equation}
		\label{eq:Fab=F|abopFsqabopF-ab}
		\redspn{\lnsp{V}_{\alpha \beta}}{A_I B_I}{F} =
		F^{|}_{\alpha \beta} \oplus F^{\square}_{\alpha \beta} \oplus F^{-}_{\alpha \beta}
	\end{equation}
	using subspaces $F^|_{\alpha \beta}$, $F^\square_{\alpha \beta}$, $F^-_{\alpha \beta}$ of $F$ satisfying
	\begin{align}
		\label{eq:rigoWF|absubFexab}
		F^|_{\alpha \beta}                             & \subset \redspn{ \lnsp{V}_{\overline{\alpha} \beta} }{A_I B_I}{F}, \\
		\label{eq:rigoWF|abperpFexab}
		F^-_{\alpha \beta},\, F^\square_{\alpha \beta} & \perp \redspn{ \lnsp{V}_{\overline{\alpha} \beta} }{A_I B_I}{F},   \\
		F^-_{\alpha \beta}                             & \subset \redspn{ \lnsp{V}_{\alpha \overline{\beta}} }{A_I B_I}{F}, \\
		\label{eq:F|abandFsqabperpFanotb}
		F^|_{\alpha \beta},\, F^\square_{\alpha \beta} & \perp \redspn{ \lnsp{V}_{\alpha \overline{\beta}} }{A_I B_I}{F}.
	\end{align}
\end{prop}


\subsection{Decomposing $\lnsp{V}_{\alpha \beta}$}
\label{secApp:divideV}


We decompose the subspace $\lnsp{V}_{\alpha \beta}$ into $\lnsp{V}^|_{\alpha \beta} \oplus \lnsp{V}^\square_{\alpha \beta} \oplus \lnsp{V}^-_{\alpha \beta}$.
Let $\ket{\alpha}^{A_I} \in A_I$ and $\ket{\beta}^{B_I} \in B_I$. Define $\lnsp{V}^{|}_{\alpha \beta}$, $\lnsp{V}^{\square}_{\alpha \beta}$, $\lnsp{V}^{-}_{\alpha \beta}$ as
\begin{equation}
	\lnsp{V}^|_{\alpha \beta} := \projop^|_{\alpha \beta} (\lnsp{V}_{\alpha \beta}),
\end{equation}
\begin{equation}
	\lnsp{V}^\square_{\alpha \beta} := \projop^\square_{\alpha \beta} (\lnsp{V}_{\alpha \beta}),
\end{equation}
\begin{equation}
	\lnsp{V}^-_{\alpha \beta} := \projop^-_{\alpha \beta} (\lnsp{V}_{\alpha \beta}),
\end{equation}
where $\projop^|_{\alpha \beta}$, $\projop^\square_{\alpha \beta}$, $\projop^-_{\alpha \beta}$ are projections onto $A_I \otimes B_I \otimes F^|_{\alpha \beta}$, $A_I \otimes B_I \otimes F^\square_{\alpha \beta}$, $A_I \otimes B_I \otimes F^-_{\alpha \beta}$, respectively.
By definition we have
\begin{align}
	\label{eq:VF|00subF|00}
	\redspn{\lnsp{V}^|_{\alpha \beta}}{A_I B_I}{F}       & = F^|_{\alpha \beta},       \\
	\label{eq:VFsq00subFsq00}
	\redspn{\lnsp{V}^\square_{\alpha \beta}}{A_I B_I}{F} & = F^\square_{\alpha \beta}, \\
	\label{eq:VF-00subF-00}
	\redspn{\lnsp{V}^-_{\alpha \beta}}{A_I B_I}{F}       & = F^-_{\alpha \beta}.
\end{align}
From $\lnsp{V}_{\alpha \beta} \subset A_I \otimes B_I \otimes \redspn{\lnsp{V}_{\alpha \beta}}{A_I B_I}{F} =  (A_I \otimes B_I \otimes F^|_{\alpha \beta}) \oplus (A_I \otimes B_I \otimes F^\square_{\alpha \beta}) \oplus (A_I \otimes B_I \otimes F^-_{\alpha \beta})$, we have
\begin{equation}
	\label{eq:V00subV|opVsqopV-grlijgrogj}
	\lnsp{V}_{\alpha \beta} \subset \lnsp{V}^|_{\alpha \beta} \oplus \lnsp{V}^\square_{\alpha \beta} \oplus \lnsp{V}^-_{\alpha \beta}.
\end{equation}
To show that $\lnsp{V}_{\alpha \beta} \supset \lnsp{V}^|_{\alpha \beta} \oplus \lnsp{V}^\square_{\alpha \beta} \oplus \lnsp{V}^-_{\alpha \beta}$, we prepare the following proposition.
\begin{prop}
	\label{prop:g78ayg}
	\begin{align}
		\label{eq:A(BF)00subV00prop}
		\lnsp{V}_{\alpha \beta} & \supset A_I \otimes \redspn{\lnsp{V}^|_{\alpha \beta}}{A_I}{B_I F},                   \\
		\label{eq:AB(F)00subV00prop}
		\lnsp{V}_{\alpha \beta} & \supset A_I \otimes B_I \otimes \redspn{\lnsp{V}^\square_{\alpha \beta}}{A_I B_I}{F}, \\
		\label{eq:B(AF)00subV00prop}
		\lnsp{V}_{\alpha \beta} & \supset B_I \otimes \redspn{\lnsp{V}^-_{\alpha \beta}}{B_I}{A_I F}.
	\end{align}
\end{prop}
\begin{proof}
	First, we show Eq.\,\eqref{eq:A(BF)00subV00prop}.
	From Eq.\,\eqref{eq:VF|00subF|00} and $F^|_{\alpha \beta} \perp \redspn{\lnsp{V}_{A_O \overline{\beta}}}{A_I B_I}{ F}$ (due to Eq.\,\eqref{A} and Eq.\,\eqref{eq:F|abandFsqabperpFanotb} in Prop.\,\ref{prop:FabDecomposition}),
	we obtain $\redspn{\lnsp{V}^|_{\alpha \beta}}{A_I B_I}{F} \perp \redspn{\lnsp{V}_{A_O \overline{\beta}}}{A_I B_I}{F}$. Thus, we obtain
	\begin{equation}
		\label{eq:gjjkjgieb}
		\redspn{\lnsp{V}^|_{\alpha \beta}}{A_I}{B_I F} \perp \redspn{\lnsp{V}_{A_O \overline{\beta}}}{A_I}{B_I F}.
	\end{equation}
	Let $\projop^|_{\alpha \beta} \tket{\eta}^{A_I B_I F}$ be an arbitrary vector of $\lnsp{V}^|_{\alpha \beta}$ where $\tket{\eta}^{A_I B_I F} \in \lnsp{V}_{\alpha \beta}$.
	From $\lnsp{V}_{\alpha \beta} \subset A_I \otimes B_I \otimes \redspn{\lnsp{V}_{\alpha \beta}}{A_I B_I}{F}$ and Eq.\,\eqref{eq:Fab=F|abopFsqabopF-ab} in Prop.\,\ref{prop:FabDecomposition},
	$\tket{\eta}^{A_I B_I F}$ can be decomposed into
	$\tket{\eta}^{A_I B_I F} = \projop^|_{\alpha \beta} \tket{\eta}^{A_I B_I F} + \projop^\square_{\alpha \beta} \tket{\eta}^{A_I B_I F} + \projop^-_{\alpha \beta} \tket{\eta}^{A_I B_I F}$. Thus,
	$\projop^|_{\alpha \beta} \tket{\eta}^{A_I B_I F} = \tket{\eta}^{A_I B_I F} - \projop^\square_{\alpha \beta} \tket{\eta}^{A_I B_I F} - \projop^-_{\alpha \beta} \tket{\eta}^{A_I B_I F} \in \lnsp{V}_{\alpha \beta} + \lnsp{V}^\square_{\alpha \beta} + \lnsp{V}^-_{\alpha \beta}$.
	Therefore, $\lnsp{V}^|_{\alpha \beta} \subset \lnsp{V}_{\alpha \beta} + \lnsp{V}^\square_{\alpha \beta} + \lnsp{V}^-_{\alpha \beta}$ holds. We obtain
	\begin{equation}
		\label{eq:gajifiv}
		\redspn{\lnsp{V}^|_{\alpha \beta}}{A_I}{B_I F} \subset \redspn{\lnsp{V}_{\alpha \beta}}{A_I}{B_I F} + \redspn{\lnsp{V}^\square_{\alpha \beta}}{A_I}{B_I F} + \redspn{\lnsp{V}^-_{\alpha \beta}}{A_I}{B_I F}.
	\end{equation}
	From Eq.\,\eqref{B}, we have
	\begin{equation}
		\label{eq:gajifiv2}
		\redspn{\lnsp{V}_{\alpha \beta}}{A_I}{B_I F} \perp \redspn{\lnsp{V}_{\overline{\alpha} \beta}}{A_I}{B_I F}.
	\end{equation}
	Equation \eqref{eq:VFsq00subFsq00}, Eq.\,\eqref{eq:VF-00subF-00} and Eq.\,\eqref{eq:rigoWF|abperpFexab} in Prop.\,\ref{prop:FabDecomposition}
	imply that $\redspn{\lnsp{V}^\square_{\alpha \beta}}{A_I B_I}{F}, \redspn{\lnsp{V}^-_{\alpha \beta}}{A_I B_I}{F} \perp \redspn{\lnsp{V}_{\overline{\alpha} \beta}}{A_I B_I}{F}$. Thus,
	\begin{equation}
		\label{eq:gajifiv3}
		\redspn{\lnsp{V}^\square_{\alpha \beta}}{A_I}{B_I F}, \redspn{\lnsp{V}^-_{\alpha \beta}}{A_I}{B_I F} \perp \redspn{\lnsp{V}_{\overline{\alpha} \beta}}{A_I}{B_I F}.
	\end{equation}
	From Eq.\,\eqref{eq:gajifiv}, Eq.\,\eqref{eq:gajifiv2} and Eq.\,\eqref{eq:gajifiv3}, we obtain
	\begin{equation}
		\label{eq:gajifiv4}
		\redspn{\lnsp{V}^|_{\alpha \beta}}{A_I}{B_I F} \perp \redspn{\lnsp{V}_{\overline{\alpha} \beta}}{A_I}{B_I F}.
	\end{equation}
	Since $A_I \otimes (B_I \otimes F) = \lnsp{V}_{\alpha \beta} \oplus (\lnsp{V}_{\overline{\alpha} \beta} \oplus \lnsp{V}_{A_O \overline{\beta}})$,
	taking $A_I$ as $\mathcal{E}$, $B_I \otimes F$ as $\mathcal{F}$, $\lnsp{V}_{\alpha \beta}$ as $\lnsp{W}_0$, $\lnsp{V}_{\overline{\alpha} \beta} \oplus \lnsp{V}_{A_O \overline{\beta}}$ as $\lnsp{W}_1$ and $\redspn{\lnsp{V}^|_{\alpha \beta}}{A_I}{B_I F}$ as $\mathcal{X}$ in Lem.\,\ref{lem:UperpV0FeqEotUsubV1}
	and  using Eq.\,\eqref{eq:gajifiv4} and Eq.\,\eqref{eq:gjjkjgieb}, we derive Eq.\,\eqref{eq:A(BF)00subV00prop}.
	Similarly, we derive Eq.\,\eqref{eq:B(AF)00subV00prop}.

	From Eq.\,\eqref{eq:VFsq00subFsq00},
	$F^\square_{\alpha \beta} \perp \redspn{\lnsp{V}_{\overline{\alpha} \overline{\beta}}}{A_I B_I}{F}$ (due to Eq.\,\eqref{A}) and using
	Eq.\,\eqref{eq:rigoWF|abperpFexab} and Eq.\,\eqref{eq:F|abandFsqabperpFanotb} in Prop.\,\ref{prop:FabDecomposition}, we obtain
	\begin{equation}
		\label{eq:gajifiv5}
		\redspn{\lnsp{V}^\square_{\alpha \beta}}{A_I B_I}{F} \perp \redspn{\lnsp{V}_{\overline{\alpha} \overline{\beta}}}{A_I B_I}{F}, \redspn{\lnsp{V}_{\overline{\alpha} \beta}}{A_I B_I}{F}, \redspn{\lnsp{V}_{\alpha \overline{\beta}}}{A_I B_I}{F}.
	\end{equation}
	Since $(A_I \otimes B_I) \otimes F = \lnsp{V}_{\alpha \beta} \oplus (\lnsp{V}_{\overline{\alpha} \overline{\beta}} \oplus \lnsp{V}_{\overline{\alpha} \beta} \oplus \lnsp{V}_{\alpha \overline{\beta}})$,
	taking $A_I \otimes B_I$ as $\mathcal{E}$, $F$ as $\mathcal{F}$,  $\lnsp{V}_{\alpha \beta}$ as $\lnsp{W}_0$, $\lnsp{V}_{\overline{\alpha} \overline{\beta}} \oplus \lnsp{V}_{\overline{\alpha} \beta} \oplus \lnsp{V}_{\alpha \overline{\beta}}$ as $\lnsp{W}_1$
	and $\redspn{\lnsp{V}^\square_{\alpha \beta}}{A_I B_I}{F}$ as $\mathcal{X}$ in Lem.\,\ref{lem:UperpV0FeqEotUsubV1}
	and
	using Eq.\,\eqref{eq:gajifiv5}, we derive Eq.\,\eqref{eq:AB(F)00subV00prop}.
\end{proof}
From Eq.\,\eqref{eq:V00subV|opVsqopV-grlijgrogj}, we have
$\lnsp{V}_{\alpha \beta} \subset (A_I \otimes \redspn{\lnsp{V}^|_{\alpha \beta}}{A_I}{B_I F}) \oplus (A_I \otimes B_I \otimes \redspn{\lnsp{V}^\square_{\alpha \beta}}{A_I B_I}{F}) \oplus (B_I \otimes \redspn{\lnsp{V}^-_{\alpha \beta}}{B_I}{A_I F})$.
Note that from Eq.\,\eqref{eq:Fab=F|abopFsqabopF-ab} in Prop.\,\ref{prop:FabDecomposition},
$A_I \otimes \redspn{\lnsp{V}^|_{\alpha \beta}}{A_I}{B_I F}$, $A_I \otimes B_I \otimes \redspn{\lnsp{V}^\square_{\alpha \beta}}{A_I B_I}{F}$ and $B_I \otimes \redspn{\lnsp{V}^-_{\alpha \beta}}{B_I}{A_I F}$ are orthogonal.
Thus, using Prop.\,\ref{prop:g78ayg} we obtain
\begin{equation}
	\label{eq:ghbuomld}
	\lnsp{V}_{\alpha \beta} = (A_I \otimes \redspn{\lnsp{V}^|_{\alpha \beta}}{A_I}{B_I F}) \oplus (A_I \otimes B_I \otimes \redspn{\lnsp{V}^\square_{\alpha \beta}}{A_I B_I}{F}) \oplus (B_I \otimes \redspn{\lnsp{V}^-_{\alpha \beta}}{B_I}{A_I F}).
\end{equation}
From Eq.\,\eqref{eq:V00subV|opVsqopV-grlijgrogj} and Eq.\,\eqref{eq:ghbuomld}, we have
\begin{align}
	\lnsp{V}^|_{\alpha \beta}       & = A_I \otimes \redspn{\lnsp{V}^|_{\alpha \beta}}{A_I}{B_I F},                   \\
	\lnsp{V}^\square_{\alpha \beta} & = A_I \otimes B_I \otimes \redspn{\lnsp{V}^\square_{\alpha \beta}}{A_I B_I}{F}, \\
	\lnsp{V}^-_{\alpha \beta}       & = B_I \otimes \redspn{\lnsp{V}^-_{\alpha \beta}}{B_I}{A_I F},
\end{align}
and thus
\begin{equation}
	\lnsp{V}_{\alpha \beta} = \lnsp{V}^|_{\alpha \beta} \oplus \lnsp{V}^\square_{\alpha \beta} \oplus \lnsp{V}^-_{\alpha \beta}.
\end{equation}

To summarize, we obtain the following the decomposition of $\lnsp{V}_{\alpha \beta}$.
\begin{prop}
	\label{prop:DecompositionOfVab}
	For all $\ket{\alpha}^{A_O} \in A_O$ and for all $\ket{\beta}^{B_O} \in B_O$,
	$\lnsp{V}_{\alpha \beta}$ is decomposed into
	\begin{equation}
		\label{eq:Vab=V|abopVsqabopV-absgahra8h9u4}
		\lnsp{V}_{\alpha \beta} = \lnsp{V}^|_{\alpha \beta} \oplus \lnsp{V}^\square_{\alpha \beta} \oplus \lnsp{V}^-_{\alpha \beta},
	\end{equation}
	using subspaces $\lnsp{V}^|_{\alpha \beta}$, $\lnsp{V}^\square_{\alpha \beta}$, $\lnsp{V}^-_{\alpha \beta}$ of $A_I \otimes B_I \otimes F$ satisfying
	\begin{align}
		\label{eq:VF|absubF|ab}
		\redspn{\lnsp{V}^|_{\alpha \beta}}{A_I B_I}{F}       & = F^|_{\alpha \beta},       \\
		\label{eq:VFsqabsubFsqab}
		\redspn{\lnsp{V}^\square_{\alpha \beta}}{A_I B_I}{F} & = F^\square_{\alpha \beta}, \\
		\label{eq:VF-absubF-ab}
		\redspn{\lnsp{V}^-_{\alpha \beta}}{A_I B_I}{F}       & = F^-_{\alpha \beta},
	\end{align}
	and
	\begin{align}
		\lnsp{V}^|_{\alpha \beta}       & = A_I \otimes \redspn{\lnsp{V}^|_{\alpha \beta}}{A_I}{B_I F}, \label{eq:V=AAVfuahurbu} \\
		\lnsp{V}^\square_{\alpha \beta} & = A_I \otimes B_I \otimes \redspn{\lnsp{V}^\square_{\alpha \beta}}{A_I B_I}{F},        \\
		\lnsp{V}^-_{\alpha \beta}       & = B_I \otimes \redspn{\lnsp{V}^-_{\alpha \beta}}{B_I}{A_I F}.
	\end{align}
\end{prop}


\subsection{Decomposing $P$}
\label{secApp:dividePintoPABopPBAopPparallel}


From $U \left( P \otimes \vspan{\ket{\alpha}^{A_O}} \otimes \vspan{\ket{\beta}^{B_O}} \right) = \lnsp{V}_{\alpha \beta}$ and
Eq.\,\eqref{eq:Vab=V|abopVsqabopV-absgahra8h9u4} in Prop.\,\ref{prop:DecompositionOfVab}, we obtain the following decomposition of $P$.

\begin{prop} \label{prop:DecompositionOfPab}
	For all $\ket{\alpha}^{A_O} \in A_O$ and for all $\ket{\beta}^{B_O} \in B_O$, $P$ is decomposed into
	\begin{equation}
		\label{eq:Pdecomjfaijgwji3}
		P = P^|_{\alpha \beta} \oplus P^\square_{\alpha \beta} \oplus P^-_{\alpha \beta}.
	\end{equation}
	using subspaces $P^|_{\alpha \beta}$, $P^\square_{\alpha \beta}$, $P^-_{\alpha \beta}$ of $P$ satisfying
	\begin{equation}
		\label{eq:defOfP|ab}
		U \left( P^|_{\alpha \beta} \otimes \vspan{\ket{\alpha}^{A_O}} \otimes \vspan{\ket{\beta}^{B_O}} \right) = \lnsp{V}^|_{\alpha \beta},
	\end{equation}
	\begin{equation}
		U \left( P^\square_{\alpha \beta} \otimes \vspan{\ket{\alpha}^{A_O}} \otimes \vspan{\ket{\beta}^{B_O}} \right) = \lnsp{V}^\square_{\alpha \beta},
	\end{equation}
	\begin{equation}
		U \left( P^-_{\alpha \beta} \otimes \vspan{\ket{\alpha}^{A_O}} \otimes \vspan{\ket{\beta}^{B_O}} \right) = \lnsp{V}^-_{\alpha \beta},
	\end{equation}
\end{prop}

Now we have \ReviseFurther{a} critical proposition about $P^|_{\alpha \beta}$ and $P^-_{\alpha \beta}$.
\begin{prop}
	\label{prop:P|00=P|0kp}
	Let $\ket{\alpha}^{A_O} \in A_O$ and $\ket{\beta}^{B_O} \in B_O$. For all $\ket{\beta_0}^{B_O}, \ket{\beta_1}^{B_O} \in B_O$ and for all $\ket{\alpha_0}^{A_O}, \ket{\alpha_1}^{A_O} \in B_O$, the equalities
	\begin{equation}
		\label{eq:P|a0=Pab}
		P^|_{\alpha \beta_0} = P^|_{\alpha \beta_1},
	\end{equation}
	\begin{equation}
		\label{eq:P-0b=Pab}
		P^-_{\alpha_0 \beta} = P^-_{\alpha_1 \beta},
	\end{equation}
	hold.
\end{prop}
\noindent
A proof of Prop.\,\ref{prop:P|00=P|0kp} is presented in Appendix \ref{secApp:proofsInPDecomp}.

Define $\tilde{P}^|$, $\tilde{P}^\square$, $\tilde{P}^-$ as
\begin{align} \label{def:P}
	\tilde{P}^|       & := \bigplus_{\substack{
			\ket{\alpha}^{A_O} \in A_O \\
			\ket{ \beta}^{B_O} \in B_O
		}} P^|_{\alpha \beta},\\
	\tilde{P}^\square & := \bigcap_{\substack{
			\ket{\alpha}^{A_O} \in A_O \\
			\ket{ \beta}^{B_O} \in B_O
		}} P^\square_{\alpha \beta}, \label{def:P2} \\
	\tilde{P}^-       & := \bigplus_{\substack{
			\ket{\alpha}^{A_O} \in A_O \\
			\ket{ \beta}^{B_O} \in B_O
		}} P^-_{\alpha \beta}. \label{def:P3}
\end{align}

We find that $\tilde{P}^| \perp \tilde{P}^-$ since for all $\ket{\alpha}^{A_O}, \ket{\alpha^\prime}^{A_O} \in A_O$ and $\ket{\beta}^{B_O}, \ket{\beta^\prime}^{B_O} \in B_O$,
\begin{alignat}{2}
	\label{prop:5gqga09}
	P^|_{\alpha \beta}
	  & = P^|_{\alpha \beta^\prime}        &   & \quad (\because \text{Eq.\,\eqref{eq:P|a0=Pab}})                                                 \\
	  & \perp P^-_{\alpha \beta^\prime}    &   & \quad (\because \text{Eq.\,\eqref{eq:Pdecomjfaijgwji3} in Prop.\,\ref{prop:DecompositionOfPab}}) \\
	  & = P^-_{\alpha^\prime \beta^\prime} &   & \quad (\because \text{Eq.\,\eqref{eq:P-0b=Pab}}).
\end{alignat}
For all $\ket{\alpha}^{A_O} \in A_O$ and $\ket{\beta}^{B_O} \in B_O$, since $P^\square_{\alpha \beta} \perp P^|_{\alpha \beta}$
from Eq.\,\eqref{eq:Pdecomjfaijgwji3} in Prop.\,\ref{prop:DecompositionOfPab} and $\tilde{P}^\square \subset P^\square_{\alpha \beta}$ by definition,
we have $\tilde{P}^\square \perp P^|_{\alpha \beta}$.
This implies $\tilde{P}^\square \perp \tilde{P}^|$. Analogously,
$\tilde{P}^\square \perp \tilde{P}^-$ holds true. Hence, $\tilde{P}^|$, $\tilde{P}^\square$ , $\tilde{P}^-$ are mutually orthogonal.

We show $P = \tilde{P}^| + \tilde{P}^\square + \tilde{P}^-$. 	It is clear that $P \supset \tilde{P}^| + \tilde{P}^\square + P^-$. We need to show that $P \subset \tilde{P}^| + \tilde{P}^\square + P^-$.
For all $\ket{\alpha}^{A_O} \in A_O$ and for all $\ket{\beta}^{B_O} \in B_O$, we have $\tilde{P}^| + \tilde{P}^- \supset P^|_{\alpha \beta} + P^-_{\alpha \beta}$,
thus implying $(\tilde{P}^| + \tilde{P}^-)^{\perp} \subset P^\square_{\alpha \beta}$ from Eq.\,\eqref{eq:Pdecomjfaijgwji3} in Prop.\,\ref{prop:DecompositionOfPab}.
Therefore, by definition we obtain
$(\tilde{P}^| + \tilde{P}^-)^{\perp} \subset \tilde{P}^\square$.
This shows that	$P = (\tilde{P}^| + \tilde{P}^-) \oplus (\tilde{P}^| + \tilde{P}^-)^{\perp}$ is a subspace of $\tilde{P}^| + \tilde{P}^\square + \tilde{P}^-$.
Therefore, we conclude that $P = \tilde{P}^| + \tilde{P}^\square + \tilde{P}^-$.

Finally, Prop.\,\ref{prop:P|00=P|0kp} implies
$\tilde{P}^| = \bigplus_{\ket{\alpha}^{A_O} \in A_O} P^|_{\alpha \beta_0}$ and
$\tilde{P}^- = \bigplus_{\ket{ \beta}^{B_O} \in B_O} P^-_{\alpha_0 \beta}$
for all $\ket{\alpha_0}^{A_O} \in A_O$ and $\ket{\beta_0}^{B_O} \in B_O$.

The analysis above summarizes to the following decomposition of $P$.
\begin{prop}
	\label{prop:DeconpositionOfP}
	$P$ decomposes as
	\begin{equation}
		\label{eq:P=PABopPBAopPparallel}
		P = \tilde{P}^| \oplus \tilde{P}^\square \oplus \tilde{P}^-
	\end{equation}
	using subspaces $\tilde{P}^|, \tilde{P}^\square, \tilde{P}^-$ of $P$ such that for all $\ket{\alpha_0}^{A_O} \in A_O$ and for all $\ket{\beta_0}^{B_O} \in B_O$,
	\begin{align}
		\tilde{P}^|       & = \bigplus_{\substack{
				\ket{\alpha}^{A_O} \in A_O
			}} P^|_{\alpha \beta_0},\\
		\tilde{P}^\square & = \bigcap_{\substack{
				\ket{\alpha}^{A_O} \in A_O \\
				\ket{ \beta}^{B_O} \in B_O
			}} P^\square_{\alpha \beta}, \\
		\tilde{P}^-       & = \bigplus_{\substack{
				\ket{ \beta}^{B_O} \in B_O
			}} P^-_{\alpha_0 \beta}.
	\end{align}
\end{prop}


\subsection{Proof of Prop.\,\ref{prop:P|00=P|0kp}}
\label{secApp:proofsInPDecomp}


To prove Prop.\,\ref{prop:P|00=P|0kp}, we need some preparations.
\begin{prop}
	\label{prop:fgijoiaojhb}
	Let $\ket{\pi}^P,\, \ket{\pi^\prime}^P \in P$, $\ket{\alpha}^{A_O},\, \ket{\alpha^\perp}^{A_O} \in A_O$ such that $\ket{\alpha}^{A_O} \perp \ket{\alpha^\perp}^{A_O}$
	and $\dket{\Gamma}^{A_I B_I}$, $\dket{\Gamma^\prime}^{A_I B_I} \in A_I \otimes B_I$.
	For all $\ket{\beta_0}^{B_O}, \ket{\beta_1}^{B_O} \in B_O$ such that $\| \ket{\beta_0}^{B_O} \| = \| \ket{\beta_1}^{B_O} \|$, it holds that
	\begin{equation}
		\label{eq:inde8ghr9a8}
		\begin{aligned}
			  & \left(
			{}^{A_I B_I}\dbra{\Gamma} U \left( \ket{\pi}^P \ket{\alpha}^{A_O} \ket{\beta_0}^{B_O} \right), {}^{A_I B_I}\dbra{\Gamma^\prime}
			U \left( \ket{\pi^\prime}^P \ket{\alpha^\perp}^{A_O} \ket{\beta_0}^{B_O} \right)
			\right)  \\
			= & \left(
			{}^{A_I B_I}\dbra{\Gamma} U \left( \ket{\pi}^P \ket{\alpha}^{A_O} \ket{\beta_1}^{B_O} \right), {}^{A_I B_I}\dbra{\Gamma^\prime}
			U \left( \ket{\pi^\prime}^P \ket{\alpha^\perp}^{A_O} \ket{\beta_1}^{B_O} \right)
			\right) .
		\end{aligned}
	\end{equation}
\end{prop}
\begin{proof}
	Define the following sesquilinear form on $B_O$
	\begin{equation}
		f(\ket{\beta}^{B_O}, \ket{\beta^\prime}^{B_O}) := \left(
		{}^{A_I B_I}\dbra{\Gamma} U \left( \ket{\pi}^P \ket{\alpha}^{A_O} \ket{\beta}^{B_O} \right), {}^{A_I B_I}\dbra{\Gamma^\prime}
		U \left( \ket{\pi^\prime}^P \ket{\alpha^\perp}^{A_O} \ket{\beta^\prime}^{B_O} \right)
		\right). \notag
	\end{equation}
	Then Eq.\,\eqref{eq:inde8ghr9a8} leads to the relation $f(\ket{\beta_0}^{B_O}, \ket{\beta_0}^{B_O}) = f(\ket{\beta_1}^{B_O}, \ket{\beta_1}^{B_O})$.
	Equation \eqref{a}, which is equivalent to $f(\ket{\beta}^{B_O}, \ket{\beta^\perp}^{B_O}) = 0$, and Lem.\,\ref{lem:sesquiperp} yield $f(\ket{\beta_0}^{B_O}, \ket{\beta_0}^{B_O}) = f(\ket{\beta_1}^{B_O}, \ket{\beta_1}^{B_O})$.
\end{proof}
\begin{prop}
	\label{prop:7g7hq9wah}
	Let $\ket{\pi}^{P} \in P$, $\ket{\alpha}^{A_O} \in A_O$ and $\ket{\Gamma}^{A_I B_I} \in P$. For all $\ket{\beta_0}^{B_O}, \ket{\beta_1}^{B_O} \in B_O$,
	we have that
	\begin{equation}
		\label{eq:ajfdaia}
		{}^{A_I B_I}\dbra{\Gamma} U \left( \ket{\pi}^P \ket{\alpha}^{A_O} \ket{\beta_0}^{B_O} \right) \perp \redspn{\lnsp{V}_{\overline{\alpha} \beta_0}}{A_I B_I}{F}
	\end{equation}
	if and only if
	\begin{equation}
		\label{eq:ajfdaia1}
		{}^{A_I B_I}\dbra{\Gamma} U \left( \ket{\pi}^P \ket{\alpha}^{A_O} \ket{\beta_1}^{B_O} \right) \perp \redspn{\lnsp{V}_{\overline{\alpha} \beta_1}}{A_I B_I}{F}.
	\end{equation}
\end{prop}
\begin{proof}
	For all $\ket{\beta}^{B_O} \in B_O$, since
	\begin{multline}
		\redspn{\lnsp{V}_{\overline{\alpha} \beta}}{A_I B_I}{F} = \vspan{ {}^{A_I B_I}\dbra{\Gamma^\prime} U \left( \ket{\pi^\prime}^P \ket{\alpha^\perp}^{A_O} \ket{\beta}^{B_O} \right) \relmiddle| \right. \\
			\left. \ket{\pi^\prime}^P \in P,\, \ket{\alpha^\perp}^{A_O} \in A_O,\, \ket{\alpha^\perp}^{A_O} \perp \ket{\alpha}^{A_O},\, \dket{\Gamma^\prime}^{A_I B_I} \in A_I \otimes B_I },
	\end{multline}
	we have
	\begin{equation}
		\label{eq:fa8evccv0}
		{}^{A_I B_I}\dbra{\Gamma} U \left( \ket{\pi}^P \ket{\alpha}^{A_O} \ket{\beta}^{B_O} \right) \perp \redspn{\lnsp{V}_{\overline{\alpha} \beta}}{A_I B_I}{F}
	\end{equation}
	if and only if
	\begin{multline}
		\label{eq:fa8evccv1}
		\forall \ket{\pi^\prime}^P \in P,\, \ket{\alpha^\perp}^{A_O} \in A_O \text{ s.t. } \ket{\alpha}^{A_O} \perp \ket{\alpha^\perp}^{A_O},\,
		\dket{\Gamma^\prime}^{A_I B_I} \in A_I \otimes B_I, \\
		{}^{A_I B_I}\dbra{\Gamma} U \left( \ket{\pi}^P \ket{\alpha}^{A_O} \ket{\beta}^{B_O} \right)
		\perp {}^{A_I B_I}\dbra{\Gamma^\prime} U \left( \ket{\pi^\prime}^P \ket{\alpha^\perp}^{A_O} \ket{\beta}^{B_O} \right).
	\end{multline}
	Assume Eq.\,\eqref{eq:ajfdaia}.
	Since the relation presented in Eq.\,\eqref{eq:ajfdaia} is equivalent to relation presented in Eq.\,\eqref{eq:fa8evccv0} with $\ket{\beta}^{B_O} = \ket{\beta_0}^{B_O}$, we have Eq.\,\eqref{eq:fa8evccv1} with $\ket{\beta}^{B_O} = \ket{\beta_0}^{B_O}$.
	This and Prop.\,\ref{prop:fgijoiaojhb} imply Eq.\,\eqref{eq:fa8evccv1} with $\ket{\beta}^{B_O} = \ket{\beta_1}^{B_O}$.
	Thus we obtain Eq.\,\eqref{eq:fa8evccv0} with $\ket{\beta}^{B_O} = \ket{\beta_1}^{B_O}$, which is equivalent to Eq.\,\eqref{eq:ajfdaia1}.
\end{proof}
\begin{prop}
	\label{prop:P|perp00=P|perpkp}
	Let $\ket{\alpha}^{A_O} \in A_O$ and $\ket{\beta_0}^{B_O}, \ket{\beta_1}^{B_O} \in B_O$, then we have
	\begin{equation}
		\label{eq:gh6ah7a}
		P^-_{\alpha \beta_0} \oplus P^\square_{\alpha \beta_0}
		= P^-_{\alpha \beta_1} \oplus P^\square_{\alpha \beta_1}.
	\end{equation}
	(Note that $P^-_{\alpha \beta_0} = P^-_{\alpha \beta_1}$ does  not hold in general.)

	Let $\ket{\alpha_0}^{A_O}, \ket{\alpha_1}^{A_O} \in A_O$ and $\ket{\beta}^{B_O} \in B_O$, then
	\begin{align}
		\label{eq:Pfjdasigjaroiag}
		P^|_{\alpha_0 \beta} \oplus P^\square_{\alpha_0 \beta} & = P^|_{\alpha_1 \beta} \oplus P^\square_{\alpha_1 \beta}.
	\end{align}
\end{prop}
\begin{proof}
	We only show the proof of Eq.\,\eqref{eq:gh6ah7a}. The other relation follows by similar methods.
	We must show that $P^-_{\alpha \beta} \oplus P^\square_{\alpha \beta}$ is independent of $\ket{\beta}^{B_O}$.
	Define $\lnsp{V}_{\pi \alpha \beta} := \vspan{U \left( \ket{\pi}^P \ket{\alpha}^{A_O} \ket{\beta}^{B_O} \right)}$.
	Let $\ket{\pi}^P \in P$ and $\ket{\beta}^{B_O} \in B_O$.
	Then, $\ket{\pi}^P \in P^-_{\alpha \beta} \oplus P^\square_{\alpha \beta}$ is equivalent to
	\begin{equation}
		\label{eq:8agg7ga8}
		\lnsp{V}_{\pi \alpha \beta} \subset \lnsp{V}^-_{\alpha \beta} \oplus \lnsp{V}^\square_{\alpha \beta}.
	\end{equation}
	We first prove that Eq.\,\eqref{eq:8agg7ga8} is satisfied if and only if
	\begin{equation}
		\label{eq:f89ga0}
		\redspn{\lnsp{V}_{\pi \alpha \beta}}{A_I B_I}{F} \subset F^-_{\alpha \beta} \oplus F^\square_{\alpha \beta}.
	\end{equation}
	The only-if part is implied by Eq.\,\eqref{eq:VFsqabsubFsqab} and Eq.\,\eqref{eq:VF-absubF-ab} in Prop.\,\ref{prop:DecompositionOfVab}.
	If $\redspn{\lnsp{V}_{\pi \alpha \beta}}{A_I B_I}{F} \subset F^-_{\alpha \beta} \oplus F^\square_{\alpha \beta}$ is satisfied,
	then $\redspn{\lnsp{V}_{\pi \alpha \beta}}{A_I B_I}{F} \perp F^|_{\alpha \beta}$ holds.
	This implies $\lnsp{V}_{\pi \alpha \beta} \perp \lnsp{V}^|_{\alpha \beta}$, hence Eq.\,\eqref{eq:8agg7ga8}.
	Next, we prove that Eq.\,\eqref{eq:f89ga0} is satisfied if and only if
	\begin{equation}
		\label{eq:b8a7b9dbs98}
		\redspn{\lnsp{V}_{\pi \alpha \beta}}{A_I B_I}{F} \perp \redspn{\lnsp{V}_{\overline{\alpha} \beta}}{A_I B_I}{F}.
	\end{equation}
	The only-if part is implied from Eq.\,\eqref{eq:rigoWF|abperpFexab} in Prop.\,\ref{prop:FabDecomposition}.
	If $\redspn{\lnsp{V}_{\pi \alpha \beta}}{A_I B_I}{F} \perp \redspn{\lnsp{V}_{\overline{\alpha} \beta}}{A_I B_I}{F}$,
	then $\redspn{\lnsp{V}_{\pi \alpha \beta}}{A_I B_I}{F} \perp F^|_{\alpha \beta}$ due to Eq.\,\eqref{eq:rigoWF|absubFexab} in Prop.\,\ref{prop:FabDecomposition}.
	Since $\redspn{\lnsp{V}_{\pi \alpha \beta}}{A_I B_I}{F} \subset \redspn{\lnsp{V}_{\alpha \beta}}{A_I B_I}{F}$ and Eq.\,\eqref{eq:Fab=F|abopFsqabopF-ab} in Prop.\,\ref{prop:FabDecomposition} hold, the relation presented in Eq.\,\eqref{eq:f89ga0} is satisfied.

	Assume $\ket{\pi}^P \in P^-_{\alpha \beta_0} \oplus P^\square_{\alpha \beta_0}$.
	Equation \eqref{eq:b8a7b9dbs98} with $\ket{\beta}^{B_O} = \ket{\beta_0}^{B_O}$ and Prop.\,\ref{prop:7g7hq9wah} imply Eq.\,\eqref{eq:b8a7b9dbs98} with $\ket{\beta}^{B_O} = \ket{\beta_1}^{B_O}$.
	Thus we obtain $\ket{\pi}^P \in P^-_{\alpha \beta_1} \oplus P^\square_{\alpha \beta_1}$, which proves Eq.\,\eqref{eq:gh6ah7a}.
\end{proof}

\begin{proof} \ (Prop.\,\ref{prop:P|00=P|0kp})
	We now see that Prop.\,\ref{prop:P|00=P|0kp} can be proved by combining Prop.\,\ref{prop:P|perp00=P|perpkp} and Eq.\,\eqref{eq:Pdecomjfaijgwji3} in Prop.\,\ref{prop:DecompositionOfPab}.
\end{proof}


\subsection{Decomposing $F$}
\label{secApp:divideFintoFABoplFBAoplFpal}


For given subspaces $A_{\mathrm{sub}}$, $B_{\mathrm{sub}}$ of $A_O$, $B_O$, respectively, define subspaces of $A_I \otimes B_I \otimes F$ as
\begin{equation}
	\label{eq:defvtildevertab}
	\tilde{\lnsp{V}}^|(A_{\mathrm{sub}}, B_{\mathrm{sub}}) := U \left( \tilde{P}^| \otimes A_{\mathrm{sub}} \otimes B_{\mathrm{sub}} \right),
\end{equation}
\begin{equation}
	\tilde{\lnsp{V}}^\square(A_{\mathrm{sub}}, B_{\mathrm{sub}}) := U \left( \tilde{P}^\square \otimes A_{\mathrm{sub}} \otimes B_{\mathrm{sub}} \right),
\end{equation}
\begin{equation}
	\tilde{\lnsp{V}}^-(A_{\mathrm{sub}}, B_{\mathrm{sub}}) := U \left( \tilde{P}^- \otimes A_{\mathrm{sub}} \otimes B_{\mathrm{sub}} \right).
\end{equation}
From Eq.\,\eqref{eq:P=PABopPBAopPparallel} in Prop.\,\ref{prop:DeconpositionOfP}, $A_I \otimes B_I \otimes F$ decomposes to $\tilde{\lnsp{V}}^|_{A_O B_O} \oplus \tilde{\lnsp{V}}^\square_{A_O B_O} \oplus \tilde{\lnsp{V}}^-_{A_O B_O}$.
This implies that
\begin{equation}
	\label{eq:FeqFAB+FBA+Fpal}
	F = \redspn{\tilde{\lnsp{V}}^|_{A_O B_O}}{A_I B_I}{F} + \redspn{\tilde{\lnsp{V}}^\square_{A_O B_O}}{A_I B_I}{F} + \redspn{\tilde{\lnsp{V}}^-_{A_O B_O}}{A_I B_I}{F}.
\end{equation}
We now show that $\redspn{\tilde{\lnsp{V}}^|_{A_O B_O}}{A_I B_I}{F}, \redspn{\tilde{\lnsp{V}}^\square_{A_O B_O}}{A_I B_I}{F}, \redspn{\tilde{\lnsp{V}}^-_{A_O B_O}}{A_I B_I}{F}$ in Eq.\,\eqref{eq:FeqFAB+FBA+Fpal} are orthogonal to each other in Prop.\,\ref{prop:FABperpFBA} and Prop.\,\ref{prop:FeqFABoplFBAoplFpal}.
\begin{prop}
	\label{prop:FABperpFBA}
	The subspaces $\redspn{\tilde{\lnsp{V}}^|_{A_O B_O}}{A_I B_I}{F}$ and $\redspn{\tilde{\lnsp{V}}^-_{A_O B_O}}{A_I B_I}{F}$ are orthogonal.
\end{prop}
\begin{prop}
	\label{prop:FeqFABoplFBAoplFpal}
	The subspace $\redspn{\tilde{\lnsp{V}}^\square_{A_O B_O}}{A_I B_I}{F} $ is orthogonal to $\redspn{\tilde{\lnsp{V}}^|_{A_O B_O}}{A_I B_I}{F}$ and to $\redspn{\tilde{\lnsp{V}}^-_{A_O B_O}}{A_I B_I}{F}$.
\end{prop}
A proof of Prop.\,\ref{prop:FABperpFBA} is given in Appendix \ref{secApp:proofFFFFdecmpdoijfaj} and
a proof of Prop.\,\ref{prop:FeqFABoplFBAoplFpal} is given in Appendix \ref{secApp:proofofFeqFABoplFBAoplFpal}.
From Eq.\,\eqref{eq:FeqFAB+FBA+Fpal}, Prop.\,\ref{prop:FABperpFBA} and Prop.\,\ref{prop:FeqFABoplFBAoplFpal}, we obtain
\begin{equation}
	\label{eq:FeqFABoplFBAoplFpalinproof}
	F = \redspn{\tilde{\lnsp{V}}^|_{A_O B_O}}{A_I B_I}{F} \oplus \redspn{\tilde{\lnsp{V}}^\square_{A_O B_O}}{A_I B_I}{F} \oplus \redspn{\tilde{\lnsp{V}}^-_{A_O B_O}}{A_I B_I}{F},
\end{equation}
that is, these subspaces are mutually orthogonal.
We define each subspace in Eq.\,\eqref{eq:FeqFABoplFBAoplFpalinproof} as
$\tilde{F}^| := \redspn{\tilde{\lnsp{V}}^|_{A_O B_O}}{A_I B_I}{F}$,
$\tilde{F}^\square := \redspn{\tilde{\lnsp{V}}^\square_{A_O B_O}}{A_I B_I}{F}$, and
$\tilde{F}^- := \redspn{\tilde{\lnsp{V}}^-_{A_O B_O}}{A_I B_I}{F}$.
Since $A_I \otimes B_I \otimes F = \tilde{\lnsp{V}}^|_{A_O B_O} \oplus \tilde{\lnsp{V}}^\square_{A_O B_O} \oplus \tilde{\lnsp{V}}^-_{A_O B_O}$, by taking
$A_I \otimes B_I$ as $\mathcal{E}$, $F$ as $\mathcal{F}$, $\lnsp{V}^|_{A_O B_O}$ as $\lnsp{W}_0$ and $\lnsp{V}^\square_{A_O B_O} \oplus \lnsp{V}^-_{A_O B_O}$ as $\lnsp{W}_1$
in Cor.\,\ref{cor:F0perpF1=>V0=EtensorF0} and using Eq.\,\eqref{eq:FeqFABoplFBAoplFpalinproof}, we obtain $\tilde{\lnsp{V}}^|_{A_O B_O} = A_I \otimes B_I \otimes \tilde{F}^|$.
Similarly, we obtain $\tilde{\lnsp{V}}^|_{A_O B_O} = A_I \otimes B_I \otimes \tilde{F}^|$,
$\tilde{\lnsp{V}}^\square_{A_O B_O} = A_I \otimes B_I \otimes \tilde{F}^\square$, and
$\tilde{\lnsp{V}}^-_{A_O B_O} = A_I \otimes B_I \otimes \tilde{F}^-$.

To summarize, we obtain the following the decomposition of $F$.
\begin{prop}
	\label{prop:DecompositionOfF}
	$F$ is decomposed into
	\begin{equation}
		F = \tilde{F}^| \oplus \tilde{F}^\square \oplus \tilde{F}^-
	\end{equation}
	using subspaces $\tilde{F}^|, \tilde{F}^\square, \tilde{F}^-$ of $F$ satisfying
	\begin{equation}
		\label{VvTilde=AotimesBotimesFvTilde}
		\tilde{\lnsp{V}}^|_{A_O B_O} = A_I \otimes B_I \otimes \tilde{F}^|,
	\end{equation}
	\begin{equation}
		\label{VsTilde=AotimesBotimesFsTilde}
		\tilde{\lnsp{V}}^\square_{A_O B_O} = A_I \otimes B_I \otimes \tilde{F}^\square,
	\end{equation}
	\begin{equation}
		\label{VhTilde=AotimesBotimesFhTilde}
		\tilde{\lnsp{V}}^-_{A_O B_O} = A_I \otimes B_I \otimes \tilde{F}^-.
	\end{equation}
\end{prop}


\subsection{Proof of Prop.\,\ref{prop:FABperpFBA}}
\label{secApp:proofFFFFdecmpdoijfaj}


To prove Prop.\,\ref{prop:FABperpFBA}, we need some preparations.
\begin{prop}
	\label{prop:AB00perp01}
	Let $\ket{\alpha}^{A_O} \in A_O$ and $\ket{\beta}^{B_O} \in B_O$. The following orthogonality relations hold:
	\begin{equation}
		\label{eq:AB00perp01}
		\redspn{\tilde{\lnsp{V}}^|_{\alpha \beta}}{A_I B_I}{F} \perp \redspn{\lnsp{V}_{\alpha \overline{\beta}}}{A_I B_I}{F},
	\end{equation}
	\begin{equation} \label{eq:ABABAB}
		\redspn{\tilde{\lnsp{V}}^-_{\alpha \beta}}{A_I B_I}{F} \perp \redspn{\lnsp{V}_{\overline{\alpha} \beta}}{A_I B_I}{F}.
	\end{equation}
\end{prop}
\begin{proof}
	We can show Eq.\,\eqref{eq:AB00perp01} by decomposing $\tilde{P}^|$ as
	\begin{align}
		\tilde{P}^|
		  & = \bigplus_{\ket{\alpha^\prime}^{A_O} \in A_O} P^|_{\alpha^\prime \beta} \notag                                                                                        \\
		  & \subset \bigplus_{\ket{\alpha^\prime}^{A_O} \in A_O} P^|_{\alpha^\prime \beta} \oplus P^\square_{\alpha^\prime \beta} \notag                                           \\
		  & = \bigplus_{\ket{\alpha^\prime}^{A_O} \in A_O} P^|_{\alpha \beta} \oplus P^\square_{\alpha \beta} \quad (\because \mathrm{Prop.\,\ref{prop:P|perp00=P|perpkp}}) \notag \\
		  & = P^|_{\alpha \beta} \oplus P^\square_{\alpha \beta}.
	\end{align}
	Thus, considering the images of $U$ from $\tilde{P}^| \otimes \vspan{\ket{\alpha}^{A_O}} \otimes \vspan{\ket{\beta}^{B_O}}$ and from $(P^|_{\alpha \beta} \oplus P^\square_{\alpha \beta}) \otimes \vspan{\ket{\alpha}^{A_O}} \otimes \vspan{\ket{\beta}^{B_O}}$, we obtain $\tilde{\lnsp{V}}^|_{\alpha \beta} \subset \lnsp{V}^|_{\alpha \beta} \oplus \lnsp{V}^\square_{\alpha \beta}$,
	and thus
	$\redspn{\tilde{\lnsp{V}}^|_{\alpha \beta}}{A_I B_I}{F} \subset \redspn{\lnsp{V}^|_{\alpha \beta}}{A_I B_I}{F} + \redspn{\lnsp{V}^\square_{\alpha \beta}}{A_I B_I}{F}$, where the latter is equal to $F^|_{\alpha \beta} \oplus F^\square_{\alpha \beta}$ from Prop.\,\ref{prop:DecompositionOfVab}\ReviseFurther{.}
	Thus,
	using Eq.\,\eqref{eq:F|abandFsqabperpFanotb} in Prop.\,\ref{prop:FabDecomposition} we obtain $\redspn{\tilde{\lnsp{V}}^|_{\alpha \beta}}{A_I B_I}{F} \perp \redspn{\lnsp{V}_{\alpha \overline{\beta}}}{A_I B_I}{F}$.
	The proof of Eq.\,\eqref{eq:ABABAB} follows analogous steps.
\end{proof}
\begin{prop}
	\label{prop:fag4ghhas}
	Let $\ket{\pi}^P, \ket{\pi^\prime}^P \in P$, $\ket{\beta}^{B_O} \in B_O$, $\ket{\phi}^{A_I}$, $\ket{\phi^\prime}^{A_I} \in A_I$.
	For all $\ket{\alpha_0}^{A_O},\, \ket{\alpha_1}^{A_O} \in A_O$ such that ${}^{A_O}\braket{\alpha_0 | \alpha_0}^{A_O} = {}^{A_O}\braket{\alpha_1 | \alpha_1}^{A_O}$, it holds that
	\begin{equation}
		\label{eq:fadsffasg}
		\begin{aligned}
			  & \left(
			{}^{A_I}\bra{\phi} U \left( \ket{\pi}^P \ket{\alpha_0}^{A_O} \ket{\beta}^{B_O} \right), {}^{A_I}\bra{\phi^\prime}
			U \left( \ket{\pi^\prime}^P \ket{\alpha_0}^{A_O} \ket{\beta}^{B_O} \right)
			\right) \\
			= & \left(
			{}^{A_I}\bra{\phi} U \left( \ket{\pi}^P \ket{\alpha_1}^{A_O} \ket{\beta}^{B_O} \right), {}^{A_I}\bra{\phi^\prime}
			U \left( \ket{\pi^\prime}^P \ket{\alpha_1}^{A_O} \ket{\beta}^{B_O} \right)
			\right).
		\end{aligned}
	\end{equation}
\end{prop}
\begin{proof}
	Define the sesquilinear form  $f$ on $A_O$ as
	\begin{equation}
		f(\ket{\alpha}^{A_O}, \ket{\alpha^\prime}^{A_O}) := \left(
		{}^{A_I}\bra{\phi} U \left( \ket{\pi}^P \ket{\alpha}^{A_O} \ket{\beta}^{B_O} \right), {}^{A_I}\bra{\phi^\prime}
		U \left( \ket{\pi^\prime}^P \ket{\alpha^\prime}^{A_O} \ket{\beta}^{B_O} \right)
		\right) .\notag
	\end{equation}
	Then Eq.\,\eqref{c}, which is equivalent to $f(\ket{\alpha}^{A_O}, \ket{\alpha^\perp}^{A_O}) = 0$, and Lem.\,\ref{lem:sesquiperp} prove this proposition.
\end{proof}
\begin{prop}
	\label{prop:a_bf0->a_bf1}
	Let $\ket{\pi}^P, \ket{\pi^\prime}^P \in P$, $\ket{\alpha_0}^{A_O} \in A_O$, $\ket{\beta}^{B_O} \in B_O$ and $\ket{\phi}^{A_I}, \ket{\phi^\prime}^{A_I} \in A_I$.
	If there exists $\dket{\Psi_{\alpha_0}}^{B_I F} \in B_I \otimes F$ such that
	\begin{align}
		\label{eq:r38g89q}
		\left\{
		\begin{aligned}
			U \left(\ket{\pi}^P \ket{\alpha_0}^{A_O} \ket{\beta}^{B_O} \right)        & = \ket{\phi}^{A_I} \dket{\Psi_{\alpha_0}}^{B_I F},        \\
			U \left(\ket{\pi^\prime}^P \ket{\alpha_0}^{A_O} \ket{\beta}^{B_O} \right) & = \ket{\phi^\prime}^{A_I} \dket{\Psi_{\alpha_0}}^{B_I F},
		\end{aligned} \right.
	\end{align}
	then for all $\ket{\alpha_1}^{A_O} \in A_O$, there exists $\dket{\Psi_{\alpha_1}}^{B_I F} \in B_I \otimes F$ such that
	\begin{align}
		\label{eq:r38g89q2}
		\left\{
		\begin{aligned}
			U \left(\ket{\pi}^P \ket{\alpha_1}^{A_O} \ket{\beta}^{B_O} \right)        & = \ket{\phi}^{A_I} \dket{\Psi_{\alpha_1}}^{B_I F},        \\
			U \left(\ket{\pi^\prime}^P \ket{\alpha_1}^{A_O} \ket{\beta}^{B_O} \right) & = \ket{\phi^\prime}^{A_I} \dket{\Psi_{\alpha_1}}^{B_I F}.
		\end{aligned} \right.
	\end{align}
\end{prop}
\begin{proof}
	We can assume $\ket{\pi}^P, \ket{\pi^\prime}^P, \ket{\beta}^{B_O}, \ket{\phi}^{A_I}, \ket{\phi}^{A_I}, \ket{\alpha_0}^{A_O}, \ket{\alpha_1}^{A_O}$ are all normalized. Assume Eq.\,\eqref{eq:r38g89q} holds and then
	$\dket{\Psi_{\alpha_0}}^{B_I F}$ is also normalized. Thus, we obtain
	\begin{align}
		  & \left(
		{}^{A_I}\bra{\phi} U \left( \ket{\pi}^P \ket{\alpha_1}^{A_O} \ket{\beta}^{B_O} \right), {}^{A_I}\bra{\phi^\prime}
		U \left( \ket{\pi^\prime}^P \ket{\alpha_1}^{A_O} \ket{\beta}^{B_O} \right)
		\right) \notag \\
		= & \left(
		{}^{A_I}\bra{\phi} U \left( \ket{\pi}^P \ket{\alpha_0}^{A_O} \ket{\beta}^{B_O} \right), {}^{A_I}\bra{\phi^\prime}
		U \left( \ket{\pi^\prime}^P \ket{\alpha_0}^{A_O} \ket{\beta}^{B_O} \right)
		\right) \quad (\because \mathrm{Prop.\,\ref{prop:fag4ghhas}}) \notag \\
		= & \left( \dket{\Psi_{\alpha_0}}^{B_I F}, \dket{\Psi_{\alpha_0}}^{B_I F} \right) = 1
	\end{align}
	where the second equality follows from Eq.\,\eqref{eq:r38g89q}.
	Therefore, Eq.\,\eqref{eq:r38g89q2} is satisfied.
\end{proof}

\begin{proof} (Prop.\,\ref{prop:FABperpFBA})
	From bi-additivity, $\tilde{\lnsp{V}}^|_{A_O B_O}$ can be written by
	\begin{equation}
		\tilde{\lnsp{V}}^|_{A_O B_O} = \bigplus_{\substack{\ket{\alpha_0}^{A_O} \in A_O \\ \ket{\beta_0}^{B_O} \in B_O}} \tilde{\lnsp{V}}^|_{\alpha_0 \beta_0},
	\end{equation}
	which leads to
	\begin{equation}
		\redspn{\tilde{\lnsp{V}}^|_{A_O B_O}}{A_I B_I}{F} = \bigplus_{\substack{\ket{\alpha_0}^{A_O} \in A_O \\ \ket{\beta_0}^{B_O} \in B_O}}  \redspn{\tilde{\lnsp{V}}^|_{\alpha_0 \beta_0}}{A_I B_I}{F}.
	\end{equation}
	Thus it suffices to show $\redspn{\tilde{\lnsp{V}}^|_{\alpha_0 \beta_0}}{A_I B_I}{F} \perp \redspn{\tilde{\lnsp{V}}^-_{A_O B_O}}{A_I B_I}{F}$ for all $\ket{\alpha_0}^{A_O}$ and $\ket{\beta_0}^{B_O}$.
	Since $A_O = \vspan{\ket{\alpha_0}^{A_O}} \oplus \vspan{\ket{\alpha_0}^{A_O}}^\perp$ and $B_O = \vspan{\ket{\beta_0}^{B_O}} \oplus \vspan{\ket{\beta_0}^{B_O}}^\perp$, we have
	\begin{align}
		\tilde{\lnsp{V}}^-_{A_O B_O} & = \tilde{\lnsp{V}}^-_{\alpha_0 B_O} + \tilde{\lnsp{V}}^-_{\overline{\alpha_0} B_O}          \\
		                             & = \tilde{\lnsp{V}}^-_{\alpha_0 \beta_0} +\tilde{\lnsp{V}}^-_{\alpha_0 \overline{\beta_0}} +
		\tilde{\lnsp{V}}^-_{\overline{\alpha_0} \beta_0} + \tilde{\lnsp{V}}^-_{\overline{\alpha_0} \overline{\beta_0}}.
	\end{align}
	Therefore we have
	$\redspn{\tilde{\lnsp{V}}^-_{A_O B_O}}{A_I B_I}{F} = \redspn{\tilde{\lnsp{V}}^-_{\alpha_0 \beta_0}}{A_I B_I}{F} + \redspn{\tilde{\lnsp{V}}^-_{\alpha_0 \overline{\beta_0}}}{A_I B_I}{F} +
		\redspn{\tilde{\lnsp{V}}^-_{\overline{\alpha_0} \beta_0}}{A_I B_I}{F} + \redspn{\tilde{\lnsp{V}}^-_{\overline{\alpha_0} \overline{\beta_0}}}{A_I B_I}{F}$.
	Since $\redspn{\tilde{\lnsp{V}}^|_{\alpha_0 \beta_0}}{A_I B_I}{F} \subset \redspn{\lnsp{V}_{\alpha_0 \beta_0}}{A_I B_I}{F}$ and
	$\redspn{\tilde{\lnsp{V}}^-_{\overline{\alpha_0} \overline{\beta_0}}}{A_I B_I}{F} \subset \redspn{\lnsp{V}_{\overline{\alpha_0} \overline{\beta_0}}}{A_I B_I}{F}$,
	Eq.\,\eqref{A} yields
	$\redspn{\tilde{\lnsp{V}}^|_{\alpha_0 \beta_0}}{A_I B_I}{F} \perp \redspn{\tilde{\lnsp{V}}^-_{\overline{\alpha_0} \overline{\beta_0}}}{A_I B_I}{F}$
	and Prop.\,\ref{prop:AB00perp01} imply that
	$\redspn{\tilde{\lnsp{V}}^|_{\alpha_0 \beta_0}}{A_I B_I}{F} \perp \redspn{\tilde{\lnsp{V}}^-_{\alpha_0 \overline{\beta_0}}}{A_I B_I}{F}$ and
	$\redspn{\tilde{\lnsp{V}}^|_{\alpha_0 \beta_0}}{A_I B_I}{F} \perp \redspn{\tilde{\lnsp{V}}^-_{\overline{\alpha_0} \beta_0}}{A_I B_I}{F}$.

	We also need to show that $\redspn{\tilde{\lnsp{V}}^|_{\alpha_0 \beta_0}}{A_I B_I}{F} \perp \redspn{\tilde{\lnsp{V}}^-_{\alpha_0 \beta_0}}{A_I B_I}{F}$. For that we reduce the orthogonality between the reduced subspaces into a condition in terms of their vectors given by
	\begin{align}
		\redspn{\tilde{\lnsp{V}}^|_{\alpha_0 \beta_0}}{A_I B_I}{F} = & \vspan{{}^{A_I B_I}\dbra{\Gamma} U \left( \ket{\tilde{\pi}^|}^P \ket{\alpha_0}^{A_O} \ket{\beta_0}^{B_O} \right) \relmiddle| \ket{\tilde{\pi}^|}^P \in \tilde{P}^|,\, \dket{\Gamma}^{A_I B_I} \in A_I \otimes B_I} \notag \\
		\label{eq:f7aab0b33}
		=                                                            &
		\begin{multlined}
			\vspan{{}^{A_I B_I}\dbra{\Gamma} U \left( \ket{\pi^|_{\alpha_1 \beta_0}}^P \ket{\alpha_0}^{A_O} \ket{\beta_0}^{B_O} \right) \relmiddle| \right. \\ \left. \qquad \quad
			\ket{\alpha_1}^{A_O} \in A_O,\, \ket{\pi^|_{\alpha_1 \beta_0}}^P \in P^|_{\alpha_1 \beta_0},\, \dket{\Gamma}^{A_I B_I} \in A_I \otimes B_I}.
		\end{multlined}
	\end{align}
	We now show that
	\small
	\begin{equation*}
		P^|_{\alpha_1 \beta_0} = \vspan{ \ket{\pi^|_{\alpha_1 \beta_0}}^P \in P^|_{\alpha_1 \beta_0} \relmiddle| U \left( \ket{\pi^|_{\alpha_1 \beta_0}}^P \ket{\alpha_1}^{A_O} \ket{\beta_0}^{B_O} \right)
		\text{ is separable between } A_I \text{ and } B_I \otimes F }.
	\end{equation*}
	\normalsize
	The subspace $P^|_{\alpha_1 \beta_0}$ can be written as
	\begin{align}
		P^|_{\alpha_1 \beta_0}
		  & = \left\{ \ket{\pi}^P \in P \relmiddle| U \left( \ket{\pi}^P \ket{\alpha_1}^{A_O} \ket{\beta_0}^{B_O} \right) \in \lnsp{V}^|_{\alpha_1 \beta_0} \right\} \notag \\
		  & \subset \left\{ \ket{\pi}^P \in P \relmiddle| \ket{\pi}^P \ket{\alpha_1}^{A_O} \ket{\beta_0}^{B_O} =
		U^\dagger \tket{\eta^|_{\alpha_1 \beta_0}}^{A_I B_I F},\, \tket{\eta^|_{\alpha_1 \beta_0}}^{A_I B_I F}
		\in \lnsp{V}^|_{\alpha_1 \beta_0} \right\} \notag \\
		  & \subset \left\{ \left( {}^{A_O}\bra{\alpha_1} {}^{B_O}\bra{\beta_0} \right) U^\dagger \tket{\eta^|_{\alpha_1 \beta_0}}^{A_I B_I F}
		\relmiddle| \tket{\eta^|_{\alpha_1 \beta_0}}^{A_I B_I F} \in \lnsp{V}^|_{\alpha_1 \beta_0} \right\}. \label{eq:7f78ag0f0f}
	\end{align}
	From Eq.\,\eqref{eq:V=AAVfuahurbu} in Prop.\,\ref{prop:DecompositionOfVab}, we obtain
	\begin{equation}
		\lnsp{V}^|_{\alpha_1 \beta_0} = \vspan{
		\ket{\phi}^{A_I} \dket{\Psi^|_{\alpha_1 \beta_0}}^{B_I F}
		\relmiddle| \ket{\phi}^{A_I} \in A_I,\, \dket{\Psi^|_{\alpha_1 \beta_0}}^{B_I F} \in
		\redspn{\lnsp{V}^|_{\alpha_1 \beta_0}}{A_I}{B_I F}
		}.
	\end{equation}
	Thus, we derive
	\small
	\begin{multline}
		\mathrm{Eq.\,\eqref{eq:7f78ag0f0f}} =
		\vspan{ \left( {}^{A_O}\bra{\alpha_1} {}^{B_O}\bra{\beta_0} \right) U^\dagger \left( \ket{\phi}^{A_I} \dket{\Psi^|_{\alpha_1 \beta_0}}^{B_I F} \right) \relmiddle| \right. \\ \left.
		\ket{\phi}^{A_I} \in A_I,\, \dket{\Psi^|_{\alpha_1 \beta_0}}^{B_I F} \in
		\redspn{\lnsp{V}^|_{\alpha_1 \beta_0}}{A_I}{B_I F} }. \label{eq:f8a9u8g8ha9}
	\end{multline}
	\normalsize
	For all $\ket{\phi}^{A_I} \in A_I$ and $\dket{\Psi^|_{\alpha_1 \beta_0}}^{B_I F} \in
		\redspn{\lnsp{V}^|_{\alpha_1 \beta_0}}{A_I}{B_I F}$,
	from Eq.\,\eqref{eq:V=AAVfuahurbu} in Prop.\,\ref{prop:DecompositionOfVab}, we have $\ket{\phi}^{A_I} \dket{\Psi^|_{\alpha_1 \beta_0}}^{B_I F} \in \lnsp{V}^|_{\alpha_1 \beta_0} \subset \lnsp{V}_{\alpha_1 \beta_0}$,
	hence there exists $\ket{\pi}^P \in P$ satisfying $U \left( \ket{\pi}^P \ket{\alpha_1}^{A_O} \ket{\beta_0}^{B_O} \right) = \ket{\phi}^{A_I} \dket{\Psi^|_{\alpha_1 \beta_0}}^{B_I F}$. Thus,
	\small
	\begin{multline}
		\label{eq:sadgaueabbdfa}
		\text{Eq.\,\eqref{eq:f8a9u8g8ha9}} \subset
		\vspan{ \ket{\pi}^P \in P \relmiddle| U \left( \ket{\pi}^P \ket{\alpha_1}^{A_O} \ket{\beta_0}^{B_O} \right) = \ket{\phi}^{A_I}
		\dket{\Psi^|_{\alpha_1 \beta_0}}^{B_I F},\, \right. \\
		\left. \ket{\phi}^{A_I} \in A_I,\, \dket{\Psi^|_{\alpha_1 \beta_0}}^{B_I F} \in \redspn{\lnsp{V}^|_{\alpha_1 \beta_0}}{A_I}{B_I F} }.
	\end{multline}
	\normalsize
	Moreover, by Eq.\,\eqref{eq:V=AAVfuahurbu} in Prop.\,\ref{prop:DecompositionOfVab} and Eq.\,\eqref{eq:defOfP|ab} in Prop.\,\ref{prop:DecompositionOfPab},
	\small
	\begin{multline}
		\text{Eq.\,\eqref{eq:sadgaueabbdfa}} \subset
		\text{SPAN} \bigg[ \ket{\pi^|_{\alpha_1 \beta_0}}^P \in P^|_{\alpha_1 \beta_0} \bigg| U \left( \ket{\pi^|_{\alpha_1 \beta_0}}^P \ket{\alpha_1}^{A_O} \ket{\beta_0}^{B_O} \right) \\
		\text{ is separable between } A_I \text{ and } B_I \otimes F \bigg].
	\end{multline}
	\normalsize
	The inclusion relation
	\small
	\begin{multline}
		P^|_{\alpha_1 \beta_0} \supset \vspan{ \ket{\pi^|_{\alpha_1 \beta_0}}^P \in P^|_{\alpha_1 \beta_0} \relmiddle| U \left( \ket{\pi^|_{\alpha_1 \beta_0}}^P \ket{\alpha_1}^{A_O} \ket{\beta_0}^{B_O} \right) \right. \\ \left. \rule{0pt}{15pt}
		\text{ is separable between } A_I \text{ and } B_I \otimes F }
	\end{multline}
	\normalsize
	is trivial.
	We have shown that
	\begin{multline}
		P^|_{\alpha_1 \beta_0} = \vspan{ \ket{\pi^|_{\alpha_1 \beta_0}}^P \in P^|_{\alpha_1 \beta_0} \relmiddle| U \left( \ket{\pi^|_{\alpha_1 \beta_0}}^P \ket{\alpha_1}^{A_O} \ket{\beta_0}^{B_O} \right) \right. \\ \left. \rule{0pt}{15pt}
		\text{ is separable between } A_I \text{ and } B_I \otimes F }.
	\end{multline}
	Thus, we derive
	\small
	\begin{multline}
		\redspn{\tilde{\lnsp{V}}^|_{\alpha_0 \beta_0}}{A_I B_I}{F} =
		\vspan{{}^{A_I B_I}\dbra{\Gamma} U \left( \ket{\pi^|_{\alpha_1 \beta_0}}^P \ket{\alpha_0}^{A_O} \ket{\beta_0}^{B_O} \right) \relmiddle|
		\ket{\pi^|_{\alpha_1 \beta_0}}^P \in P^|_{\alpha_1 \beta_0},\, \ket{\alpha_1}^{A_O} \in A_O,\, \right. \\
		\left. \dket{\Gamma}^{A_I B_I} \in A_I \otimes B_I,\,
		U \left( \ket{\pi^|_{\alpha_1 \beta_0}}^P \ket{\alpha_1}^{A_O} \ket{\beta_0}^{B_O} \right) \text{ is separable between } A_I \text{ and } B_I \otimes F }.
	\end{multline}
	\normalsize
	Similarly, we can derive
	\small
	\begin{multline}
		\redspn{\tilde{\lnsp{V}}^-_{\alpha_0 \beta_0}}{A_I B_I}{F} =
		\vspan{{}^{A_I B_I}\dbra{\Gamma^\prime} U \left( \ket{\pi^-_{\alpha_0 \beta_1}}^P \ket{\alpha_0}^{A_O} \ket{\beta_0}^{B_O} \right) \relmiddle|
		\ket{\pi^-_{\alpha_0 \beta_1}}^P \in P^-_{\alpha_0 \beta_1},\, \ket{\beta_1}^{B_O} \in B_O,\, \right. \\
		\left. \dket{\Gamma^\prime}^{A_I B_I} \in A_I \otimes B_I,\, U \left( \ket{\pi^-_{\alpha_0 \beta_1}}^P \ket{\alpha_0}^{A_O} \ket{\beta_1}^{B_O} \right) \text{ is separable between } B_I \text{ and } A_I \otimes F }.
	\end{multline}
	\normalsize
	Thus, it suffices to show that
	\begin{equation}
		{}^{A_I B_I}\dbra{\Gamma} U \left( \ket{\pi^|_{\alpha_1 \beta_0}}^P \ket{\alpha_0}^{A_O} \ket{\beta_0}^{B_O} \right) \perp
		{}^{A_I B_I}\dbra{\Gamma^\prime} U \left( \ket{\pi^-_{\alpha_0 \beta_1}}^P \ket{\alpha_0}^{A_O} \ket{\beta_0}^{B_O} \right)
	\end{equation}
	for arbitrary $\dket{\Gamma}^{A_I B_I}, \dket{\Gamma^\prime}^{A_I B_I} \in A_I \otimes B_I$, $\ket{\alpha_1}^{A_O} \in A_O$, $\ket{\beta_1}^{B_O} \in B_O$,
	$\ket{\pi^|_{\alpha_1 \beta_0}}^P \in P^|_{\alpha_1 \beta_0}$ and
	$\ket{\pi^-_{\alpha_0 \beta_1}}^P \in P^-_{\alpha_0 \beta_1}$
	such that
	\begin{equation}
		\label{eq:7fg78ag87gdg}
		U \left( \ket{\pi^|_{\alpha_1 \beta_0}}^P \ket{\alpha_1}^{A_O} \ket{\beta_0}^{B_O} \right) = \ket{\phi_{\pi^|}}^{A_I} \dket{\Psi_{\alpha_1 }}^{B_I F},
	\end{equation}
	for some $\ket{\phi_{\pi^|}}^{A_I} \in A_I$ and $\dket{\Psi_{\alpha_1}}^{B_I F} \in B_I \otimes F$, and
	\begin{equation}
		\label{eq:g78gy7hr93jf}
		U \left( \ket{\pi^-_{\alpha_0 \beta_1}}^P \ket{\alpha_0}^{A_O} \ket{\beta_1}^{B_O} \right) = \ket{\psi_{\pi^-}}^{B_I} \dket{\Phi_{ \beta_1}}^{A_I F}
	\end{equation}
	for some $\ket{\psi_{\pi^-}}^{B_I} \in B_I$ and $\dket{\Phi_{ \beta_1}}^{A_I F} \in A_I \otimes F$.
	From Eq.\,\eqref{eq:7fg78ag87gdg} and Prop.\,\ref{prop:a_bf0->a_bf1}, we obtain
	\begin{equation}
		\label{eq:4btqtbn46a}
		U \left( \ket{\pi^|_{\alpha_1 \beta_0}}^P \ket{\alpha_0}^{A_O} \ket{\beta_0}^{B_O} \right) =
		\ket{\phi_{\pi^|}}^{A_I} \dket{\Psi_{\alpha_0}}^{B_I F}
	\end{equation}
	with $\dket{\Psi_{\beta_0}}^{B_I F} \in B_I \otimes F$ and
	\begin{equation}
		\label{eq:4btqtbn46a2}
		U \left( \ket{\pi^-_{\alpha_0 \beta_1}}^P \ket{\alpha_0}^{A_O} \ket{\beta_0}^{B_O} \right) =
		\ket{\psi_{\pi^-}}^{B_I} \dket{\Phi_{\beta_0}}^{A_I F}
	\end{equation}
	with $\dket{\Phi_{\beta_0}}^{A_I F} \in A_I \otimes F$.

	Let $\ket{\phi}^{A_I} \in A_I$ and $\ket{\psi}^{B_I} \in B_I$.
	From Eq.\,\eqref{eq:V=AAVfuahurbu} in Prop.\,\ref{prop:DecompositionOfVab} and Eq.\,\eqref{eq:7fg78ag87gdg}, there exists $\ket{\pi^|_{\alpha_1 \beta_0,\, \phi}}^P \in P^|_{\alpha_1 \beta_0}$ such that
	\begin{equation}
		\label{eq:fg89uqb5nj}
		U \left( \ket{\pi^|_{\alpha_1 \beta_0,\, \phi}}^P \ket{\alpha_1}^{A_O} \ket{\beta_0}^{B_O} \right) = \ket{\phi}^{A_I} \dket{\Psi_{\alpha_1}}^{B_I F}.
	\end{equation}
	Then, from Eq.\,\eqref{eq:7fg78ag87gdg}, Eq.\,\eqref{eq:g78gy7hr93jf} and Eq.\,\eqref{eq:fg89uqb5nj} and Prop.\,\ref{prop:a_bf0->a_bf1}, we obtain
	\begin{equation}
		\label{eq:g98a9vjj}
		U \left( \ket{\pi^|_{\alpha_1 \beta_0,\, \phi}}^P \ket{\alpha_0}^{A_O} \ket{\beta_0}^{B_O} \right) = \ket{\phi}^{A_I} \dket{\Psi_{\alpha_0}}^{B_I F}.
	\end{equation}
	Similarly, for all $\ket{\psi}^{B_O} \in B_O$, there exists $\ket{\pi^-_{\alpha_0 \beta_1,\, \psi}}^P \in P^-_{\alpha_0 \beta_1}$ such that
	\begin{equation}
		\label{eq:g98a9vjj2}
		U \left( \ket{\pi^-_{\alpha_0 \beta_1,\, \psi}}^P \ket{\alpha_0}^{A_O} \ket{\beta_0}^{B_O} \right) = \ket{\psi}^{B_I} \dket{\Phi_{\beta_0}}^{A_I F}.
	\end{equation}
	Therefore, we can transform the inner product as
	\begin{align}
		  & \left( {}^{B_I}\bradket{\psi | \Psi_{\alpha_0}}^{B_I F}, {}^{A_I}\bradket{\phi | \Phi_{\beta_0}}^{A_I F} \right)
		= \left( \ket{\phi}^{A_I} \dket{\Psi_{\alpha_0}}^{B_I F}, \ket{\psi}^{B_I} \dket{\Phi_{\beta_0}}^{A_I F} \right) \notag \\
		= & \left( U \left( \ket{\pi^|_{\alpha_1 \beta_0,\, \phi}}^P \ket{\alpha_0}^{A_O} \ket{\beta_0}^{B_O} \right), U \left( \ket{\pi^-_{\alpha_0 \beta,\, \psi}}^P \ket{\alpha_0}^{A_O} \ket{\beta_0}^{B_O} \right) \right)
		\quad (\because \text{Eq.\,\eqref{eq:g98a9vjj}, Eq.\,\eqref{eq:g98a9vjj2}}) \notag \\
		= & \left( \ket{\pi^|_{\alpha_1 \beta_0,\, \phi}}^P, \ket{\pi^-_{\alpha_0 \beta_1,\, \psi}}^P \right) = 0,
	\end{align}
	where in the last equality we have used Eq.\,\eqref{eq:P=PABopPBAopPparallel} in Prop.\,\ref{prop:DeconpositionOfP}.
	Thus, we have
	\begin{equation}
		\label{eq:bq469aubjv}
		\left( {}^{A_I B_I}\dbra{\Gamma} \left( \ket{\phi_{\pi^|}}^{A_I} \dket{\Psi_{\alpha_0}}^{B_I F} \right), {}^{A_I B_I}\dbra{\Gamma^\prime} \left( \ket{\psi_{\pi^-}}^{B_I} \dket{\Phi_{\beta_0}}^{A_I F} \right) \right) = 0.
	\end{equation}
	From Eq.\,\eqref{eq:4btqtbn46a}, Eq.\,\eqref{eq:4btqtbn46a2} and Eq.\,\eqref{eq:bq469aubjv}, we obtain
	\begin{equation}
		{}^{A_I B_I}\dbra{\Gamma} U \left( \ket{\pi^|_{\alpha_1 \beta_0}}^P \ket{\alpha_0}^{A_O} \ket{\beta_0}^{B_O} \right) \perp
		{}^{A_I B_I}\dbra{\Gamma^\prime} U \left( \ket{\pi^-_{\alpha_0 \beta_1}}^P \ket{\alpha_0}^{A_O} \ket{\beta_0}^{B_O} \right),
	\end{equation}
	which proves
	$\redspn{\tilde{\lnsp{V}}^|_{\alpha_0 \beta_0}}{A_I B_I}{F} \perp \redspn{\tilde{\lnsp{V}}^-_{\alpha_0 \beta_0}}{A_I B_I}{F}.$
\end{proof}


\subsection{Proof of Prop.\,\ref{prop:FeqFABoplFBAoplFpal}}
\label{secApp:proofofFeqFABoplFBAoplFpal}


To prove Prop.\,\ref{prop:FeqFABoplFBAoplFpal}, we need some preparations.
\begin{prop}
	\label{prop:AB00perp01sq}
	Let $\ket{\alpha}^{A_O} \in A_O$ and $\ket{\beta}^{B_O} \in B_O$, then $\redspn{\tilde{\lnsp{V}}^\square_{\alpha \beta}}{A_I B_I}{F} \perp \redspn{\lnsp{V}_{\alpha \overline{\beta}}}{A_I B_I}{F}, \redspn{\lnsp{V}_{\overline{\alpha} {\beta}}}{A_I B_I}{F}$.
\end{prop}
\begin{proof}
	Since $\tilde{P}^{\square} \subset P^\square_{\alpha \beta}$ by definition, we have $\tilde{\lnsp{V}}^\square_{\alpha \beta} \subset \lnsp{V}^\square_{\alpha \beta}$.
	Thus, we obtain $\redspn{\tilde{\lnsp{V}}^\square_{\alpha \beta}}{A_I B_I}{F} \subset \redspn{\lnsp{V}^\square_{\alpha \beta}}{A_I B_I}{F} = F^\square_{\alpha \beta}$. This and
	Eq.\,\eqref{eq:rigoWF|abperpFexab} and Eq.\,\eqref{eq:F|abandFsqabperpFanotb} in Prop.\,\ref{prop:FabDecomposition} imply Prop.\,\ref{prop:AB00perp01sq}.
\end{proof}
\begin{prop}
	\label{prop:fd78gd9sg8}
	Let $\ket{\tilde{\pi}^\square}^P \in \tilde{P}^\square$, $\ket{\pi^\prime}^P \in P$, $\dket{\Gamma}^{A_I B_I}, \dket{\Gamma^\prime}^{A_I B_I} \in A_I \otimes B_I$.
	For all $\ket{\alpha_0}^{A_O},\, \ket{\alpha_1}^{A_O} \in A_O$ and
	for all $\ket{\beta_0}^{B_O},\, \ket{\beta_1}^{B_O} \in B_O$ such that
	$\| \ket{\alpha_0}^{A_O} \| = \| \ket{\alpha_1}^{A_O} \|$ and $\| \ket{\beta_0}^{B_O} \| = \| \ket{\beta_1}^{B_O} \|$, we have
	\begin{equation}
		\label{eq:7e9d0s-w}
		\begin{aligned}
			  & \left( {}^{A_I B_I}\dbra{\Gamma} U \left( \ket{\tilde{\pi}^\square}^P \ket{\alpha_0}^{A_O} \ket{\beta_0}^{B_O} \right),
			{}^{A_I B_I}\dbra{\Gamma^\prime}  U \left( \ket{\pi^\prime}^P \ket{\alpha_0}^{A_O} \ket{\beta_0}^{B_O} \right) \right) \\
			= & \left( {}^{A_I B_I}\dbra{\Gamma} U \left( \ket{\tilde{\pi}^\square}^P \ket{\alpha_1}^{A_O} \ket{\beta_1}^{B_O} \right),
			{}^{A_I B_I}\dbra{\Gamma^\prime}  U \left( \ket{\pi^\prime}^P \ket{\alpha_1}^{A_O} \ket{\beta_1}^{B_O} \right) \right)
		\end{aligned}
	\end{equation}
\end{prop}
\begin{proof}
	Define the sesquilinear functions
	\small
	\begin{align}
		f_{\beta_0}(\ket{\alpha}^{A_O},\, \ket{\alpha^\prime}^{A_O}) & :=
		\left( {}^{A_I B_I}\dbra{\Gamma} U \left( \ket{\tilde{\pi}^\square}^P \ket{\alpha}^{A_O} \ket{\beta_0}^{B_O} \right),
		{}^{A_I B_I}\dbra{\Gamma^\prime}  U \left( \ket{\pi^\prime}^P \ket{\alpha^\prime}^{A_O} \ket{\beta_0}^{B_O} \right) \right), \notag \\
		g_{\alpha_1}(\ket{\beta}^{B_O},\, \ket{\beta^\prime}^{B_O})  & :=
		\left( {}^{A_I B_I}\dbra{\Gamma} U \left( \ket{\tilde{\pi}^\square}^P \ket{\alpha_1}^{A_O} \ket{\beta}^{B_O} \right),
		{}^{A_I B_I}\dbra{\Gamma^\prime}  U \left( \ket{\pi^\prime}^P \ket{\alpha_1}^{A_O} \ket{\beta^\prime}^{B_O} \right) \right). \notag
	\end{align}
	\normalsize
	Then Eq.\,\eqref{eq:7e9d0s-w} is written as $f_{\beta_0}(\ket{\alpha_0}^{A_O}, \ket{\alpha_0}^{A_O}) = g_{\alpha_1}(\ket{\beta_1}^{B_O}, \ket{\beta_1}^{B_O})$.
	Prop.\,\ref{prop:AB00perp01sq} implies $f_{\beta_0}(\ket{\alpha}^{A_O}, \ket{\alpha^\perp}^{A_O}) = 0$ and $g_{\alpha_1}(\ket{\beta}^{B_O}, \ket{\beta^\perp}^{B_O}) = 0$
	where $\ket{\alpha}^{A_O} \perp \ket{\alpha^\perp}^{A_O}$ and $\ket{\beta}^{B_O} \perp \ket{\beta^\perp}^{B_O}$.
	From these relations and Lem.\,\ref{lem:sesquiperp}, we obtain $f_{\beta_0}(\ket{\alpha_0}^{A_O}, \ket{\alpha_0}^{A_O}) = f_{\beta_0}(\ket{\alpha_1}^{A_O}, \ket{\alpha_1}^{A_O}) =
		g_{\alpha_1}(\ket{\beta_0}^{B_O}, \ket{\beta_0}^{B_O}) = g_{\alpha_1}(\ket{\beta_1}^{B_O}, \ket{\beta_1}^{B_O})$ where the second equality is from direct substitution.
\end{proof}
\begin{prop}
	\label{prop:fa79agg}
	Let $\ket{\tilde{\pi}^{\square}}^P \in \tilde{P}^{\square}$, $\dket{\Gamma}^{A_I B_I} \in A_I \otimes B_I$.
	For all $\ket{\alpha_0}^{A_O},\, \ket{\alpha_1}^{A_O} \in A_O$ and
	for all $\ket{\beta_0}^{B_O},\, \ket{\beta_1}^{B_O} \in B_O$ such that
	$\| \ket{\alpha_0}^{A_O} \| = \| \ket{\alpha_1}^{A_O} \|$ and $\| \ket{\beta_0}^{B_O} \| = \| \ket{\beta_1}^{B_O} \|$, it holds that
	\begin{equation}
		\left\| {}^{A_I B_I}\dbra{\Gamma} U \left( \ket{\tilde{\pi}^{\square}}^P \ket{\alpha_0}^{A_O} \ket{\beta_0}^{B_O} \right) \right\| =
		\left\| {}^{A_I B_I}\dbra{\Gamma} U \left( \ket{\tilde{\pi}^{\square}}^P \ket{\alpha_1}^{A_O} \ket{\beta_1}^{B_O} \right) \right\|.
	\end{equation}
\end{prop}
\begin{proof}
	Prop.\,\ref{prop:fd78gd9sg8} with $\ket{\pi^\prime}^P = \ket{\tilde{\pi}^{\square}}^P$ and $\dket{\Gamma^\prime}^{A_I B_I} = \dket{\Gamma}^{A_I B_I}$ implies this proposition.
\end{proof}

\begin{prop}
	\label{prop:g78awyjivb3}
	Let $\ket{\tilde{\pi}^\square}^{P} \in \tilde{P}^\square$ and $\dket{\Gamma}^{A_I B_I},\, \dket{\Gamma^\prime}^{A_I B_I} \in A_I \otimes B_I$. For all $\ket{\alpha_0}^{A_O} \in A_O$ and $\ket{\beta_0}^{B_O} \in B_O$, we have
	\begin{equation}
		\label{eq:resusol}
		\dket{\Gamma^\prime}^{A_I B_I} \dbra{\Gamma} U \left( \ket{\tilde{\pi}^\square}^P \ket{\alpha_0}^{A_O} \ket{\beta_0}^{B_O} \right) \in \tilde{\lnsp{V}}^\square_{\alpha_0 \beta_0}.
	\end{equation}
\end{prop}

\begin{proof}
	Without loss of generality, assume $\ket{\tilde{\pi}^\square}^{P}$, $\dket{\Gamma}^{A_I B_I}$ and $\dket{\Gamma^\prime}^{A_I B_I}$ are normalized.
	Define the non-negative real number \ReviseFurther{$\lambda \geq 0$} and the normalized vector $\ket{{f}^\square_{\alpha_0 \beta_0}}^F \in F^\square_{\alpha_0 \beta_0}$ as
	\begin{equation}
		\label{eq:fa8dfg89a9ad}
		{}^{A_I B_I}\dbra{\Gamma} U \left( \ket{\tilde{\pi}^\square}^P \ket{\alpha_0}^{A_O} \ket{\beta_0}^{B_O} \right) =: \sqrt{\lambda} \ket{{f}^\square_{\alpha_0 \beta_0}}^F.
	\end{equation}
	When $\lambda = 0$, Eq.\,\eqref{eq:resusol} is satisfied. We assume $\lambda \neq 0$ in the following.
	It suffices to show $	\dket{\Gamma^\prime}^{A_I B_I} \ket{{f}^\square_{\alpha_0 \beta_0}}^F \in \tilde{\lnsp{V}}^\square_{\alpha_0 \beta_0}$.

	From $\dket{\Gamma^\prime}^{A_I B_I} \ket{{f}^\square_{\alpha_0 \beta_0}}^F \in A_I \otimes B_I \otimes F^\square_{\alpha_0 \beta_0} = \lnsp{V}^\square_{\alpha_0 \beta_0}$, there exists $\ket{{\pi^{\prime}}^{\square}_{\alpha_0 \beta_0}}^P \in P^\square_{\alpha_0 \beta_0}$ such that
	\begin{equation}
		\label{eq:f7d9a8fuugi}
		U \left( \ket{{\pi^{\prime}}^{\square}_{\alpha_0 \beta_0}}^P \ket{\alpha_0}^{A_O} \ket{\beta_0}^{B_O} \right) = \dket{\Gamma^\prime}^{A_I B_I} \ket{{f}^\square_{\alpha_0 \beta_0}}^F.
	\end{equation}
	We show that $\ket{{\pi^{\prime}}^{\square}_{\alpha_0 \beta_0}}^P \in P^\square_{\alpha_1 \beta_1}$ for all $\ket{\alpha_1}^{A_O} \in A_O$ and for all $\ket{\beta_1}^{B_O} \in B_O$.
	Without loss of generality, assume $\ket{\alpha_1}^{A_O}$ and $\ket{\beta_1}^{B_O}$ are normalized.
	Take $\ket{\tilde{\pi}^\square}^P \in \tilde{P}^\square $, Eq.\,\eqref{eq:fa8dfg89a9ad} and Prop.\,\ref{prop:fa79agg} imply
	$\left\| {}^{A_I B_I}\dbra{\Gamma} U \left( \ket{\tilde{\pi}^\square}^P \ket{\alpha_1}^{A_O} \ket{\beta_1}^{B_O} \right) \right\| = \sqrt{\lambda}$
	and thus there exists a normalized vector $\ket{{f}^\square_{\alpha_1 \beta_1}}^F \in F^\square_{\alpha_1 \beta_1}$ such that
	\begin{equation}
		\label{eq:8fyg9dlb}
		{}^{A_I B_I}\dbra{\Gamma} U \left( \ket{\tilde{\pi}^\square}^P \ket{\alpha_1}^{A_O} \ket{\beta_1}^{B_O} \right) = \sqrt{\lambda} \ket{{f}^\square_{\alpha_1 \beta_1}}^F.
	\end{equation}
	Therefore, we can calculate the inner product as
	\begin{align}
		  & \left( \dket{\Gamma^\prime}^{A_I B_I} \ket{{f}^\square_{\alpha_1 \beta_1}}^F, U \left( \ket{{\pi^{\prime}}^{\square}_{\alpha_0 \beta_0}}^P \ket{\alpha_1}^{A_O} \ket{\beta_1}^{B_O} \right) \right) \notag   \\
		= & \left( \ket{{f}^\square_{\alpha_1 \beta_1}}^F, {}^{A_I B_I}\dbra{\Gamma^\prime} U \left( \ket{{\pi^{\prime}}^{\square}_{\alpha_0 \beta_0}}^P \ket{\alpha_1}^{A_O} \ket{\beta_1}^{B_O} \right) \right) \notag \\
		= & \sqrt{\lambda^{-1}} \left( {}^{A_I B_I}\dbra{\Gamma} U \left( \ket{\tilde{\pi}^\square}^P \ket{\alpha_1}^{A_O} \ket{\beta_1}^{B_O} \right),
		{}^{A_I B_I}\dbra{\Gamma^\prime} U \left( \ket{{\pi^{\prime}}^{\square}_{\alpha_0 \beta_0}}^P \ket{\alpha_1}^{A_O} \ket{\beta_1}^{B_O} \right) \right) \notag \\
		= & \sqrt{\lambda^{-1}} \left( {}^{A_I B_I}\dbra{\Gamma} U \left( \ket{\tilde{\pi}^\square}^P \ket{\alpha_0}^{A_O} \ket{\beta_0}^{B_O} \right),
		{}^{A_I B_I}\dbra{\Gamma^\prime} U \left( \ket{{\pi^{\prime}}^{\square}_{\alpha_0 \beta_0}}^P
		\ket{\alpha_0}^{A_O} \ket{\beta_0}^{B_O} \right) \right) \quad (\because \mathrm{Prop.\,\ref{prop:fd78gd9sg8}}) \notag \\
		= & \sqrt{\lambda^{-1}} \left( \sqrt{\lambda} \ket{{f}^\square_{\alpha_0 \beta_0}}^F,\,
		\ket{{f}^\square_{\alpha_0 \beta_0}}^F \right) = 1.
	\end{align}
	Therefore, $U \left( \ket{{\pi^{\prime}}^{\square}_{\alpha_0 \beta_0}}^P \ket{\alpha_1}^{A_O} \ket{\beta_1}^{B_O} \right) = \dket{\Gamma^\prime}^{A_I B_I} \ket{{f}^\square_{\alpha_1 \beta_1}}^F$.
	Since $\dket{\Gamma^\prime}^{A_I B_I} \ket{{f}^\square_{\alpha_1 \beta_1}}^F \in A_I \otimes B_I \otimes F^\square_{\alpha_1 \beta_1} = \lnsp{V}^\square_{\alpha_1 \beta_1}$,
	it must be that
	$\ket{{\pi^{\prime}}^{\square}_{\alpha_0 \beta_0}}^P \in P^\square_{\alpha_1 \beta_1}$.
	Thus, we have shown $\ket{{\pi^{\prime}}^{\square}_{\alpha_0 \beta_0}}^P \in {P}^\square_{\alpha_1 \beta_1}$
	for all $\ket{\alpha_1}^{A_O} \in A_O$ and for all $\ket{\beta_1}^{B_O} \in B_O$, which implies $\ket{{\pi^{\prime}}^{\square}_{\alpha_0 \beta_0}}^P \in \tilde{P}^\square$. From this and Eq.\,\eqref{eq:f7d9a8fuugi}, we obtain
	$\dket{\Gamma^\prime}^{A_I B_I} \ket{{f}^\square_{\alpha_0 \beta_0}}^F \in \tilde{\lnsp{V}}^\square_{\alpha_0 \beta_0}$.
\end{proof}

Finally we show the proof of Prop.\,\ref{prop:FeqFABoplFBAoplFpal}.

\begin{proof} (Prop.\,\ref{prop:FeqFABoplFBAoplFpal})
	If $A_I \otimes B_I \otimes \redspn{\tilde{\lnsp{V}}^\square_{A_O B_O}}{A_I B_I}{F} = \tilde{\lnsp{V}}^\square_{A_O B_O}$,
	then we can show $\redspn{\tilde{\lnsp{V}}^\square_{A_O B_O}}{A_I B_I}{F} \perp \redspn{\tilde{\lnsp{V}}^|_{A_O B_O}}{A_I B_I}{F}$ and $\redspn{\tilde{\lnsp{V}}^\square_{A_O B_O}}{A_I B_I}{F} \perp  \redspn{\tilde{\lnsp{V}}^-_{A_O B_O}}{A_I B_I}{F}$ by using Cor.\,\ref{cor:F0perpF1=>V0=EtensorF0}
	with $\mathcal{E} = A_I \otimes B_I$, $\mathcal{F} = F$, $\lnsp{W}_0 = \tilde{\lnsp{V}}^\square_{A_O B_O}$ and
	$\lnsp{W}_1 = \tilde{\lnsp{V}}^|_{A_O B_O} \oplus \tilde{\lnsp{V}}^-_{A_O B_O}$.
	$A_I \otimes B_I \otimes \redspn{\tilde{\lnsp{V}}^\square_{A_O B_O}}{A_I B_I}{F} \supset \tilde{\lnsp{V}}^\square_{A_O B_O}$ is trivial. We show $A_I \otimes B_I \otimes \redspn{\tilde{\lnsp{V}}^\square_{A_O B_O}}{A_I B_I}{F} \subset \tilde{\lnsp{V}}^\square_{A_O B_O}$. Since we can transform
	\begin{align}
		\redspn{\tilde{\lnsp{V}}^\square_{A_O B_O}}{A_I B_I}{F}
		  & = \vspan{{}^{A_I B_I}\dbratket{\Gamma | \tilde{\eta}^\square}^{A_I B_I F} \relmiddle|  \tket{\tilde{\eta}^\square}^{A_I B_I F} \in \tilde{\lnsp{V}}^\square_{A_O B_O},\, \dket{\Gamma}^{A_I B_I} \in A_I \otimes B_I} \notag \\
		  & = \begin{multlined}
			\vspan{{}^{A_I B_I}\dbra{\Gamma} U \left( \ket{\tilde{\pi}^\square}^P \ket{\alpha}^{A_O} \ket{\beta}^{B_O} \right) \right. \\
				\left. \relmiddle|
				\ket{\tilde{\pi}^\square}^P \in \tilde{P}^\square,\, \ket{\alpha}^{A_O} \in A_O,\, \ket{\beta}^{B_O} \in B_O,\, \dket{\Gamma}^{A_I B_I} \in A_I \otimes B_I },
		\end{multlined} \notag
	\end{align}
	we have
	\begin{multline}
		A_I \otimes B_I \otimes \redspn{\tilde{\lnsp{V}}^\square_{A_O B_O}}{A_I B_I}{F} =
		\vspan{\dket{\Gamma^\prime}^{A_I B_I} \dbra{\Gamma} U \left( \ket{\tilde{\pi}^\square}^P \ket{\alpha}^{A_O} \ket{\beta}^{B_O} \right) \right. \\
			\left. \relmiddle|
			\ket{\tilde{\pi}^\square}^P \in \tilde{P}^\square,\, \ket{\alpha}^{A_O} \in A_O,\, \ket{\beta}^{B_O} \in B_O,\,
			\dket{\Gamma}^{A_I B_I},\,\dket{\Gamma^\prime}^{A_I B_I} \in A_I \otimes B_I}.
	\end{multline}
	From this and Prop.\,\ref{prop:g78awyjivb3}, we obtain
	$A_I \otimes B_I \otimes \redspn{\tilde{\lnsp{V}}^\square_{A_O B_O}}{A_I B_I}{F} \subset \tilde{\lnsp{V}}^\square_{A_O B_O}$.
\end{proof}


\subsection{Causality}
\label{secApp:causality}


Define the restriction $\tilde{U}^|$, $\tilde{U}^\square$, $\tilde{U}^-$ of $U$ as:
\begin{align}
	\tilde{U}^|       & := U|_{\tilde{P}^| \otimes A_O \otimes B_O},  \label{eq:defUtildevert} \\
	\tilde{U}^\square & := U|_{\tilde{P}^\square \otimes A_O \otimes B_O},                     \\
	\tilde{U}^-       & := U|_{\tilde{P}^- \otimes A_O \otimes B_O}.
\end{align}
Due to Eq.\,\eqref{VvTilde=AotimesBotimesFvTilde} - Eq.\,\eqref{VhTilde=AotimesBotimesFhTilde} in Prop.\,\ref{prop:DecompositionOfF}, $\Ima U^| = A_I \otimes B_I \otimes \tilde{F}^|$, $\Ima U^\square = A_I \otimes B_I \otimes \tilde{F}^\square$ and $\Ima U^- = A_I \otimes B_I \otimes \tilde{F}^-$ and we can define unitary operators
\begin{align}
	\tilde{U}^|       & \colon \tilde{P}^| \otimes A_O \otimes B_O \to A_I \otimes B_I \otimes \tilde{F}^|,             \\
	\tilde{U}^\square & \colon \tilde{P}^\square \otimes A_O \otimes B_O \to A_I \otimes B_I \otimes \tilde{F}^\square, \\
	\tilde{U}^-       & \colon \tilde{P}^- \otimes A_O \otimes B_O \to A_I \otimes B_I \otimes \tilde{F}^-.
\end{align}
Equation \eqref{eq:P=PABopPBAopPparallel} in Prop.\,\ref{prop:DeconpositionOfP} implies that $P \otimes A_O \otimes B_O = (\tilde{P}^| \otimes A_O \otimes B_O) \oplus (\tilde{P}^\square \otimes A_O \otimes B_O) \oplus(\tilde{P}^- \otimes A_O \otimes B_O)$ and $U$ can be decomposed as
$U = \tilde{U}^| \oplus \tilde{U}^\square \oplus \tilde{U}^-$.

\begin{prop}
	\label{prop:UtildevertisAprecB}
	$\tilde{U}^|$ and $\tilde{U}^-$ represent quantum combs of $A \prec B$ and $B \prec A$, respectively.
\end{prop}
\begin{proof}
	\Add{
	By definitions Eq.\,\eqref{eq:defvtildevertab} and Eq.\,\eqref{eq:defUtildevert} of $\tilde{\lnsp{V}}^|_{A_{\mathrm{sub}}, B_{\mathrm{sub}}}$ and $\tilde{U}^|$,
	\begin{equation}
		\label{eq:VtildverteqUtildevertPtildevertab}
		\tilde{\lnsp{V}}^|_{A_{\mathrm{sub}}, B_{\mathrm{sub}}} = \tilde{U}^| \left( \tilde{P}^| \otimes A_{\mathrm{sub}} \otimes B_{\mathrm{sub}} \right).
	\end{equation}
	Let $\ket{\alpha}^{A_O} \in A_O$ and $\ket{\beta}^{B_O} \in B_O$. Prop.\,\ref{prop:AB00perp01} implies $\redspn{\tilde{\lnsp{V}}^|_{\alpha \beta}}{A_I B_I}{F} \perp \redspn{\tilde{\lnsp{V}}^|_{\alpha \overline{\beta}}}{A_I B_I}{F}$ and Eq.\,\eqref{A} implies $\redspn{\tilde{\lnsp{V}}^|_{\alpha \beta}}{A_I B_I}{F} \perp \redspn{\tilde{\lnsp{V}}^|_{\overline{\alpha} \overline{\beta}}}{A_I B_I}{F}$. Thus, considering $A_O = \alpha \oplus \overline{\alpha}$,
	\begin{equation}
		\label{eq:tildeborthforth}
		\redspn{\tilde{\lnsp{V}}^|_{A_O \beta}}{A_I B_I}{F} \perp \redspn{\tilde{\lnsp{V}}^|_{A_O \overline{\beta}}}{A_I B_I}{F}.
	\end{equation}
	Eq.\,\eqref{B} implies
	\begin{equation}
		\label{eq:tildeaorthbforth}
		\redspn{\tilde{\lnsp{V}}^|_{\alpha B_O}}{A_I}{B_I F} \perp \redspn{\tilde{\lnsp{V}}^|_{\overline{\alpha} B_O}}{A_I}{B_I F}.
	\end{equation}
	From Cor.\,\ref{cor:N=2purecomb}, Eq.\,\eqref{eq:VtildverteqUtildevertPtildevertab}, Eq.\,\eqref{eq:tildeborthforth}, and Eq.\,\eqref{eq:tildeaorthbforth} means that $\tilde{U}^|$ represents a quantum comb of $A \prec B$.
	}
	Similarly, we can show that $\tilde{U}^-$ represents a quantum comb of $B \prec A$.
\end{proof}
\begin{prop}
	$\tilde{U}^\square$ represents a quantum comb of $A \parallel B$.
\end{prop}
\begin{proof}
	\Add{
		Similarly as the proof of Prop.\,\ref{prop:UtildevertisAprecB}, we can obtain by definitions of $\tilde{\lnsp{V}}^\square(A_{\mathrm{sub}}, B_{\mathrm{sub}})$ and $\tilde{U}^\square$,
		\begin{equation}
			\tilde{\lnsp{V}}^\square(A_{\mathrm{sub}}, B_{\mathrm{sub}}) = \tilde{U}^\square \left( \tilde{P}^\square \otimes A_{\mathrm{sub}} \otimes B_{\mathrm{sub}} \right)
		\end{equation}
		and from Prop.\,\ref{prop:AB00perp01sq}
		\begin{equation}
			\redspn{\tilde{\lnsp{V}}^\square_{A_O \beta}}{A_I B_I}{F} \perp \redspn{\tilde{\lnsp{V}}^\square_{A_O \overline{\beta}}}{A_I B_I}{F}.
		\end{equation}
		Eq.\,\eqref{B} implies
		\begin{equation}
			\redspn{\tilde{\lnsp{V}}^\square_{\alpha B_O}}{A_I}{B_I F} \perp \redspn{\tilde{\lnsp{V}}^\square_{\overline{\alpha} B_O}}{A_I}{B_I F}.
		\end{equation}
		Therefore, $\tilde{U}^\square$ represents a quantum comb of $A \prec B$.
		Similarly, we can show that $\tilde{U}^\square$ represents a quantum comb of $B \prec A$.
	}
\end{proof}

Now, we define
$P^{A \prec B} := \tilde{P}^| \oplus \tilde{P}^\square$, $P^{B \prec A} := \tilde{P}^-$,
$F^{A \prec B} := \tilde{F}^| \oplus \tilde{F}^\square$, $F^{B \prec A} := \tilde{F}^-$,
$U^{A \prec B} := \tilde{U}^| \oplus \tilde{U}^\square$, $U^{B \prec A} := \tilde{U}^-$. Then we obtain the following proposition.
\begin{prop}
	\label{prop:finaleqs}
	\begin{equation}
		\label{eq:finaleqsDecompositionOfP}
		P = P^{A \prec B} \oplus P^{B \prec A},
	\end{equation}
	\begin{equation}
		F = F^{A \prec B} \oplus F^{B \prec A},
	\end{equation}
	\begin{equation}
		U = U^{A \prec B} \oplus U^{B \prec A},
	\end{equation}
	\begin{equation}
		U^{A \prec B} \colon P^{A \prec B} \otimes A_O \otimes B_O \to A_I \otimes B_I \otimes F^{A \prec B},
	\end{equation}
	\begin{equation}
		U^{B \prec A} \colon P^{B \prec A} \otimes A_O \otimes B_O \to A_I \otimes B_I \otimes F^{B \prec A},
	\end{equation}
	\begin{equation}
		U^{A \prec B} \text{ represents a quantum comb of } A \prec B,
	\end{equation}
	\begin{equation}
		\label{eq:finaleqsUBtoAisQComb}
		U^{B \prec A} \text{ represents a quantum comb of } B \prec A.
	\end{equation}
\end{prop}
\noindent Therefore, $U$ represents a direct sum of pure combs, that is, the condition (3) in Thm.\,\ref{thm:equivalent},
which completes the proof of (2) $\implies$ (3) in Thm.\,\ref{thm:equivalent}. \qed


\end{document}